\documentclass{article}
\usepackage[left=2cm, right=2cm, top=2cm]{geometry}
\usepackage{bussproofs}
\usepackage{authblk}
\usepackage{amssymb}
\usepackage{amsmath}
\usepackage{amsthm}
\usepackage{mathtools}
\usepackage{comment}
\usepackage[inline]{enumitem}
\usepackage[colorlinks=true,linkcolor=blue,urlcolor=gray]{hyperref}
\usepackage{xy}
\xyoption{all}
\usepackage{tikz}
\usepackage{tikz-cd}
\newtheorem{definition}{Definition}
\newtheorem{example}[definition]{Example}
\newtheorem{theorem}[definition]{Theorem}

\newtheorem{lemma}[definition]{Lemma}

\newtheorem{remark}[definition]{Remark}

\usepackage{calc}
\usepackage{longtable}

\usepackage{array,etoolbox}
\makeatletter

\newlength{\PS@lastparam}
\newlength{\PSlastparam}
\newcommand{\PSlp}{%
  \setlength{\PSlastparam}{\PS@lastparam}%
  \the\PSlastparam
}
\def\PS@sub@lastparam{}

\newcommand{\PS@numwidth}{99}
\newcommand{\PSnumwidth}[1]{%
  \renewcommand{\PS@numwidth}{#1}%
}

\newcommand{\PS@style}{\small}
\newcommand{\PS@numstyle}{\footnotesize}

\newlength{\PSindent}
\newlength{\PS@extraindent}
\newlength{\PSpre}
\newlength{\PSpost}
\newlength{\PS@Nwidth}
\newlength{\PS@Swidth}
\newlength{\PS@Ewidth}
\newlength{\PScolsep}
\newcommand{\PS@rownumber}{%
  \ifPS@subsubsteps
  \thePSsubstepc.%
  \the\numexpr\value{PSsubsubstepc}+1\relax
  \else
  \ifPS@substeps
  \thePSstepc.%
  \the\numexpr\value{PSsubstepc}+1\relax
  \else
  \the\numexpr\value{PSstepc}+1\relax
  \fi\fi
}
\newcommand{\PS@step}{%
  \ifPS@subsubsteps
  \refstepcounter{PSsubsubstepc}%
  \else
  \ifPS@substeps
  \refstepcounter{PSsubstepc}%
  \else
  \refstepcounter{PSstepc}
  \fi\fi%
}

\newif\ifPS@inprogress
\newif\ifPS@substeps
\newif\ifPS@subsubsteps
\newif\ifPS@continued
\newif\ifPS@subcontinued
\newcounter{PSc}
\newcounter{PSstepc}[PSc]
\newcounter{PSsubstepc}[PSstepc]
\renewcommand{\thePSsubstepc}{\thePSstepc.\arabic{PSsubstepc}}
\newcounter{PSsubsubstepc}[PSsubstepc]

\newenvironment{proofsteps}[1]{%
  \global\settowidth{\PS@lastparam}{\PS@style\hspace*{#1}}
  \ifPS@continued\else\refstepcounter{PSc}\fi
  \begingroup
  \setlength{\LTpre}{\PSpre}%
  \setlength{\LTpost}{\PSpost}%
  
  \setlength{\tabcolsep}{0pt}
  \noindent\PS@style
  \settowidth{\PS@Nwidth}{\PS@numstyle\PS@numwidth}%
  \setlength{\PS@Swidth}{#1}%
  \addtolength{\PS@Swidth}{-\PS@extraindent}%
  \setlength{\PS@Ewidth}{\linewidth}%
  \addtolength{\PS@Ewidth}{-\PSindent}%
  \addtolength{\PS@Ewidth}{-\PS@extraindent}%
  \addtolength{\PS@Ewidth}{-\PS@Nwidth}%
  \addtolength{\PS@Ewidth}{-\PScolsep}%
  \addtolength{\PS@Ewidth}{-\PS@Swidth}%
  \addtolength{\PS@Ewidth}{-\PScolsep}%
  \PS@inprogresstrue
  \longtable{%
    @{\hspace*{\PSindent}\hspace*{\PS@extraindent}\makebox[\PS@Nwidth][r]{\PS@rownumber}}%
    @{\hskip\PScolsep}>{\PS@step}p{\PS@Swidth}%
    @{\hskip\PScolsep}>{\footnotesize\raggedright\arraybackslash}p{\PS@Ewidth}%
  }%
}{%
  \ifPS@inprogress
  \addtocounter{table}{-1}%
  \endlongtable  
  \endgroup
  \PS@continuedfalse
  \PS@inprogressfalse
  \else\fi
}

\newcommand{\PSbreak}[1]{%
  \endproofsteps
  \par\medskip
  #1
  \medskip\par
  \PS@continuedtrue
  \proofsteps{\PS@lastparam}%
}

\newif\ifPS@sub@inprogress

\newif\ifPS@laststep
\newcommand{\laststep}{\global\PS@laststeptrue}
\newif\ifPS@lastsubstep
\newcommand{\lastsubstep}{\global\PS@lastsubsteptrue}

\newcommand{\adjustcol}[1]{%
  \global\advance\@colroom-#1%
}

\makeatother

\newcommand*{\pcformat}[1]{%
  [\;{\normalfont\itshape #1}\;]%
}

\newenvironment{proofcases}[1][]{%
  \description[font=\pcformat, leftmargin=\parindent, #1]%
}{\enddescription}

\newcommand{\Space}{~~~~~}
\newcommand{\DINS}{{\mathsf{DINS}}} 
\newcommand{\FOL}{{\mathsf{FOL}}}
\newcommand{\CRing}{{\mathsf{CRing}}} 
\newcommand{\INS}{{\mathcal{I}}} 
\newcommand{\INM}{{\mathcal{L}}} 
\newcommand{\A}{{\mathfrak{A}}}
\newcommand{\B}{{\mathfrak{B}}}

\newcommand{\Hom}{\mathrm{Hom}} 
\newcommand{\id}{\mathrm{id}} 
\newcommand{\Cod}{\mathrm{Cod}} 
\newcommand{\Dom}{\mathrm{Dom}} 
\newcommand{\cod}{\mathrm{cod}} 
\newcommand{\dom}{\mathrm{dom}} 
\newcommand{\Sig}{\mathsf{Sig}}
\newcommand{\Mod}{\mathsf{Mod}}
\newcommand{\Dex}{\mathcal{D}} 
\newcommand{\Sen}{\mathsf{Sen}}
\newcommand{\card}{\mathsf{card}}
\newcommand{\power}{\mathsf{power}}
\newcommand{\pos}[1]{{\langle}#1{\rangle}}
\newcommand{\T}{{\mathsf{T}}}
\newcommand{\gene}{{\mathsf{gene}}}
\newcommand{\BS}[0]{\boldsymbol}
\newcommand{\Cat}{\mathtt{Cat}}
\newcommand{\CAT}{\mathtt{CAT}} 
\newcommand{\Set}{\mathtt{Set}}

\newcommand{\red}{\!\upharpoonright\!}
\newcommand\oldcoprod{\coprod} %
\newcommand\newcoprod{\textstyle\coprod\nolimits} %
\newcommand\oldprod{\prod} %
\newcommand\newprod{\textstyle\prod\nolimits} %

\newcommand{\Exists}[1]{\exists #1\,{\cdot}\,}
\newcommand{\Init}{\mathsf{Init}}
\newcommand{\Cut}{\mathsf{Cut}}
\newcommand{\Modify}{\mathsf{Modify}}
\newcommand{\Atom}{\mathsf{Atom}}
\newcommand{\OL}[1]{#1}
\newcommand{\UL}[1]{#1}
\newcommand{\proofrule}[2]{\displaystyle\frac{#1}{#2}}
\newcommand{\proofrulet}[2]{\genfrac{}{}{0pt}{0}{#1}{#2}}

\usepackage{tcolorbox} 
\usepackage{tablefootnote}
\title{
A formulation of D-institution using functor categories
}
\author{Go Hashimoto} 
\date{2026}
\begin{document}
\maketitle
\begin{abstract}
Variables are a crucial element in logic and are also addressed in institution theory, an effort to axiomatize logic. In institution theory, we typically use extensions (signature morphisms) obtained from variables instead of introducing variables directly. While this approach appears simple at first glance because it does not introduce new structures, it often requires numerous conditions to describe variable structures, which can actually complicate the discussion. In this paper, we propose introducing variable structures directly by utilizing a generalization of category of functors. We define a category of "predicate logics" and formulate the introduction of compound sentences as a functor. We also introduce a proof system and prove a completeness theorem.
\end{abstract}
\section{Introduction}
\subparagraph*{Institution}
$\INS=\pos{\Sig^{\INS},\Mod^{\INS},\Sen^{\INS},\models^{\INS}}$ is an institution if
\begin{itemize}
\item $\Sig^{\INS}\in|\CAT|$ is a locally small category
\item $\Mod^{\INS}:(\Sig^{\INS})^{op}\to\CAT$ is a functor
\item $\Sen^{\INS}:\Sig^{\INS}\to\Set$ is a functor
\item satisfaction condition holds true i.e. each $\chi:\Sigma\to\Sigma'\in\Sig^{\INS}$, $\gamma\in\Sen^{\INS}(\Sigma)$, $\A'\in|\Mod^{\INS}(\Sigma')|$ satisfy the following
\[\Mod^{\INS}(\chi)(\A')\models^{\INS}\gamma\iff\A'\models^{\INS}\Sen^{\INS}(\chi)(\gamma)\]
\end{itemize}
Through the study of institutions, theorems can be generalized, and results for newly defined institutions can be automatically obtained simply by verifying that the preconditions are met.
Furthermore, it becomes possible to define categories of institutions, enabling comparisons between institutions, importing results through translation, and deriving new institutions.
\subparagraph*{Variables}
Let $\Sigma=\pos{S,F,P}$ be a first-order signature.
Let $\Sigma[x]=\pos{S,F\cup\{x\},P}$ and let $\iota:\Sigma\to\Sigma[x]$ be the inclusion.
The set $\Sen(\Sigma[x])$ contains sentences that include variables, and by applying quantification to $\gamma\in\Sen(\Sigma[x])$, a new sentence $\Exists{x}\gamma\in\Sen(\Sigma)$ can be defined.
In reality, we need other tools to translate sentences and obtain satisfaction conditions for added sentences.
In typical institution theory, these are not directly introduced as structures. \cite{Diaconescu2025}
The extension $\iota:\Sigma\to\Sigma[X]$ is used in place of the variable $X$,
and we write a subclass of $\Sig$ consisting of such extensions as $\Dex$.
While this may seem simple at first glance, and many results have been obtained based on this method,
it has some drawbacks, such as the following:
\begin{itemize}
\item
$\Dex\subseteq\Sig$ does not have category-theoretic information such as "it is a subcategory," so we need to add conditions, which results in a haphazard approach.
\item
Since quantified sentences are treated internally (i.e., as sentences that the institution originally possesses),
handling quantified sentences requires manually introducing them (and proving that what is created is an institution).
\item
The concept of $\Dex$-institution morphisms and the categories of $\Dex$-institutions have not been defined.
\end{itemize}
This paper proposes introduction of a variable structure using ( generalization of ) functor categories to solve these problems.
\begin{itemize}
\item By using a generalized functor category (similar to the Cartesian product $\prod_{a\in A}B(a)$ when comparing it to a set of functions $A^{B}$), we can capture the signature-dependent behavior of $\Sigma$-variables.
\item The definition of $\Dex$-institution morphism is naturally derived, and the category of $\Dex$-institutions ($\DINS$) can be defined.
\item Furthermore, this allows us to view the introduction of compound sentences as a functor $\FOL:\DINS\to\DINS$.
\end{itemize}
\section{Preliminary and notation}\label{sec:pre}
To define a direct sum of sets, we first label the elements so that they do not overlap and then take the union.
This construction can be extended to categories and is called the category of elements, or the Grothendieck construction for indexed categories.
This operation can be viewed as a functor.
\begin{definition}[Category of elements]\label{def:elements}
Let $S\in|\Cat|$ be a category.
We define the functor $\newcoprod_{S}\in|\Cat^{\Cat^{S^{op}}}|$
\footnote{Larger categories cannot be treated as mathematical objects, so this is not actually a category, but since there is no problem if we do not perform set-like operations, expressions such as $F:D\to C\in\Cat^{\Cat^{S^{op}}}$ and $C,D\in|\Cat^{\Cat^{S^{op}}}|$ are also used as shorthand.} as follows. 
\begin{itemize}
\item Let $M\in|\Cat^{S^{op}}|$ be a functor. $\newcoprod_{S}M$ is defined as follows:
\begin{itemize}
\item $\pos{\A,\Sigma}\in|\newcoprod_{S}M|$ iff $\Sigma\in|S|$ and $\A\in|M(\Sigma)|$
\item $\pos{h,\chi}:\pos{\A,\Sigma}\to\pos{\A',\Sigma'}\in\newcoprod_{S}M$ iff $\chi:\Sigma\to\Sigma'\in S$ and $h:\A\to M(\chi)(\A')\in M(\Sigma)$ 
\item $\id_{\pos{\A,\Sigma}}:=\pos{\id_{\A},\id_{\Sigma}}:\pos{\A,\Sigma}\to\pos{\A,\Sigma}$
\item $\pos{h',\chi'}\circ\pos{h,\chi}:=\pos{M(\chi)(h')\circ h,\chi'\circ_{S}\chi}:\pos{\A,\Sigma}\to\pos{\A'',\Sigma''}$, where $\pos{h,\chi}:\pos{\A,\Sigma}\to\pos{\A',\Sigma'}$ and $\pos{h',\chi'}:\pos{\A',\Sigma'}\to\pos{\A'',\Sigma''}$
\end{itemize}
\item Let $\eta:M\to M'\in\Cat^{S^{op}}$ be a natural transformation. $\newcoprod_{S}\eta:\newcoprod_{S}M\to\newcoprod_{S}M'$ is defined as follows:
\begin{itemize}
\item $\newcoprod_{S}\eta(\pos{\A,\Sigma}):=\pos{\eta_{\Sigma}(\A),\Sigma}$
\item $\newcoprod_{S}\eta(\pos{h,\chi}:\pos{\A,\Sigma}\to\pos{\A',\Sigma'}):=\pos{\eta_{\Sigma}(h),\chi}:\pos{\eta_{\Sigma}(\A),\Sigma}\to\pos{\eta_{\Sigma'}(\A'),\Sigma'}$
\end{itemize}
\end{itemize}
\end{definition}
\begin{definition}\label{def:elements-nat}
Let $L:S\to S'\in\Cat$ be a functor.
We define the natural transformation $\newcoprod_{L}:\newcoprod_{S}\circ\,\Cat^{L^{op}}\to\newcoprod_{S'}\in\Cat^{\Cat^{S'^{op}}}$ as follows.
\begin{itemize}
\item Let $M\in|\Cat^{S'^{op}}|$. $\newcoprod_{L}M:\newcoprod_{S}(M\circ L^{op})\to\newcoprod_{S'}M\in\Cat$ is defined as follows:
\begin{itemize}
\item $(\newcoprod_{L}M)(\pos{\A,\Sigma}):=\pos{\A,L(\Sigma)}\in|\newcoprod_{S'}M|$ where $\pos{\A,\Sigma}\in|\newcoprod_{S}(M\circ L^{op})|$
\item $(\newcoprod_{L}M)(\pos{h,\chi}:\pos{\A,\Sigma}\to\pos{\A',\Sigma'}):=\pos{h,L(\chi)}:\pos{\A,L(\Sigma)}\to\pos{\A',L(\Sigma')}\in\newcoprod_{S'}M$ where $\pos{h,\chi}:\pos{\A,\Sigma}\to\pos{\A',\Sigma'}\in\newcoprod_{S}(M\circ L^{op})$
\end{itemize}
\end{itemize}
\end{definition}
\begin{lemma}\label{lemma:coprod_is_a_list}
The following hold true.
\begin{itemize}
\item $\newcoprod_{\id_{S}}=\id_{\oldcoprod_{S}}$, where $S\in|\Cat|$.
\item $\newcoprod_{L\circ^{op}L'}=\Cat^{\Cat^{L'^{op}}}(\newcoprod_{L})\circ_{V}^{op}\newcoprod_{L'}$, where $L:S\to S',\,L':S'\to S''\in\Cat$.
\end{itemize}
\end{lemma}
The operation of constructing a Cartesian product of sets can also be extended to categories.
Intuitively, it is a generalization of the category of functors.
Here we will call this category the category of lists.
This operation can also be extended to a functor.
\begin{definition}[Category of lists]\label{def:lists}
Let $S\in|\Cat|$ be a category.
We define the functor $\newprod_{S}\in|\Cat^{\Cat^{S^{op}}}|$ as follows.
\begin{itemize}
\item Let $M:S^{op}\to\Cat$ be a functor.
Let $\pi:\newcoprod_{S}M\ni\pos{h,\chi}\mapsto\chi\in S$ be the forgetful functor.
We definine $\newprod_{S}M$ as a subcategory of $(\newcoprod_{S}M)^{S}$ as follows:
\begin{itemize}
\item $c\in|\newprod_{S}M|$ iff
$c\in|(\newcoprod_{S}M)^{S}|$ is functor such that $\pi\circ c=\id_{S}$. (i.e. $c$ is a section.)
\item $h:c\to c'\in\newprod_{S}M$ iff
$h:c\Rightarrow c'\in(\newcoprod_{S}M)^{S}$ is a natural transformation such that horizontal composition satisfies $\id_{\pi}\circ_{H}h=\id_{\id_{S}}(\id_{\pi}:\pi\Rightarrow\pi, \id_{\id_{S}}:\id_{S}\Rightarrow\id_{S})$.
\end{itemize}
Since the output can be difficult to handle if it contains extra labels, we will denote the first components by $(\cdot)^{c}$ and $(\cdot)^{h}$, i.e. $c(\cdot)=\pos{(\cdot)^{c},\cdot}, h(\cdot)=\pos{(\cdot)^{h},\cdot}$.
\footnote{$c$ and $c(\cdot)$ are synonymous. This notation is mainly used when we want to combine multiple functors. The co-domain of functors and natural transformations, and the range of possible values of $\cdot$, depend on the context.}
\item Let $\eta:M\to M'\in\Cat^{S^{op}}$ be a natural transformation. We definine $\newprod_{S}\eta:\newprod_{S}M\to\newprod_{S}M'$ as follows:
\begin{itemize}
\item $\newprod_{S}\eta(c):=\newcoprod_{S}\eta\circ c$
\item $\newprod_{S}\eta(h:c\to c'):=\newcoprod_{S}\eta\circ_{H} h:\newcoprod_{S}\eta\circ c\to\newcoprod_{S}\eta\circ c'$\footnote{When it is necessary to distinguish between them, horizontal composition is written as $\circ_{H}$ and vertical composition as $\circ_{V}$.}
\end{itemize}
\end{itemize}
\end{definition}
\begin{definition}\label{def:lists-nat}
Let $L:S\to S'\in\Cat$ be a functor.
We define the natural transformation $\newprod_{L}:\newprod_{S'}\to\newprod_{S}\circ\,\Cat^{L^{op}}\in\Cat^{(\Cat^{S'^{op}})}$ as follows.
\begin{itemize}
\item Let $M\in|\Cat^{S'^{op}}|$. $\newprod_{L}M:\newprod_{S'}M\to\newprod_{S}(M\circ L^{op})\in\Cat$ is defined as follows:
\begin{itemize}
\item $(\newprod_{L}M)(c):=\{\pos{L(\chi)^{c},\chi}\}_{\chi\in S}\in|\newprod_{S}(M\circ L^{op})|$ where $c\in|\newprod_{S'}M|$
\item $(\newprod_{L}M)(h:c\to c'):=\{\pos{L(\Sigma)^{h},\id_{\Sigma}}\}_{\Sigma\in|S|}:\{\pos{L(\chi)^{c},\chi}\}_{\chi\in S}\to\{\pos{L(\chi)^{c'},\chi}\}_{\chi\in S}\in\newprod_{S}(M\circ L^{op})$ where $h:c\to c'\in\newprod_{S'}M$
\end{itemize}
\end{itemize}
\end{definition}
\begin{lemma}\label{lemma:prod_is_a_list}
The following hold true.
\begin{itemize}
\item $\newprod_{\id_{S}}=\id_{\oldprod_{S}}$, where $S\in|\Cat|$.
\item $\newprod_{L\circ^{op}L'}=\Cat^{\Cat^{L'^{op}}}(\newprod_{L})\circ_{V}\newprod_{L'}$, where $L:S\to S',\,L':S'\to S''\in\Cat$.
\end{itemize}
\end{lemma}
\begin{remark}
$\Cat^{\Cat^{(\cdot)}}:((\Cat)^{op})^{op}\to\CAT$ is a functor.
$(\cdot)^{op}:\Cat\to\Cat\,\text{ or }(\,:\CAT\to\CAT)$ is a functor which assigns $F^{op}:D^{op}\to C^{op}$ to $F:D\to C$.
Let $\Cat^{\Cat^{(\cdot)^{op}}}:=\Cat^{\Cat^{(\cdot)}}\circ(\cdot)^{op}_{\Cat}:((\Cat)^{op})^{op}\to\CAT$, and $(\Cat^{\Cat^{(\cdot)^{op}}})^{op}:=(\cdot)^{op}_{\CAT}\circ\,\Cat^{\Cat^{(\cdot)}}\circ(\cdot)^{op}_{\Cat}:((\Cat)^{op})^{op}\to\CAT$.
We define labeled $\newcoprod,\,\newprod$ as follows:
\begin{itemize}
\item $\newcoprod'_{L:S'\to S\in\Cat^{op}}:=\pos{\newcoprod_{L}:\newcoprod_{S'}\to(\Cat^{\Cat^{L^{op}}})^{op}(\newcoprod_{S})\,\in(\Cat^{\Cat^{S'^{op}}})^{op},L:S'\to S}:\pos{\newcoprod_{S'},S'}\to\pos{\newcoprod_{S},S}$
\item $\newprod'_{L:S'\to S\in\Cat^{op}}:=\pos{\newprod_{L}:\newprod_{S'}\to\Cat^{\Cat^{L^{op}}}(\newprod_{S}),L:S'\to S}:\pos{\newprod_{S'},S'}\to\pos{\newprod_{S},S}$
\end{itemize}
Lemma~\ref{lemma:coprod_is_a_list} and Lemma~\ref{lemma:prod_is_a_list} mean the following, respectively.
\begin{itemize}
\item $\newcoprod':\Cat^{op}\to\newcoprod_{\Cat^{op}}(\Cat^{\Cat^{(\cdot)^{op}}})^{op}$ is a functor.
Since the right component is $\id$, $\newcoprod'\in|\newprod_{\Cat^{op}}(\Cat^{\Cat^{(\cdot)^{op}}})^{op}|$.
\item $\newprod':\Cat^{op}\to\newcoprod_{\Cat^{op}}\Cat^{\Cat^{(\cdot)^{op}}}$ is a functor.
Since the right component is $\id$, $\newprod'\in|\newprod_{\Cat^{op}}\Cat^{\Cat^{(\cdot)^{op}}}|$.
\end{itemize}
(Here, we extend $\newcoprod_{S},\,\newprod_{S}$ from $S\in\Cat$ to $S\in\CAT$.)
In other words, $\newcoprod$ and $\newprod$ themselves are also large lists in a sense.
\end{remark}

The above constructions can be performed as concrete operations even if we use $\CAT$, $S\in|\CAT|$ and $L:S\to S'\in\CAT$ instead of $\Cat$, $S\in|\Cat|$ and $L:S\to S'\in\CAT$, so we use the same notation.

\section{$\Dex$-institutions and quantifications}\label{sec:Dq}
\begin{definition}[$\Dex$-institutions]\label{def:dins}
$\INS=\pos{\Sig^{\INS},\Mod^{\INS},\Sen^{\INS},\models^{\INS};\Dex^{\INS}}$ is $\Dex$-institution iff
\begin{itemize}
\item $\Sig^{\INS}\in|\CAT|$ is a locally small category
\item $\Mod^{\INS}:(\Sig^{\INS})^{op}\to\CAT$ is a functor
\item $\Sen^{\INS}:\Sig^{\INS}\to\Set$ is a functor
\item $\Mod^{\INS}(\chi)(\A')\models^{\INS}\gamma$ iff $\A'\models^{\INS}\Sen^{\INS}(\chi)(\gamma)$
\item $\Dex^{\INS}=\pos{(\cdot)_{\Dex^{\INS}}\in|\Cat^{\Sig^{\INS}}|,\,\Dex^{\INS}(\cdot)\in|\newprod_{\cdot\in\Sig^{\INS}}(\cdot/\Sig^{\INS})^{(\cdot)_{\Dex^{\INS}}}|}$
\\
\small
where $\chi/\Sig^{\INS}:\Sigma'/\Sig^{\INS}\to\Sigma/\Sig^{\INS}\,(\chi:\Sigma\to\Sigma\in\Sig^{\INS})$ is a functor which assign $\chi''':\chi'\circ\chi\to\chi''\circ\chi\in\Sigma/\Sig^{\INS}$ to $\chi''':\chi'\to\chi''\in\Sigma'/\Sig^{\INS}$.
This gives us a functor $\cdot/\Sig^{\INS}:(\Sig^{\INS})^{op}\to\CAT$.
Composing the two functors with the $\Hom$ functor $(\chi/\Sig)^{\chi_{\Dex}}:=\Hom_{\CAT}(\chi_{\Dex},(\chi/\Sig))$ gives us a contravariant functor $(\cdot/\Sig)^{(\cdot)_{\Dex}}:\Sig^{op}\to\CAT$, so we can define $\newprod_{\cdot\in\Sig}(\cdot/\Sig)^{(\cdot)_{\Dex}}$.
\normalsize
\item For all $\chi:\Sigma\to\Sigma'\in\Sig^{\INS}$, and $X\in|\Sigma_{\Dex^{\INS}}|$, the pair $\pos{\chi^{\Dex^{\INS}}(X):\Sigma^{\Dex^{\INS}}\to(\chi/\Sig^{\INS})(\Sigma'^{\Dex^{\INS}}(\chi_{\Dex^{\INS}}(X))),\Sigma'^{\Dex^{\INS}}(\chi_{\Dex^{\INS}}(X))}$
is an universal morphism from $\Sigma^{\Dex^{\INS}}$ to $\chi/\Sig^{\INS}:\Sigma'/\Sig^{\INS}\to\Sigma/\Sig^{\INS}$.
\item For all $\chi:\Sigma\to\Sigma'\in\Sig^{\INS}$, $\A'\in\Mod^{\INS}(\Sigma')$, $X\in|\Sigma_{\Dex^{\INS}}|$ and $\Sigma^{\Dex^{\INS}}(X)$-expansion $\hat{\A}$ of $\A:=\Mod^{\INS}(\chi)(\A')$,
there is unique $\Sigma'^{\Dex^{\INS}}(\chi_{\Dex^{\INS}}(X))$-expansion $\hat{\A}'$ of $\A'$
such that $\Mod^{\INS}(\Sigma^{\Dex^{\INS}})(\hat{\A}')=\hat{\A}$.
\begin{center}
$\xymatrix@R=40pt@C=40pt{
\hat{\A} \ar@{.}[r] \ar@{.>}[d]|{\red} & \Mod^{\INS}(\Sigma^{\Dex^{\INS}}[X])
\ar@{<-}[rrr]|{\Mod^{\INS}(\chi^{\Dex^{\INS}}[X])} \ar[d]|{\Mod^{\INS}(\Sigma^{\Dex^{\INS}}(X))} &&& \Mod^{\INS}(\Sigma'^{\Dex^{\INS}}[\chi_{\Dex^{\INS}}(X)]) \ar[d]|{\Mod^{\INS}(\Sigma'^{\Dex^{\INS}}(\chi_{\Dex^{\INS}}(X)))} & \ar@{.}[l] \hat{\A}' \ar@{.>}[d]|{\red} \\
\A \ar@{.}[r] & \Mod^{\INS}(\Sigma) \ar@{<-}[rrr]|{\Mod^{\INS}(\chi)} &&& \Mod^{\INS}(\Sigma') & \ar@{.}[l] \A'
}$
\footnote{A coslice category has a projection $\Cod:c/C\ni(\chi:\chi'(:\Sigma\to\Sigma')\to\chi''(:\Sigma\to\Sigma''))\mapsto(\chi:\Sigma'\to\Sigma'')\in C$.
From now on, we use the notations $F[\cdot]:=\Cod\circ F$ and $\eta[\cdot]=\Cod\circ_{H}\eta$ for functors $F\in|(c/C)^{S}|$ and natural transformations $\eta\in(c/C)^{S}$.
Note that $\Sigma^{\Dex^{\INS}}(X)$ and $\Sigma^{\Dex^{\INS}}[X]$ are different but $\chi^{\Dex^{\INS}}(X)$ and $\chi^{\Dex^{\INS}}[X]$ have the same value.}
\end{center}
\end{itemize}
\end{definition}
For clarity, we omit the subscript $\INS$ when it is clear from the context.
\begin{enumerate}
\item
The part before $\Dex$ is the definition of institution. It assumes only the minimum translation conditions for the model and sentence.
\item
Intuitively, $(\cdot)_{\Dex}:\Sig\to\Cat$ gives, for each signature, a set of variable sets that may be used to extend (or quantify) the signature.
In this case, there is an inclusion relation between the variable sets, so $\Sigma_{\Dex}$ can be considered a category by order.
The objects in the category are sometimes called blocks.
\item
Intuitively, $\Dex(\cdot)=\pos{(\cdot)^{\Dex},\cdot}\in|\newprod_{\cdot\in\Sig}(\cdot/\Sig)^{(\cdot)_{\Dex}}|)$ gives the actual extension $\Sigma^{\Dex}(X):\Sigma\to\Sigma^{\Dex}[X]$ for each set of variables $X\in|\Sigma_{\Dex}|$.
$(\cdot)^{\Dex}$ is a list of such extension rules that is category-theoretically consistent.
This can be rephrased using Definition~\ref{def:lists} as follows:
\begin{itemize}
\item $\pos{(\cdot)^{\Dex},\cdot}\in|\newprod_{\cdot\in\Sig}(\cdot/\Sig)^{(\cdot)_{\Dex}}|$ iff
\begin{itemize}
\item $\Sigma^{\Dex}\in|(\Sigma/\Sig)^{\Sigma_{\Dex}}|$ where $(\Sigma\in|\Sig|)$ and
\item $\chi^{\Dex}:\Sigma^{\Dex}\to(\chi/\Sig)^{\chi_{\Dex}}(\Sigma'^{\Dex})\in(\Sigma/\Sig)^{\Sigma_{\Dex}}$ where $(\chi:\Sigma\to\Sigma'\in\Sig)$ such that
\end{itemize}
the correspondence $\Dex:\Sig\ni(\chi:\Sigma\to\Sigma')\mapsto\pos{\chi^{\Dex},\chi}:\pos{\Sigma^{\Dex},\Sigma}\to\pos{\Sigma'^{\Dex},\Sigma'}\in\newcoprod_{\cdot\in\Sig}(\cdot/\Sig)^{(\cdot)_{\Dex}}$ is a functor.
\end{itemize}
We check in order.
\begin{proofcases}
\item[$\Sigma^{\Dex}\in|(\Sigma/\Sig)^{\Sigma_{\Dex}}|$] Since the destination is a co-slice category, it becomes:
\begin{center}
$\xymatrix{
\Sigma_{\Dex} & \ni & (\iota:X\to Y ) & \mapsto & \Sigma^{\Dex}[X] \ar[rr]|{\Sigma^{\Dex}[\iota]} & & \Sigma^{\Dex}[Y] & \in & \Sigma/\Sig \\
 & & & & & \Sigma \ar[ul]|{\Sigma^{\Dex}(X)} \ar[ur]|{\Sigma^{\Dex}(Y)} & &
}$
\end{center}
The condition is that $\Sigma^{\Dex}:\Sigma_{\Dex}\to(\Sigma/\Sig)$ is a functor, but since the co-domain is a co-slice category, it can also be considered a condition for morphisms on $\Sig$.
\begin{equation}\label{eq:IDex-1}
\begin{split}
&\Sigma^{\Dex}[\id_{X}]=\id_{\Sigma^{\Dex}[X]} \text{ for all } \Sigma\in|\Sig| \text{ and } X\in|\Sigma_{\Dex}| \\
&\Sigma^{\Dex}[\iota'\circ\iota]=\Sigma^{\Dex}[\iota']\circ\Sigma^{\Dex}[\iota] \text{ for all } \Sigma\in|\Sig| \text{ and } \iota:X\to Y,\iota':Y\to Z\in\Sigma_{\Dex}
\end{split}
\end{equation}
\item[$\chi^{\Dex}:\Sigma^{\Dex}\to(\chi/\Sig)^{\chi_{\Dex}}(\Sigma'^{\Dex})\in(\Sigma/\Sig)^{\Sigma_{\Dex}}$]
The target is $(\chi/\Sig)^{\chi_{\Dex}}(\Sigma'^{\Dex})(X)=\Sigma'^{\Dex}(\chi_{\Dex}(X))\circ\chi$ for each element $X\in\Sigma_{\Dex}$.
Since $\chi^{\Dex}$ is a morphism of functor categories, it is a natural transformation. Therefore, for each element $X\in\Sigma_{\Dex}$, $\chi^{\Dex}(X):\Sigma^{\Dex}(X)\to\Sigma'^{\Dex}(\chi_{\Dex}(X))\circ\chi\in\Sigma/\Sig$.
\begin{center}
$\xymatrix@C=80pt{
\Sigma^{\Dex}[X] \ar[r]|{\chi^{\Dex}[X]} & \Sigma'^{\Dex}[\chi_{\Dex}(X)] \\
\Sigma \ar[u]|{\Sigma^{\Dex}(X)} \ar[r]|{\chi} & \Sigma' \ar[u]|{\Sigma'^{\Dex}(\chi_{\Dex}(X))}
}$
\end{center}
This is a diagram that is often seen when explaining the translation of extensions as push-out.
Combining the above two, we get the following diagram. The commutativity of the outer rectangle of the diagram also holds true, which just shows the commutativity of natural transformations.
\begin{equation}\label{eq:IDex-2}
\xymatrix@R=18pt@C=60pt{
\Sigma^{\Dex}[X] \ar[ddd]|{\chi^{\Dex}[X]} \ar[rr]|{\Sigma^{\Dex}[\iota]} & & \Sigma^{\Dex}[Y] \ar[ddd]|{\chi^{\Dex}[Y]} \\
& \Sigma \ar[ul]|{\Sigma^{\Dex}(X)} \ar[ur]|{\Sigma^{\Dex}(Y)} \ar[d]|{\chi} & \\
& \Sigma' \ar[dl]|{\Sigma'^{\Dex}(\chi_{\Dex}(X))} \ar[dr]|{\Sigma'^{\Dex}(\chi_{\Dex}(Y))} & \\
\Sigma'^{\Dex}[\chi_{\Dex}(X)] \ar[rr]|{\Sigma'^{\Dex}[\chi_{\Dex}(\iota)]} & & \Sigma'^{\Dex}[\chi_{\Dex}(Y)]
}
\end{equation}
\item[the correspondence $\Dex:\Sig\ni(\chi:\Sigma\to\Sigma')\mapsto(\chi^{\Dex},\chi)\in\newcoprod_{\cdot\in\Sig}(\cdot/\Sig)^{(\cdot)_{\Dex}}$ is a functor.]
The functor condition can be divided into $\id$ and $\circ$ as follows:
\begin{proofcases}
\item[$\id$]
$\pos{\id_{\Sigma}^{\Dex},\id_{\Sigma}}=\pos{\id_{\Sigma^{\Dex}},\id_{\Sigma}}$
where $\Sigma\in|\Sig|$.
\item[$\circ$]
$\pos{(\chi'\circ\chi)^{\Dex},\chi'\circ\chi}
=\pos{\chi'^{\Dex},\chi'}\circ\pos{\chi^{\Dex},\chi}
=\pos{(\chi/\Sig)^{\chi_{\Dex}}(\chi'^{\Dex})\circ\chi^{\Dex},\chi'\circ\chi}
=\pos{\{\chi'^{\Dex}(\chi_{\Dex}(X))\}_{X\in|\Sigma_{\Dex}|}\circ\{\chi^{\Dex}(X)\}_{X\in|\Sigma_{\Dex}|},\chi'\circ\chi}
=\pos{\{\chi'^{\Dex}(\chi_{\Dex}(X))\circ\chi^{\Dex}(X)\}_{X\in|\Sigma_{\Dex}|},\chi'\circ\chi}$
where $\chi:\Sigma\to\Sigma',\chi':\Sigma'\to\Sigma''\in\Sig$.
\end{proofcases}
This means that the following "differential law" hold true:
\begin{equation}\label{eq:IDex-3}
\begin{split} 
&\id_{\Sigma}^{\Dex}[X]=\id_{\Sigma^{\Dex}[X]} \text{ for all } X\in|\Sigma_{\Dex}| \\
&(\chi'\circ\chi)^{\Dex}[X]=\chi'^{\Dex}[\chi_{\Dex}(X)]\circ\chi^{\Dex}[X] \text{ for all } X\in|\Sigma_{\Dex}|
\end{split}
\end{equation}
which will be used in translations later when quantifiers are introduced.
\end{proofcases}
In summary, $\Dex^{\INS}(\cdot)\in|\newprod_{\cdot\in\Sig^{\INS}}(\cdot/\Sig^{\INS})^{(\cdot)_{\Dex^{\INS}}}|$ consists of three components:
\begin{itemize}
\item $\{\Sigma^{\Dex}(X):\Sigma\to\Sigma^{\Dex}[X]\in\Sig\}_{\Sigma\in|\Sig|,\,X\in|\Sigma_{\Dex}|}$,
\item $\{\Sigma^{\Dex}[\iota]:\Sigma^{\Dex}[X]\to\Sigma^{\Dex}[X']\in\Sig\}_{\Sigma\in|\Sig|,\,\iota:X\to X'\in\Sigma_{\Dex}}$, and
\item $\{\chi^{\Dex}[X]:\Sigma^{\Dex}[X]\to\Sigma'^{\Dex}[\chi_{\Dex}(X)]\in\Sig\}_{\chi:\Sigma\to\Sigma'\in\Sig,\,X\in|\Sigma_{\Dex}|}$
\end{itemize}
which satisfy \ref{eq:IDex-1}, \ref{eq:IDex-2}, and \ref{eq:IDex-3}.
These conditions correspond to those assumed for traditional 
variable structures \cite[p114]{Diaconescu2025}.
The notion of lists category allows us to express them relatively succinctly and naturally.
\item
The condition for universal morphisms is, in other words, that the following diagram is a push-out:
\begin{center}
$\xymatrix@C=80pt{
\Sigma^{\Dex}[X] \ar[r]|{\chi^{\Dex}[X]} & \Sigma'^{\Dex}[\chi_{\Dex}(X)] \\
\Sigma \ar[u]|{\Sigma^{\Dex}(X)} \ar[r]|{\chi} & \Sigma' \ar[u]|{\Sigma'^{\Dex}(\chi_{\Dex}(X))}
}$
\end{center}
The reason we are writing it in the form of a universal morphism is because it makes it easier to see the correspondence with institution morphism.
We will not use this condition very often, but we will assume it because it is convenient when we discuss substitution, equality, and the introduction rules for existential quantification.
\item
$\Sen$ is initially assumed to contain only atomic formulas.
Other complex sentences are constructed recursively.
The condition on the model extension (the last condition) is used to ensure that the satisfaction condition also holds true for the extended $\Sen$.
By the way, the commutativity of the diagram in the definition can be seen from the $X$ part of \ref{eq:IDex-2}.
\end{enumerate}
\begin{example}[$\INS_{\FOL_{0}}$]
$\INS_{\FOL_{0}}=\pos{\Sig^{\INS_{\FOL_{0}}},\Mod^{\INS_{\FOL_{0}}},\Sen^{\INS_{\FOL_{0}}},\models^{\INS_{\FOL_{0}}};\Dex^{\INS_{\FOL_{0}}}}$ is ordinary many-sorted first-order logic with sentences restricted to atomic sentences only.
$\Dex^{\INS_{\FOL_{0}}}=\pos{(\cdot)_{\Dex^{\INS_{\FOL_{0}}}},\Dex^{\INS_{\FOL_{0}}}(\cdot)}$ is defined as follows:
\begin{itemize}
\item $(\cdot)_{\Dex}\in|\Cat^{\Sig}|$
\begin{itemize}
\item $\Sigma_{\Dex}:=\pos{\{X\subseteq_{fin}\{\pos{n,s,\Sigma}\mid n\in\omega,\,s\text{ is a sort of }\Sigma\}\mid s=s'\text{ for all }\pos{n,s,\Sigma},\pos{n,s',\Sigma}\in X\},{\subseteq}}$
\item $\chi_{\Dex}(\{\pos{n,s,\Sigma}\cdots\}\subseteq\{\pos{n,s,\Sigma}\cdots\}):=\{\pos{n,\chi(s),\Sigma'}\cdots\}\subseteq\{\pos{n,\chi(s),\Sigma'}\cdots\}$ where $\chi:\Sigma\to\Sigma'$
\end{itemize}
\item $\Dex(\cdot)\in|\newprod_{\cdot\in\Sig}(\cdot/\Sig)^{(\cdot)_{\Dex}}|$ consists of the following components:
\begin{itemize}
\item $\{\Sigma^{\Dex}(X):\Sigma\to\Sigma^{\Dex}[X]\in\Sig\}_{\Sigma\in|\Sig|,\,X\in|\Sigma_{\Dex}|}$ where $\Sigma^{\Dex}[X]:=\pos{S,F\cup X,P}$ and $\Sigma^{\Dex}(X)$ is the inclusion,
\item $\{\Sigma^{\Dex}[\iota]:\Sigma^{\Dex}[X]\to\Sigma^{\Dex}[X']\in\Sig\}_{\Sigma\in|\Sig|,\,\iota:X\to X'\in\Sigma_{\Dex}}$ where $\Sigma^{\Dex}[\iota]$ is the inclusion, and
\item $\{\chi^{\Dex}[X]:\Sigma^{\Dex}[X]\to\Sigma'^{\Dex}[\chi_{\Dex}(X)]\in\Sig\}_{\chi:\Sigma\to\Sigma'\in\Sig,\,X\in|\Sigma_{\Dex}|}$ where $\chi^{\Dex}[X]$ is a expansion of $\chi$ which assigns $\pos{n,\chi(s),\Sigma'}$ for each $\pos{n,s,\Sigma}\in X$.
\end{itemize}
\end{itemize}
\end{example}
\begin{example}[$\INS_{\CRing}$]
$\INS_{\CRing}=\pos{\Sig^{\INS_{\CRing}},\Mod^{\INS_{\CRing}},\Sen^{\INS_{\CRing}},\models^{\INS_{\CRing}};\Dex^{\INS_{\CRing}}}$ is defined as follows:
\begin{itemize}
\item $\Sig:=\CRing$ (Category of commutative rings)
\item $\Mod:\Sig^{op}\to\CAT$
\begin{itemize}
\item $\Mod(\Sigma):=\Sigma/\CRing$ (i.e. the category of $\Sigma$-algebras and $\Sigma$-homomorphisms in commutative ring theory)
\item $\Mod(\chi):=\chi/\CRing:\Sigma'/\CRing\to\Sigma/\CRing$
\end{itemize}
\item $\Sen:\Sig\to\Set$
\begin{itemize}
\item $\Sen(\Sigma):=\{t_{1}=t_{2}\mid t_{1},t_{2}\in\Sigma\}$
\item $\Sen(\chi)(t_{1}=t_{2}):=\chi(t_{1})=\chi(t_{2})$
\end{itemize}
\item $\A\models t_{1}=t_{2}$ iff $\A(t_{1})=\A(t_{2})$
\item $(\cdot)_{\Dex}\in|\Cat^{\Sig}|$
\begin{itemize}
\item $|\Sigma_{\Dex}|:=\mathcal{P}_{fin}\{\pos{n,\Sigma}\mid n\in\omega\}$, $\Sigma_{\Dex}:=\{\iota:X\to X'\,(injective)\mid X,X'\in|\Sigma_{\Dex}|\}$
\item $\chi_{\Dex}(\{\cdots\pos{n,\Sigma}\cdots\}):=\{\cdots\pos{n,\Sigma'}\cdots\}$ $\chi_{\Dex}(\iota)(\pos{n,\Sigma}):=\pos{n',\Sigma'}$ where $\iota(n,\Sigma)=\pos{n',\Sigma}$, $(\chi:\Sigma\to\Sigma')$
\end{itemize}
\item $\Dex(\cdot)\in|\newprod_{\cdot\in\Sig}(\cdot/\Sig)^{(\cdot)_{\Dex}}|$ consists of the following components:
\begin{itemize}
\item $\{\Sigma^{\Dex}(X):\Sigma\to\Sigma^{\Dex}[X]\in\Sig\}_{\Sigma\in|\Sig|,\,X\in|\Sigma_{\Dex}|}$ where $\Sigma^{\Dex}[X]$ is the polynomial ring and $\Sigma^{\Dex}(X)$ is the inclusion,
\item $\{\Sigma^{\Dex}[\iota]:\Sigma^{\Dex}[X]\to\Sigma^{\Dex}[X']\in\Sig\}_{\Sigma\in|\Sig|,\,\iota:X\to X'\in\Sigma_{\Dex}}$where $\Sigma^{\Dex}[\iota]$ is the injection, and
\item $\{\chi^{\Dex}[X]:\Sigma^{\Dex}[X]\to\Sigma'^{\Dex}[\chi_{\Dex}(X)]\in\Sig\}_{\chi:\Sigma\to\Sigma'\in\Sig,\,X\in|\Sigma_{\Dex}|}$ where $\chi^{\Dex}[X]$ is a expansion of $\chi$ which assigns $\pos{n,\Sigma'}$ for each $\pos{n,\Sigma}\in X$.
\end{itemize}
\end{itemize}
\end{example}
Like this, an algebraic structure itself can be regarded as a signature.
The operation of constructing a polynomial ring satisfies the axioms for extension by variables.
Since we use injective functions instead of inclusion here, we can also swap variables.
For example, by applying a group action to this, we can consider concepts similar to symmetric polynomials.
\begin{definition}[$\Dex$-institution morphisms]
$\INM=\pos{\Sig^{\INM},\Mod^{\INM},\Sen^{\INM};\Dex^{\INM}}:\INS\to\INS'$ is $\Dex$-institution morphism iff
\begin{itemize}
\item $\Sig^{\INM}:\Sig^{\INS}\to\Sig^{\INS'}\in\CAT$ is a functor
\item $\Mod^{\INM}:\Mod^{\INS}(\cdot)\Rightarrow\Mod^{\INS'}((\Sig^{\INM})^{op}(\cdot))$ is a natural transformation
\item $\Sen^{\INM}:\Sen^{\INS}(\cdot)\Leftarrow\Sen^{\INS'}(\Sig^{\INM}(\cdot))$ is a natural transformation
\item $\A\models^{\INS}\Sen^{\INM}(\Sigma)(\gamma')$ iff $\Mod^{\INM}(\Sigma)(\A)\models^{\INS'}\gamma'$ where $\Sigma\in\Sig^{\INS}$, $\A\in|\Mod^{\INS}(\Sigma)|$, and $\gamma'\in\Sen^{\INS'}(\Sig^{\INM}(\Sigma))$
\item $\Dex^{\INM}=\pos{(\cdot)_{\Dex^{\INM}}\in\Cat^{\Sig^{\INS}},\,\Dex^{\INM}(\cdot)\in\newprod_{\cdot\in\Sig^{\INS}}(\Sig^{\INM}(\cdot)/\Sig^{\INS'})^{\Sig^{\INM}(\cdot)_{\Dex^{\INS'}}}}$
\begin{itemize}
\item $(\cdot)_{\Dex^{\INM}}:(\cdot)_{\Dex^{\INS}}\Leftarrow\Sig^{\INM}(\cdot)_{\Dex^{\INS'}}$ is a natural transformation
\item $(\cdot)^{\Dex^{\INS}}$ and $(\cdot)^{\Dex^{\INS'}}$ cannot be compared as they are, so they are converted.
\begin{itemize}
\item $\Dex^{\INS}_{\INM}:=\newprod_{\Sig^{\INS}}((\cdot/\Sig^{\INM})^{(\cdot)_{\Dex^{\INM}}})(\Dex^{\INS})$
\item[]
\small
i.e. $\Dex^{\INS}_{\INM}(\chi:\Sigma\to\Sigma'):=\pos{(\Sigma/\Sig^{\INM})^{\Sigma_{\Dex^{\INM}}}(\chi^{\Dex^{\INS}}),\chi}:\pos{(\Sigma/\Sig^{\INM})^{\Sigma_{\Dex^{\INM}}}(\Sigma^{\Dex^{\INS}}),\Sigma}\to\pos{(\Sigma'/\Sig^{\INM})^{\Sigma'_{\Dex^{\INM}}}(\Sigma'^{\Dex^{\INS}}),\Sigma'}$
\\
where $(\Sigma/\Sig^{\INM}):\Sigma/\Sig^{\INS}\to\Sig^{\INM}(\Sigma)/\Sig^{\INS'}$ is the functor which assign $\Sig^{\INM}(\chi):\Sig^{\INM}(\chi')\to\Sig^{\INM}(\chi'')$ to $\chi:\chi'\to\chi''\in\Sigma/\Sig^{\INS}$.
This is also a natural transformation $(\cdot/\Sig^{\INM}):\cdot/\Sig^{\INS}\Rightarrow\Sig^{\INM}(\cdot)/\Sig^{\INS'}$.
Since $(\cdot)_{\Dex^{\INM}}:(\cdot)_{\Dex^{\INS}}\Leftarrow\Sig^{\INM}(\cdot)_{\Dex^{\INS'}}$ is also natural,
$(\cdot/\Sig^{\INM})^{(\cdot)_{\Dex^{\INM}}}=\Hom_{\CAT}((\cdot)_{\Dex^{\INM}}^{op},(\cdot/\Sig^{\INM})):(\cdot/\Sig^{\INS})^{(\cdot)_{\Dex^{\INS}}}\to(\Sig^{\INM}(\cdot)/\Sig^{\INS'})^{\Sig^{\INM}(\cdot)_{\Dex^{\INS'}}}\in\CAT^{\Sig^{\INS}}$.
Therefore we can apply $\newprod_{\Sig^{\INS}}:\CAT^{\Sig^{\INS}}\to\CAT'$.
\normalsize
\vspace{1.5mm}
\item $\Dex^{\INS'}_{\INM}:=\newprod_{\Sig^{\INM}}((\cdot/\Sig^{\INS'})^{(\cdot)_{\Dex^{\INS'}}})(\Dex^{\INS'})$
\item[]
\small
i.e. $\Dex^{\INS'}_{\INM}(\chi:\Sigma\to\Sigma'):=\pos{\Sig^{\INM}(\chi)^{\Dex^{\INS'}},\chi}:\pos{\Sig^{\INM}(\Sigma)^{\Dex^{\INS'}},\Sigma}\to\pos{\Sig^{\INM}(\Sigma')^{\Dex^{\INS'}},\Sigma'}$.
\normalsize
\end{itemize}
$\Dex^{\INM}(\cdot):\Dex^{\INS}_{\INM}(\cdot)\Leftarrow\Dex^{\INS'}_{\INM}(\cdot)\in\newprod_{\cdot\in\Sig^{\INS}}(\Sig^{\INM}(\cdot)/\Sig^{\INS'})^{\Sig^{\INM}(\cdot)_{\Dex^{\INS'}}}$
\end{itemize}
\item
For all $\Sigma\in|\Sig^{\INS}|$, and $X\in|\Sig^{\INM}(\Sigma)_{\Dex^{\INS'}}|$,
the pair $\pos{\Sigma^{\Dex^{\INM}}(X):\Sig^{\INM}(\Sigma)^{\Dex^{\INS'}}(X)\to(\Sigma/\Sig^{\INM})(\Sigma^{\Dex^{\INS}}(\Sigma_{\Dex^{\INM}}(X))),\Sigma^{\Dex^{\INS}}(\Sigma_{\Dex^{\INM}}(X))}$
is an universal morphism from $\Sig^{\INM}(\Sigma)^{\Dex^{\INS'}}(X)$ to $\Sigma/\Sig^{\INM}:\Sigma/\Sig^{\INS}\to\Sig^{\INM}(\Sigma)/\Sig^{\INS'}$
\item
For all $\Sigma\in|\Sig^{\INS}|$, $X\in|\Sig^{\INM}(\Sigma)_{\Dex^{\INS'}}|$, $\A\in|\Mod^{\INS}(\Sigma)|$, and
$\Sig^{\INM}(\Sigma)^{\Dex^{\INS'}}(X)$-expansion $\hat{\A}'$ of $\A':=\Mod^{\INM}(\Sigma)(\A)$,
there is unique $\Sigma^{\Dex^{\INS}}(\Sigma_{\Dex^{\INM}}(X))$-expansion $\hat{\A}$ of $\A$ such that
$(\Sigma^{\Dex^{\INM}}*\Mod^{\INM})(X)(\hat{\A})=\hat{\A}'$.
\begin{center}
$\xymatrix@R=40pt@C=40pt{
\hat{\A} \ar@{.}[r] \ar@{.>}[d]|{\red} & \Mod^{\INS}(\Sigma^{\Dex^{\INS}}[\Sigma_{\Dex^{\INM}}(X)]) \ar[rrr]|{(\Sigma^{\Dex^{\INM}}*\Mod^{\INM})(X)} \ar[d]|{\Mod^{\INS}(\Sigma^{\Dex^{\INS}}(\Sigma_{\Dex^{\INM}}(X)))} &&& \Mod^{\INS'}(\Sig^{\INM}(\Sigma)^{\Dex^{\INS'}}[X]) \ar[d]|{\Mod^{\INS'}(\Sig^{\INM}(\Sigma)^{\Dex^{\INS'}}(X))} & \ar@{.}[l] \hat{\A}' \ar@{.>}[d]|{\red} \\
\A \ar@{.}[r] & \Mod^{\INS}(\Sigma) \ar[rrr]|{\Mod^{\INM}(\Sigma)} &&& \Mod^{\INS'}(\Sig^{\INM}(\Sigma)) & \ar@{.}[l] \A'
}$
\end{center}
where $(\Sigma^{\Dex^{\INM}}*\Mod^{\INM}):=(\Mod^{\INS'}\circ_{H}(\Sigma^{\Dex^{\INM}}[\cdot])^{op})\circ_{V}(\Mod^{\INM}\circ_{H}(\Sigma^{\Dex^{\INS}}[\Sigma_{\Dex^{\INM}}(\cdot)])^{op})$
\begin{center}\small
\def\objectstyle{\scriptstyle}
\def\labelstyle{\scriptscriptstyle}
$\xymatrix@R=40pt@C=60pt{
&&&\ar@{}[d]|{\Uparrow\,(\Sigma^{\Dex^{\INM}}[\cdot])^{op}}&& \\
\CAT&\ar[l]^{\Mod^{\INS'}}(\Sig^{\INS'})^{op}&\ar[l]^{(\Sig^{\INM})^{op}}(\Sig^{\INS})^{op}\ar@/^50pt/[ll]^{\Mod^{\INS}}&&\ar[ll]^{(\Sigma^{\Dex^{\INS}}[\cdot])^{op}}(\Sigma_{\Dex^{\INS}})^{op}&\ar[l]^{(\Sigma_{\Dex^{\INM}})^{op}}(\Sig^{\INM}(\Sigma)_{\Dex^{\INS'}})^{op}\ar@/_50pt/[llll]_{(\Sig^{\INM}(\Sigma)^{\Dex^{\INS'}}[\cdot])^{op}} \\
&\ar@{}[u]|{\Uparrow\Mod^{\INM}}&&&&
}$
\end{center}
\end{itemize}
\end{definition}
\begin{enumerate}
\item
The part before $\Dex$ is the definition of institution morphism.
Intuitively, $\INS$ has more sentences than $\INS'$, and instead the class of models is constrained.
\item
$(\cdot)_{\Dex^{\INM}}:(\cdot)_{\Dex^{\INS}}\Leftarrow(\Sig^{\INM}(\cdot))_{\Dex^{\INS'}}$ represents how the set of variables in $\INS'$ is translated into $\INS$.
For example, if $\INS'$ is ordinary first-order logic and $\INS$ allows quantification of an infinite number of variables, this will not be surjective.
\item
$\Dex^{\INM}(\cdot):\Dex^{\INS}_{\INM}(\cdot)\Leftarrow\Dex^{\INS'}_{\INM}(\cdot)\in\newprod_{\cdot\in\Sig^{\INS}}(\Sig^{\INM}(\cdot)/\Sig^{\INS'})^{\Sig^{\INM}(\cdot)_{\Dex^{\INS'}}}$ can be rephrased as
$\{\pos{\Sigma^{\Dex^{\INM}},\id_{\Sigma}}\}_{\Sigma\in|\Sig^{\INS}|} :\{\pos{\chi^{\Dex^{\INS}_{\INM}},\chi}\}_{\chi\in\Sig^{\INS}}\Leftarrow\{\pos{\chi^{\Dex^{\INS'}_{\INM}},\chi}\}_{\chi\in\Sig^{\INS}}\in(\newcoprod_{\cdot\in\Sig^{\INS}}(\Sig^{\INM}(\cdot)/\Sig^{\INS'})^{\Sig^{\INM}(\cdot)_{\Dex^{\INS'}}})^{\Sig^{\INS}}$
(The condition on $\pi$ is implied by the right component being $\id$.)
\begin{proofcases}
\item[$\pos{\Sigma^{\Dex^{\INM}},\id_{\Sigma}}:\pos{\Sigma^{\Dex^{\INS}_{\INM}},\Sigma}\gets\pos{\Sigma^{\Dex^{\INS'}_{\INM}},\Sigma}\in(\Sig^{\INM}(\Sigma)/\Sig^{\INS'})^{\Sig^{\INM}(\Sigma)_{\Dex^{\INS'}}}\times\{\id_{\Sigma}:\Sigma\to\Sigma\}$]
Before discussing the commutativity of natural transformations, we first look at the condition for the destination to be contained in the target. Note that the elements of $\newcoprod$ are labeled, so we have ($\times\{\id_{\Sigma}:\Sigma\to\Sigma\}$).
However, the condition for the right-hand component is clear from the form, so this condition can be rephrased as $\Sigma^{\Dex^{\INM}}:\Sigma^{\Dex^{\INS}_{\INM}}\Leftarrow\Sigma^{\Dex^{\INS'}_{\INM}}\in(\Sig^{\INM}(\Sigma)/\Sig^{\INS'})^{\Sig^{\INM}(\Sigma)_{\Dex^{\INS'}}}$.
That is, for any $\iota:X\to Y\in\Sig^{\INM}(\Sigma)_{\Dex^{\INS'}}$, the following is commutative. (Note that this is a co-slice category.)
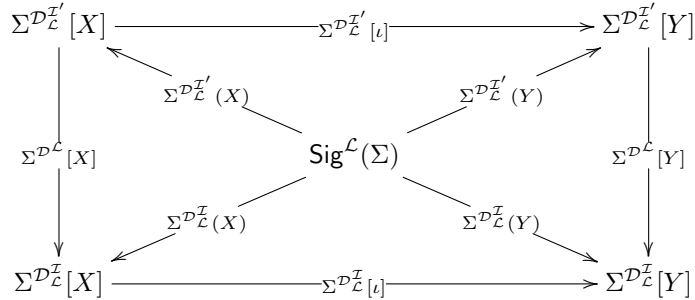
\begin{center}
$\xymatrix@R=30pt@C=70pt{
\Sigma^{\Dex^{\INS'}_{\INM}}[X] \ar[rr]|{\Sigma^{\Dex^{\INS'}_{\INM}}[\iota]} \ar[dd]|{\Sigma^{\Dex^{\INM}}[X]} & & \Sigma^{\Dex^{\INS'}_{\INM}}[Y] \ar[dd]|{\Sigma^{\Dex^{\INM}}[Y]} \\
& \Sig^{\INM}(\Sigma) \ar[ul]|{\Sigma^{\Dex^{\INS'}_{\INM}}(X)} \ar[ur]|{\Sigma^{\Dex^{\INS'}_{\INM}}(Y)} \ar[dl]|{\Sigma^{\Dex^{\INS}_{\INM}}(X)} \ar[dr]|{\Sigma^{\Dex^{\INS}_{\INM}}(Y)} & \\
\Sigma^{\Dex^{\INS}_{\INM}}[X] \ar[rr]|{\Sigma^{\Dex^{\INS}_{\INM}}[\iota]} & & \Sigma^{\Dex^{\INS}_{\INM}}[Y]
}$
\end{center}
The commutativity as a coslice category (commutativity of only $X$ or $Y$) is as follows:
\begin{align*}
&\Sigma^{\Dex^{\INS}_{\INM}}(X)
=(\Sigma/\Sig^{\INM})^{\Sigma_{\Dex^{\INM}}}(\Sigma^{\Dex^{\INS}})(X)
=\Sig^{\INM}(\Sigma^{\Dex^{\INS}}(\Sigma_{\Dex^{\INM}}(X))) \\
&\Sigma^{\Dex^{\INM}}[X]\circ\Sigma^{\Dex^{\INS'}_{\INM}}(X)
=\Sigma^{\Dex^{\INM}}[X]\circ\Sig^{\INM}(\Sigma)^{\Dex^{\INS'}}(X)
\end{align*}
\begin{equation}\label{eq:MDex-1}
\text{Therefore, }\Sig^{\INM}(\Sigma^{\Dex^{\INS}}(\Sigma_{\Dex^{\INM}}(X)))=\Sigma^{\Dex^{\INM}}[X]\circ\Sig^{\INM}(\Sigma)^{\Dex^{\INS'}}(X)
\end{equation}
The commutativity of the outer frame is as follows:
\begin{align*}
&\Sigma^{\Dex^{\INS}_{\INM}}(\iota)\circ\Sigma^{\Dex^{\INM}}(X)
=(\Sigma/\Sig^{\INM})^{\Sigma_{\Dex^{\INM}}}(\Sigma^{\Dex^{\INS}})(\iota)\circ\Sigma^{\Dex^{\INM}}(X)
=\Sig^{\INM}(\Sigma^{\Dex^{\INS}}(\Sigma_{\Dex^{\INM}}(\iota)))\circ\Sigma^{\Dex^{\INM}}(X) \\
&\Sigma^{\Dex^{\INM}}(Y)\circ\Sigma^{\Dex^{\INS'}_{\INM}}(\iota)
=\Sigma^{\Dex^{\INM}}(Y)\circ\Sig^{\INM}(\Sigma)^{\Dex^{\INS'}}(\iota)
\end{align*}
\begin{equation}\label{eq:MDex-2}
\text{Therefore, }\Sig^{\INM}(\Sigma^{\Dex^{\INS}}(\Sigma_{\Dex^{\INM}}[\iota]))\circ\Sigma^{\Dex^{\INM}}[X]=\Sigma^{\Dex^{\INM}}[Y]\circ\Sig^{\INM}(\Sigma)^{\Dex^{\INS'}}[\iota]
\end{equation}
\item[$\{\pos{\Sigma^{\Dex^{\INM}},\id_{\Sigma}}\}_{\Sigma\in|\Sig^{\INS}|}:\{\pos{\chi^{\Dex^{\INS}_{\INM}},\chi}\}_{\chi\in\Sig^{\INS}}\Leftarrow\{\pos{\chi^{\Dex^{\INS'}_{\INM}},\chi}\}_{\chi\in\Sig^{\INS}}\in(\newcoprod_{\cdot\in\Sig^{\INS}}(\Sig^{\INM}(\cdot)/\Sig^{\INS'})^{\Sig^{\INM}(\cdot)_{\Dex^{\INS'}}})^{\Sig^{\INS}}$]
we look at the overall commutativity.
\begin{center}
$\xymatrix@R=30pt@C=70pt{
\pos{\Sigma^{\Dex^{\INS'}_{\INM}},\Sigma} \ar[r]|{\pos{\chi^{\Dex^{\INS'}_{\INM}},\chi}} \ar[d]|{\pos{\Sigma^{\Dex^{\INM}},\id_{\Sigma}}} & \pos{\Sigma'^{\Dex^{\INS'}_{\INM}},\Sigma'} \ar[d]|{\pos{\Sigma'^{\Dex^{\INM}},\id_{\Sigma'}}} \\
\pos{\Sigma^{\Dex^{\INS}_{\INM}},\Sigma} \ar[r]|{\pos{\chi^{\Dex^{\INS}_{\INM}},\chi}} & \pos{\Sigma'^{\Dex^{\INS}_{\INM}},\Sigma'}
}$
\end{center}
\begin{align*}
&\pos{\chi^{\Dex^{\INS}_{\INM}},\chi}\circ\pos{\Sigma^{\Dex^{\INM}},\id_{\Sigma}}
=\pos{\chi^{\Dex^{\INS}_{\INM}}\circ_{V}\Sigma^{\Dex^{\INM}},\chi} \\
&=\pos{(\Sigma/\Sig^{\INM})^{\Sigma_{\Dex^{\INM}}}(\chi^{\Dex^{\INS}})\circ_{V}\Sigma^{\Dex^{\INM}},\chi}
=\pos{(\Sigma/\Sig^{\INM}\circ\chi^{\Dex^{\INS}}\circ\Sigma_{\Dex^{\INM}})\circ_{V}\Sigma^{\Dex^{\INM}},\chi} \\
&\pos{\Sigma'^{\Dex^{\INM}},\id_{\Sigma'}}\circ\pos{\chi^{\Dex^{\INS'}_{\INM}},\chi}
=\pos{(\Sig^{\INM}(\chi)/\Sig^{\INS'})^{\Sig^{\INM}(\chi)_{\Dex^{\INS'}}}(\Sigma'^{\Dex^{\INM}})\circ_{V}\chi^{\Dex^{\INS'}_{\INM}},\chi} \\
&=\pos{(\Sig^{\INM}(\chi)/\Sig^{\INS'})^{\Sig^{\INM}(\chi)_{\Dex^{\INS'}}}(\Sigma'^{\Dex^{\INM}})\circ_{V}\Sig^{\INM}(\chi)^{\Dex^{\INS'}},\chi}
=\pos{(\Sig^{\INM}(\chi)/\Sig^{\INS'}\circ\Sigma'^{\Dex^{\INM}}\circ\Sig^{\INM}(\chi)_{\Dex^{\INS'}})\circ_{V}\Sig^{\INM}(\chi)^{\Dex^{\INS'}},\chi}
\end{align*}
It is clear that the right components are equal, so we compare the left.
Let $X\in\Sig^{\INM}(\Sigma)_{\Dex^{\INS'}}$.
\begin{align*}
&((\Sigma/\Sig^{\INM}\circ\chi^{\Dex^{\INS}}\circ\Sigma_{\Dex^{\INM}})\circ_{V}\Sigma^{\Dex^{\INM}})[X]
=\Sig^{\INM}(\chi^{\Dex^{\INS}}[\Sigma_{\Dex^{\INM}}(X)])\circ\Sigma^{\Dex^{\INM}}[X] \\
&((\Sig^{\INM}(\chi)/\Sig^{\INS'}\circ\Sigma'^{\Dex^{\INM}}\circ\Sig^{\INM}(\chi)_{\Dex^{\INS'}})\circ_{V}\Sig^{\INM}(\chi)^{\Dex^{\INS'}})[X]
=\Sigma'^{\Dex^{\INM}}[\Sig^{\INM}(\chi)_{\Dex^{\INS'}}(X)]
\circ\Sig^{\INM}(\chi)^{\Dex^{\INS'}}[X]
\end{align*}
\begin{equation}\label{eq:MDex-3}
\text{Therefore, }
\Sig^{\INM}(\chi^{\Dex^{\INS}}[\Sigma_{\Dex^{\INM}}(X)])\circ\Sigma^{\Dex^{\INM}}[X]
=
\Sigma'^{\Dex^{\INM}}[\Sig^{\INM}(\chi)_{\Dex^{\INS'}}(X)]\circ\Sig^{\INM}(\chi)^{\Dex^{\INS'}}[X]
\end{equation}
\end{proofcases}
In summary, $\Dex^{\INM}(\cdot):\Dex^{\INS}_{\INM}(\cdot)\Leftarrow\Dex^{\INS'}_{\INM}(\cdot)\in\newprod_{\cdot\in\Sig^{\INS}}(\Sig^{\INM}(\cdot)/\Sig^{\INS'})^{\Sig^{\INM}(\cdot)_{\Dex^{\INS'}}}$ is consists of;
\begin{itemize}
\item $\{\Sigma^{\Dex^{\INM}}[X]:\Sigma^{\Dex^{\INS'}_{\INM}}[X]\to\Sigma^{\Dex^{\INS}_{\INM}}[X]\,(:\Sig^{\INM}(\Sigma)^{\Dex^{\INS'}}[X]\to\Sig^{\INM}(\Sigma^{\Dex^{\INS}}[\Sigma_{\Dex^{\INM}}(X)]))\in\Sig^{\INS'}\}_{\Sigma\in|\Sig^{\INS}|,\,X\in|\Sig^{\INM}(\Sigma)_{\Dex^{\INS'}}|}$
\end{itemize}
which satisfies \ref{eq:MDex-1}, \ref{eq:MDex-2}, and \ref{eq:MDex-3}.
{\par}
$\Sigma^{\Dex^{\INS}_{\INM}}[X]$ and
$\Sigma^{\Dex^{\INS'}_{\INM}} [X]$ are
written in a complex manner,
but essentially the difference lies in whether we first translate $\Sigma$ using $\Sig^{\INM}$ and then extend it based on the $X$ waiting ahead,
or conversely, pull back $X$ using $\Sigma_{\Dex^{\INM}}$ and use it to extend $\Sigma$.
$\Sigma^{\Dex^{\INM}}[X]$ serves to connect the two.
The following diagram will help in understanding the operations concerning extensions that will appear from here on.
\begin{center}
$\xymatrix@R=30pt@C=70pt{
\Sig^{\INM}(\Sigma)^{\Dex^{\INS'}}[X]\ar[r]^{\Sigma^{\Dex^{\INM}}[X]}\ar@{<-}[dd]|{\Sig^{\INM}(\Sigma)^{\Dex^{\INS'}}(X)}&\Sig^{\INM}(\Sigma^{\Dex^{\INS}}[\Sigma_{\Dex^{\INM}}(X)])\ar@{<-}[ddl]|{\Sig^{\INM}(\Sigma^{\Dex^{\INS}}(\Sigma_{\Dex^{\INM}}(X)))}&\Sigma^{\Dex^{\INS}}[\Sigma_{\Dex^{\INM}}(X)]\ar@{<-}[dd]|{\Sigma^{\Dex^{\INS}}(\Sigma_{\Dex^{\INM}}(X))}\ar@{--}[l]\\
&&\ar@{}[l]|{\leftarrow\Sig^{\INM}\text{\rotatebox[origin=c]{180}{$\vdash$}}\Space\Space}\\
\Sig^{\INM}(\Sigma)&&\Sigma\ar@{--}[ll]
}$
\end{center}
For example,
$(\Sigma^{\Dex^{\INM}}*\Mod^{\INM}):=(\Mod^{\INS'}\circ_{H}(\Sigma^{\Dex^{\INM}}[\cdot])^{op})\circ_{V}(\Mod^{\INM}\circ_{H}(\Sigma^{\Dex^{\INS}}[\Sigma_{\Dex^{\INM}}(\cdot)])^{op})$
means a step-by-step translation of the $\Sigma^{\Dex^{\INS}}[\Sigma_{\Dex^{\INM}}(X)]$-model along the top two edges of the diagram.
\item We rarely use the condition on universal morphisms, but we assume it because it simplifies discussions involving substitution.
\item As with institution, this is a condition for extending the set of sentences. The commutativity of the diagrams in the definition can be seen as follows:
\begin{align*}
&\Mod^{\INM}(\Sigma)\circ\Mod^{\INS}(\Sigma^{\Dex^{\INS}}(\Sigma_{\Dex^{\INM}}(X))) \\
&=\Mod^{\INS'}(\Sig^{\INM}(\Sigma^{\Dex^{\INS}}(\Sigma_{\Dex^{\INM}}(X))))\circ\Mod^{\INM}(\Sigma^{\Dex^{\INS}}[\Sigma_{\Dex^{\INM}}(X)])
\text{ since } \Mod^{\INM} \text{ is natural } \\
&=\Mod^{\INS'}(\Sigma^{\Dex^{\INM}}[X]\circ\Sig^{\INM}(\Sigma)^{\Dex^{\INS'}}(X))\circ\Mod^{\INM}(\Sigma^{\Dex^{\INS}}[\Sigma_{\Dex^{\INM}}(X)])
\text{ by \ref{eq:MDex-1} } \\
&=\Mod^{\INS'}(\Sig^{\INM}(\Sigma)^{\Dex^{\INS'}}(X))\circ\Mod^{\INS'}(\Sigma^{\Dex^{\INM}}[X])\circ\Mod^{\INM}(\Sigma^{\Dex^{\INS}}[\Sigma_{\Dex^{\INM}}(X)])
\\
&=\Mod^{\INS'}(\Sig^{\INM}(\Sigma)^{\Dex^{\INS'}}(X))\circ(\Sigma^{\Dex^{\INM}}*\Mod^{\INM})(X)
\text{ by the definition } \\
\end{align*}
\end{enumerate}
\begin{definition}\label{def:DINS}
We define the category $\DINS$\footnote{$\DINS$ is not exactly a category, since it is as large as $\CAT$.} of $\Dex$-institutions as follows:
\begin{itemize}
\item $\INS\in|\DINS|$ iff $\INS$ is a $\Dex$-institution
\item $\INM:\INS\to\INS'\in\DINS$ iff $\INM:\INS\to\INS'$ is a $\Dex$-institution morphism
\item Let $\INS\in|\DINS|$. We define $\id_{\INS}$ as follows:
\begin{itemize}
\item $\Sig^{\id_{\INS}}:=\id_{\Sig^{\INS}}:\Sig^{\INS}\to\Sig^{\INS}\in\CAT$
\item $\Mod^{\id_{\INS}}:=\id_{\Mod^{\INS}}:\Mod^{\INS}\Rightarrow\Mod^{\INS}\,\,(=\Mod^{\INS}((\Sig^{\id_{\INS}})^{op}(\cdot)))\in\CAT^{(\Sig^{\INS})^{op}}$
\item $\Sen^{\id_{\INS}}:=\id_{\Sen^{\INS}}:\Sen^{\INS}\Leftarrow\Sen^{\INS}\,\,(=\Sen^{\INS}(\Sig^{\id_{\INS}}(\cdot)))\in\Set^{\Sig^{\INS}}$
\item $(\cdot)_{\Dex^{\id_{\INS}}}:=\id_{(\cdot)_{\Dex^{\INS}}}:(\cdot)_{\Dex^{\INS}}\Leftarrow(\cdot)_{\Dex^{\INS}}\,\,(=\Sig^{\id_{\INS}}(\cdot)_{\Dex^{\INS}})\in\Cat^{\Sig^{\INS}}$
\item $\Dex^{\id_{\INS}}(\cdot):=\id_{\Dex^{\INS}}:\Dex^{\INS}\Leftarrow\Dex^{\INS}$
\end{itemize}
\item Let $\INM:\INS\to\INS',\,\INM':\INS'\to\INS''\in\DINS$. We define $\INM'\circ\INM$ as follows:
\begin{itemize}
\item $\Sig^{\INM'\circ\INM}:=\Sig^{\INM'}\circ\Sig^{\INM}:\Sig^{\INS}\to\Sig^{\INS''}\in\CAT$
\item $\Mod^{\INM'\circ\INM}:=(\Mod^{\INM'}\circ_{H}\Sig^{\INM\,op})\circ_{V}\Mod^{\INM}:\Mod^{\INS}(\cdot)\Rightarrow\Mod^{\INS''}((\Sig^{\INM'\circ\INM})^{op}(\cdot))\in\CAT^{(\Sig^{\INS})^{op}}$
\item $\Sen^{\INM'\circ\INM}:=\Sen^{\INM}\circ_{V}(\Sen^{\INM'}\circ_{H}\Sig^{\INM}):\Sen^{\INS}(\cdot)\Leftarrow\Sen^{\INS''}(\Sig^{\INM'\circ\INM}(\cdot))\in\Set^{\Sig^{\INS}}$
\item $(\cdot)_{\Dex^{\INM'\circ\INM}}:=(\cdot)_{\Dex^{\INM}}\circ_{V}((\cdot)_{\Dex^{\INM'}}\circ_{H}\Sig^{\INM}):(\cdot)_{\Dex^{\INS}}\Leftarrow\Sig^{\INM'\circ\INM}(\cdot)_{\Dex^{\INS''}}\in\Cat^{\Sig^{\INS}}$
\item
$
\Dex^{\INM'\circ\INM}:=\Dex^{\INM}_{\INM'}\circ_{V}\Dex^{\INM'}_{\INM}
$
where
\begin{itemize}
\item $\Dex^{\INM}_{\INM'}:=\newprod_{\Sig^{\INS}}((\Sig^{\INM}(\cdot)/\Sig^{\INM'})^{\Sig^{\INM}(\cdot)_{\Dex^{\INM'}}})(\Dex^{\INM})$
$=\newcoprod_{\Sig^{\INS}}(\Sig^{\INM}(\cdot)/\Sig^{\INM'})^{\Sig^{\INM}(\cdot)_{\Dex^{\INM'}}}\circ\Dex^{\INM}$
\item[] $=\{\pos{
\{
(\Sig^{\INM}(\Sigma)/\Sig^{\INM'})(\Sigma^{\Dex^{\INM}}(\Sig^{\INM}(\Sigma)_{\Dex^{\INM'}}(X)))
\}_{X\in|\Sig^{\INM'\circ\INM}(\Sigma)_{\Dex^{\INS''}}|},\id_{\Sigma}
}\}_{\Sigma\in|\Sig^{\INS}|}$
\item $\Dex^{\INM'}_{\INM}:=\newprod_{\Sig^{\INM}}((\Sig^{\INM'}(\cdot)/\Sig^{\INS''})^{\Sig^{\INM'}(\cdot)_{\Dex^{\INS''}}})(\Dex^{\INM'})$
$=\{\pos{\Sig^{\INM}(\Sigma)^{\Dex^{\INM'}},\id_{\Sigma}}\}_{\Sigma\in|\Sig^{\INS}|}$
\item[] $=\{\pos{\{\Sig^{\INM}(\Sigma)^{\Dex^{\INM'}}(X)\}_{X\in|\Sig^{\INM'\circ\INM}(\Sigma)_{\Dex^{\INS''}}|},\id_{\Sigma}}\}_{\Sigma\in|\Sig^{\INS}|}$
\end{itemize}
\end{itemize}
\end{itemize}
\end{definition}
\begin{lemma}\label{lemma:S*M}
Corresponding to the case of $\Mod$, we can also define the following for $\Sen$:
\begin{itemize}
\item $(\Sigma^{\Dex^{\INM}}*\Mod^{\INM})(X')=\Mod^{\INS'}(\Sigma^{\Dex^{\INM}}[X'])\circ\Mod^{\INM}(\Sigma^{\Dex^{\INS}}[\Sigma_{\Dex^{\INM}}(X')])$
\item $(\Sen^{\INM}*\Sigma^{\Dex^{\INM}})(X'):=\Sen^{\INM}(\Sigma^{\Dex^{\INS}}[\Sigma_{\Dex^{\INM}}(X')])\circ\Sen^{\INS'}(\Sigma^{\Dex^{\INM}}[X'])$
\end{itemize}
For these, the following holds true:
\begin{itemize}
\item For each $\INS\in|\DINS|$, $\Sigma\in|\Sig^{\INS}|$, $X\in|\Sigma_{\Dex^{\INS}}|$.
The following hold true.
\begin{itemize}
\item $(\Sigma^{\Dex^{\id_{\INS}}}*\Mod^{\id_{\INS}})(X)=\id_{\Mod^{\INS}(\Sigma^{\Dex^{\INS}}[X])}$
\item $(\Sen^{\id_{\INS}}*\Sigma^{\Dex^{\id_{\INS}}})(X)=\id_{\Sen^{\INS}(\Sigma^{\Dex^{\INS}}[X])}$
\end{itemize}
\item For each $\INM:\INS\to\INS',\,\INM':\INS'\to\INS''\in\DINS$,
\begin{itemize}
\item $\Sigma\in|\Sig^{\INS}|$, $\Sigma':=\Sig^{\INM}(\Sigma)\in|\Sig^{\INS}|$, $\Sigma'':=\Sig^{\INM'}(\Sigma')\in|\Sig^{\INS}|$,
\item $X''\in|\Sigma''_{\Dex^{\INS''}}|$, $X':=\Sigma'_{\Dex^{\INM'}}(X'')\in|\Sigma'_{\Dex^{\INS'}}|$, $X:=\Sigma_{\Dex^{\INM}}(X')\in|\Sigma_{\Dex^{\INS}}|$.
\end{itemize}
The following hold true.
\begin{itemize}
\item $(\Sigma^{\Dex^{\INM'\circ\INM}}*\Mod^{\INM'\circ\INM})(X'')=(\Sigma'^{\Dex^{\INM'}}*\Mod^{\INM'})(X'')\circ(\Sigma^{\Dex^{\INM}}*\Mod^{\INM})(X')$
\item $(\Sen^{\INM'\circ\INM}*\Sigma^{\Dex^{\INM'\circ\INM}})(X'')=(\Sen^{\INM}*\Sigma^{\Dex^{\INM}})(X')\circ(\Sen^{\INM'}*\Sigma'^{\Dex^{\INM'}})(X'')$
\end{itemize}
\end{itemize}
\end{lemma}
Since $\DINS$ is defined as a category, the introduction of compound sentences can be defined as a functor on $\DINS$.
We use the notation $\FOL:\DINS\to\DINS$, but this can also be applied to logic with second-order variables.
\begin{definition}\label{def:FOL:}
Let $\alpha$ be an infinite cardinal.
We define the functor $\FOL^{{\relax}}_{\alpha}:\DINS\to\DINS$ as follows.
\begin{itemize}
\item Let $\INS\in|\DINS|$. We define $\FOL^{{\relax}}_{\alpha}(\INS)$ as follows:
\begin{itemize}
\item Same as $\INS$ except for $\Sen$ and $\models$.
\item We define $\Sen^{\FOL^{{\relax}}_{\alpha}(\INS)}$ as follows:
\begin{itemize}
\item for $\Sigma\in|\Sig^{\INS}|(=|\Sig^{\FOL^{{\relax}}_{\alpha}(\INS)}|)$, we define $\Sen^{\FOL^{{\relax}}_{\alpha}(\INS)}(\Sigma)$ recursively as follows:
\begin{itemize}
\item $\pos{\phi,atomic}\in\Sen^{\FOL^{{\relax}}_{\alpha}(\INS)}(\Sigma)$ if $\phi\in\Sen^{\INS}(\Sigma)$
\footnote{"atomic" is a label that allows us to distinguish atomic formulas symbolically. It may be omitted if it is clear from the context.}
\item $\neg\phi:=\pos{\phi,\neg}\in\Sen^{\FOL^{{\relax}}_{\alpha}(\INS)}(\Sigma)$ if $\phi\in\Sen^{\FOL^{{\relax}}_{\alpha}(\INS)}(\Sigma)$
\item $\vee\phi:=\pos{\phi,\vee}\in\Sen^{\FOL^{{\relax}}_{\alpha}(\INS)}(\Sigma)$ if $\phi:N\to\Sen^{\FOL^{{\relax}}_{\alpha}(\INS)}(\Sigma)$, and $N\in\alpha$
\item $\Exists{X}\phi:=\pos{\pos{\phi,X},\exists}\in\Sen^{\FOL^{{\relax}}_{\alpha}(\INS)}(\Sigma)$ if $X\in|\Sigma_{\Dex^{\INS}}|$, and $\phi\in\Sen^{\FOL^{{\relax}}_{\alpha}(\INS)}(\Sigma^{\Dex^{\INS}}[X])$\footnote{As we will explain in the proof of well-definedness, we can also introduce $\exists^{\kappa}$ (there exist $\kappa$ instances).}
\end{itemize}
\item for $\chi:\Sigma\to\Sigma'\in\Sig^{\INS}(=\Sig^{\FOL^{{\relax}}_{\alpha}(\INS)})$, we define $\Sen^{\FOL^{{\relax}}_{\alpha}(\INS)}(\chi):\Sen^{\FOL^{{\relax}}_{\alpha}(\INS)}(\Sigma)\to\Sen^{\FOL^{{\relax}}_{\alpha}(\INS)}(\Sigma')$ recursively as follows:
\begin{itemize}
\item $\Sen^{\FOL^{{\relax}}_{\alpha}(\INS)}(\chi)(\pos{\phi,atomic}):=\pos{\Sen^{\INS}(\chi)(\phi),atomic}$
\item $\Sen^{\FOL^{{\relax}}_{\alpha}(\INS)}(\chi)(\neg\phi):=\neg\Sen^{\FOL^{{\relax}}_{\alpha}(\INS)}(\chi)(\phi)$
\item $\Sen^{\FOL^{{\relax}}_{\alpha}(\INS)}(\chi)(\vee\phi):=\vee\pos{\Sen^{\FOL^{{\relax}}_{\alpha}(\INS)}(\chi)(\phi_{i})\mid i\in\Dom(\phi)}$
\item $\Sen^{\FOL^{{\relax}}_{\alpha}(\INS)}(\chi)(\Exists{X}\phi):=\Exists{\chi_{\Dex^{\INS}}(X)}\Sen^{\FOL^{{\relax}}_{\alpha}(\INS)}(\chi^{\Dex^{\INS}}[X])(\phi)$
\end{itemize}
\end{itemize}
\item We define ${\models^{\FOL^{{\relax}}_{\alpha}(\INS)}}$ as follows:
\begin{itemize}
\item $\A\models^{\FOL^{{\relax}}_{\alpha}(\INS)}\pos{\phi,atomic}$ if $\A\models^{\INS}\phi$
\item $\A\models^{\FOL^{{\relax}}_{\alpha}(\INS)}\neg\phi$ if $\A\not\models^{\FOL^{{\relax}}_{\alpha}(\INS)}\phi$
\item $\A\models^{\FOL^{{\relax}}_{\alpha}(\INS)}\vee\phi$ if $\A\models^{\FOL^{{\relax}}_{\alpha}(\INS)}\phi_{i}$ for some $i\in\Dom(\phi)$
\item $\A\models^{\FOL^{{\relax}}_{\alpha}(\INS)}\Exists{X}\phi$ if $\hat{\A}\models^{\FOL^{{\relax}}_{\alpha}(\INS)}\phi$ for some $\Sigma^{\Dex^{\INS}}(X)$-expansion $\hat{\A}$ of $\A$
\end{itemize}
\end{itemize}
\item Let $\INM:\INS\to\INS'\in\DINS$. We define $\FOL^{{\relax}}_{\alpha}(\INM)$ as follows:
\begin{itemize}
\item Same as $\INM:\INS\to\INS'$ except for $\Sen$.
\item We define $\Sen^{\FOL^{{\relax}}_{\alpha}(\INM)}:\Sen^{\FOL^{{\relax}}_{\alpha}(\INS)}\Leftarrow\Sen^{\FOL^{{\relax}}_{\alpha}(\INS')}\circ\Sig^{\INM}$ as follows:
\begin{itemize}
\item Let $\Sigma\in|\Sig^{\INS}|$. We define $\Sen^{\FOL^{{\relax}}_{\alpha}(\INM)}(\Sigma):\Sen^{\FOL^{{\relax}}_{\alpha}(\INS)}(\Sigma)\gets\Sen^{\FOL^{{\relax}}_{\alpha}(\INS')}(\Sig^{\INM}(\Sigma))\in\Set$ as follows:
\begin{itemize}
\item $\Sen^{\FOL^{{\relax}}_{\alpha}(\INM)}(\Sigma)(\pos{\phi,atomic}):=\pos{\Sen^{\INM}(\Sigma)(\phi),atomic}$
\item $\Sen^{\FOL^{{\relax}}_{\alpha}(\INM)}(\Sigma)(\neg\phi):=\neg\Sen^{\FOL^{{\relax}}_{\alpha}(\INM)}(\Sigma)(\phi)$
\item $\Sen^{\FOL^{{\relax}}_{\alpha}(\INM)}(\Sigma)(\vee\phi):=\vee\pos{\Sen^{\FOL^{{\relax}}_{\alpha}(\INM)}(\Sigma)(\phi_{i})\mid i\in\Dom(\phi)}$
\item $\Sen^{\FOL^{{\relax}}_{\alpha}(\INM)}(\Sigma)(\Exists{X}\phi):=\Exists{\Sigma_{\Dex^{\INM}}(X)}(\Sen^{\FOL^{{\relax}}_{\alpha}(\INM)}*\Sigma^{\Dex^{\INM}})(X)(\phi)$
\\
where $(\Sen^{\FOL^{{\relax}}_{\alpha}(\INM)}*\Sigma^{\Dex^{\INM}})(\cdot):=(\Sen^{\FOL^{{\relax}}_{\alpha}(\INM)}\circ_{H}(\Sigma^{\Dex^{\INS}}[\Sigma_{\Dex^{\INM}}(\cdot)]))\circ_{V}(\Sen^{\FOL^{{\relax}}_{\alpha}(\INS')}\circ_{H}(\Sigma^{\Dex^{\INM}}[\cdot]))$
\\
i.e. $(\Sen^{\FOL^{{\relax}}_{\alpha}(\INM)}*\Sigma^{\Dex^{\INM}})(X):=\Sen^{\FOL^{{\relax}}_{\alpha}(\INM)}(\Sigma^{\Dex^{\INS}}[\Sigma_{\Dex^{\INM}}(X)])\circ\Sen^{\FOL^{{\relax}}_{\alpha}(\INS')}(\Sigma^{\Dex^{\INM}}[X])$
\begin{center}\small
\def\objectstyle{\scriptstyle}
\def\labelstyle{\scriptscriptstyle}
$\xymatrix@R=28pt@C=42pt{
&&&\ar@{}[d]|{\Downarrow\,(\Sigma^{\Dex^{\INM}}[\cdot])}&&&&&\\
\Set&
\ar[l]^{\Sen^{\FOL^{{\relax}}_{\alpha}(\INS')}}(\Sig^{\INS'})&
\ar[l]^{(\Sig^{\INM})}(\Sig^{\INS})\ar@/^50pt/[ll]^{\Sen^{\FOL^{{\relax}}_{\alpha}(\INS)}}&&
\ar[ll]^{(\Sigma^{\Dex^{\INS}}[\cdot])}(\Sigma_{\Dex^{\INS}})&
\ar[l]^{(\Sigma_{\Dex^{\INM}})}(\Sig^{\INM}(\Sigma)_{\Dex^{\INS'}})\ar@/_50pt/[llll]_{(\Sig^{\INM}(\Sigma)^{\Dex^{\INS'}}[\cdot])} &&&\\
&\ar@{}[u]|{\Downarrow\Sen^{\FOL^{{\relax}}_{\alpha}(\INM)}}&&&&&&&
}$
\end{center}
\end{itemize}
\end{itemize}
\end{itemize}
\end{itemize}
\end{definition}
Later, when using quantification in proof systems, the concept of substitution will be necessary.
For now, we will define substitution simply as signature morphism.
\begin{definition}[Substitutions]
Let $\chi:\Sigma\to\Sigma'\in\Sig^{\INS}$.
Here, $\chi$-substitution simply means a morphism $\theta\in\Sig^{\INS}$ such that $\theta\circ\chi=\id_{\Sigma}$.
\end{definition}
\begin{example}\label{exa:substitution}
From the properties of universal morphisms, for
$\id_{\Sigma^{\Dex^{\INS}}[X]}:\Sigma^{\Dex^{\INS}}(X)\to\Sigma^{\Dex^{\INS}}(X)$ we can get the unique substitution $\theta$ which satisfies the following commutativity
\begin{center}
$\xymatrix@R=50pt@C=100pt{
\Sigma^{\Dex^{\INS}}[X]\ar[r]|{\Sigma^{\Dex^{\INS}}(X)^{\Dex^{\INS}}[X]}\ar[dr]|{\id_{\Sigma^{\Dex^{\INS}}[X]}}&\Sigma^{\Dex^{\INS}}[X]^{\Dex^{\INS}}[X']\ar[d]|{\theta}\\
\Sigma\ar[u]|{\Sigma^{\Dex^{\INS}}(X)}\ar[r]|{\Sigma^{\Dex^{\INS}}(X)}&\Sigma^{\Dex^{\INS}}[X]
}$
\end{center}
where $X':=\Sigma^{\Dex^{\INS}}(X)_{\Dex^{\INS}}(X)$.
Intuitively, $\theta$ maps $X'$ to $X$.
\end{example}

\begin{lemma}\label{lemma:translation_of_substitution}
Using the properties of universal morphisms, we can obtain a first-order "translation" of substitution.
\begin{itemize}
\item Let $\INS\in|\DINS|$. For all $\chi:\Sigma\to\Sigma'\in\Sig^{\INS}$, $X\in|\Sigma_{\Dex^{\INS}}|$, and $\Sigma^{\Dex^{\INS}}(X)$-substitution $\theta:\Sigma^{\Dex^{\INS}}[X]\to\Sigma$,
there is unique $\Sigma'^{\Dex^{\INS}}(\chi_{\Dex^{\INS}}(X))$-substitution $\theta':\Sigma'^{\Dex^{\INS}}[\chi_{\Dex^{\INS}}(X)]\to\Sigma$ which satisfies
the commutativity $\theta'\circ\chi^{\Dex^{\INS}}[X]=\chi\circ\theta$.
\begin{center}
$\xymatrix@R=50pt@C=100pt{
\Sigma^{\Dex^{\INS}}[X]\ar[r]|{\chi^{\Dex^{\INS}}[X]}\ar[d]|{\theta}&\Sigma'^{\Dex^{\INS}}[\chi_{\Dex^{\INS}}(X)]\ar[d]|{\theta'}\\
\Sigma\ar[r]|{\chi}&\Sigma'
}$
\end{center}
\item Let $\INM:\INS\to\INS'\in\DINS$.
For all $\Sigma\in|\Sig^{\INS}|$, $X\in\Sig^{\INM}(\Sigma)_{\Dex^{\INS'}}$, and $\Sig^{\INM}(\Sigma)^{\Dex^{\INS'}}(X)$-substitution $\theta:\Sig^{\INM}(\Sigma)^{\Dex^{\INS'}}[X]\to\Sig^{\INM}(\Sigma)$,
there is unique $\Sigma^{\Dex^{\INS}}(\Sigma_{\Dex^{\INM}}(X))$-substitution $\theta':\Sigma^{\Dex^{\INS}}[\Sigma_{\Dex^{\INM}}(X)]\to\Sigma$ which satisfies
the commutativity $\Sig^{\INM}(\theta')\circ\Sigma^{\Dex^{\INM}}[X]=\theta$.
\begin{center}
$\xymatrix@R=30pt@C=70pt{
\Sig^{\INM}(\Sigma)^{\Dex^{\INS'}}[X]\ar[r]^{\Sigma^{\Dex^{\INM}}[X]}\ar[dd]_{\theta}&\Sig^{\INM}(\Sigma^{\Dex^{\INS}}[\Sigma_{\Dex^{\INM}}(X)])\ar[ddl]|{\Sig^{\INM}(\theta')}&\Sigma^{\Dex^{\INS}}[\Sigma_{\Dex^{\INM}}(X)]\ar[dd]^{\theta'}\ar@{--}[l]\\
&&\ar@{}[l]|{\leftarrow\Sig^{\INM}\text{\rotatebox[origin=c]{180}{$\vdash$}}\Space\Space}\\
\Sig^{\INM}(\Sigma)&&\Sigma\ar@{--}[ll]
}$
\end{center}
\end{itemize}
\end{lemma}
\begin{remark}
Since Substitution can be translated, we can also use it like $\theta\equiv_{X}\theta'$ instead of an equality, for example.
The star $*$ in \cite{go-icalp24} is also a type of expression that includes quantification, as it asserts the existence of a finite path semantically.
A similar expression can be achieved by using a two-variable sentence instead of an action.
The problem is that $*$ can only be applied to binary relations where the domain and codomain coincide, but $\Dex$-institution does not have a "sort" structure.
However, this becomes possible by duplicating the variable.
For example, let $\Sigma\in|\Sig^{\INS_{\FOL_{0}}}|$, $X=\{x\}\,(x=\pos{0,s,\Sigma})$.
$\Sigma^{\Dex}[X]^{\Dex}[X']\,(X'=(\Sigma^{\Dex}(X))_{\Dex}(X))$ has two variables that have the same sort.
Therefore, in the general case as well, a similar expression is possible by using the elements of $\Sen(\Sigma^{\Dex}[X]^{\Dex}[X'])$.
\end{remark}

\section{Proof systems and completeness}\label{sec:Pc}
\begin{definition}[Sequent]
A sequent of $\INS$ is a family of relations $\vdash=\{\vdash_{\Sigma}\subseteq\mathcal{P}(\Sen^{\INS}(\Sigma))\times\mathcal{P}(\Sen^{\INS}(\Sigma))\}_{\Sigma\in|\Sig^{\INS}|}$ that satisfies the following "structural rules".
\par\noindent 
\begin{center}\small
\textbf{The structural rules}\\
\begin{tabular}{lll}
\\
&$\Modify(\chi)~\proofrule{\,\Gamma\vdash_{\Sigma}\Delta\,}{\,\Gamma'\vdash_{\Sigma'}\Delta'\,}$
$\Space\chi:\Sigma\to\Sigma'\in\Sig^{\INS},\,\,\Gamma'\supseteq\Sen^{\INS}(\chi)(\Gamma),\,\,\Delta'\supseteq\Sen^{\INS}(\chi)(\Delta)$\\
&\\
&$\Init_{\psi}~\proofrule{}{\,\{\psi\}\cup\Gamma\vdash_{\Sigma}\Delta\cup\{\psi\}\,}$
$\Space\Cut^{\psi}~\proofrule{\,\Gamma\vdash_{\Sigma}\Delta\cup\{\psi\}\,\Space\,\{\psi\}\cup\Gamma'\vdash_{\Sigma}\Delta'\,}{\,\Gamma\cup\Gamma'\vdash_{\Sigma}\Delta\cup\Delta'\,}$
\\
\end{tabular}
\end{center}
\medskip
To shorten the notation, $\Gamma\cup\Gamma'$ may be abbreviated as $\Gamma\,\Gamma'$, and $\{\phi\}$ as $\phi$.
Also, $\Gamma\,\Gamma\vdash_{\Sigma}\Delta\,\Delta'$ may be written as $\Gamma'\vdash_{\Gamma\mid\Sigma\mid\Delta}\Delta'$.
\end{definition}

\begin{definition}[Semantic sequent]\label{def:Semantic-sequent}
The semantic sequent $\vDash^{\INS}$ is defined as follows:
\[
\Gamma\vDash^{\INS}\Delta\text{ iff}\text{ there is no model }\A\text{ such that }\A\models^{\INS}\gamma\text{ for all }\gamma\in\Gamma\text{ and }\A\not\models^{\INS}\gamma\text{ for all }\gamma\in\Delta
\]
In standard notation, $\Gamma\vDash\Delta$ is often used with meanings like $\wedge\Gamma\Rightarrow\wedge\Delta$, but here, in accordance with the term sequent, it is interpreted as $\wedge\Gamma\Rightarrow\vee\Delta$.
\end{definition}

\begin{definition}[Soundness, Completeness, and Compactness]
Let $\vdash$ be the sequent of $\INS$.
\begin{itemize}
\item
$\vdash$ is sound (complete)
if ${\vdash}\subseteq{\vDash^{\INS}}\,({\vdash}\supseteq{\vDash^{\INS}})$ holds true.
\item Let $\Modify(\INS)$ be the category defined as follows:
\begin{itemize}
\item $T\in|\Modify(\INS)|$ iff $T=\pos{\Gamma,\Sigma,\Delta}$ where $\Gamma,\Delta\subseteq\Sen^{\INS}(\Sigma)$
\item $\chi:\pos{\Gamma,\Sigma,\Delta}\to\pos{\Gamma',\Sigma',\Delta'}\in\Modify(\INS)$ iff $\chi:\Sigma\to\Sigma'\in\Sig^{\INS}$ such that $\Sen^{\INS}(\chi)(\Gamma)\subseteq\Gamma'$ and $\Sen^{\INS}(\chi)(\Delta)\subseteq\Delta'$.
\end{itemize}
$\vdash$ is compact if
each directed diagram $\pos{\pos{\Gamma_{i},\Sigma_{i},\Delta_{i}}\in\Modify(\INS)}_{i\in\pos{J,\leq}}$ with a co-limit $\pos{\Gamma,\Sigma,\Delta}$ preserve consistency i.e. if $\Gamma_{i}\not\vdash_{\Sigma_{i}}\Delta_{i}$ for each $i\in J$, then $\Gamma\not\vdash_{\Sigma}\Delta$ holds true.
\end{itemize}
\end{definition}

\begin{definition}[Syntactic sequent]\label{def:Syntactic-sequent}
We define the sequent $\vdash^{\alpha\,\INS}$ of $\FOL_{\alpha}(\INS)$ as the family of smallest relations satisfying the following:
\par\noindent 
\begin{center}
\textbf{The proof rule for atomic sentences}\\
\begin{tabular}{c}
\\
$\Atom\,\proofrule{}{\,\OL{\Gamma_{b}}\vdash_{\Gamma\mid\Sigma\mid\Delta}\OL{\Delta_{b}}\,}\,(\Gamma_{b}\vdash^{\INS}_{\Sigma}\Delta_{b})$
\end{tabular}
\end{center}
\medskip
\par\noindent 
\begin{center}
\textbf{The proof rules for complex sentences}\\
\begin{tabular}{ll}
\\
$\neg_{L}~\proofrule{\,\vdash_{\Gamma\mid\Sigma\mid\Delta}\phi\,}{\,\neg\phi\vdash_{\Gamma\mid\Sigma\mid\Delta}\,}$
&
$\neg_{R}~\proofrule{\,\phi\vdash_{\Gamma\,\Sigma\mid\Delta}\,}{\,\vdash_{\Gamma\mid\Sigma\mid\Delta}\neg\phi\,}$
\\
\\
${\vee}_{L}~\proofrule{\,\pos{\,\UL{\phi_{n}}\vdash_{\Gamma_{n}\mid\Sigma\mid\Delta_{n}}}_{n\in\Dom(\phi)}\,}{\,\OL{\vee\phi}\vdash_{\cup_{n\in\Dom(\phi)}\Gamma_{n}\mid\Sigma\mid\cup_{n\in\Dom(\phi)}\Delta_{n}}\,}$
&
${\vee}_{R}^{n}~\proofrule{\,\vdash_{\Gamma\mid\Sigma\mid\Delta}\UL{\phi_{n}}\,}{\,\vdash_{\Gamma\mid\Sigma\mid\Delta}\OL{\vee\phi}\,}\,(n\in\Dom(\phi))$
\\
\\
$\exists_{L}~\proofrule{\,\UL{\phi}\vdash_{\Sigma(X)(\Gamma)\mid\Sigma^{\Dex}[X]\mid\Sigma(X)(\Delta)}\,}{\,\OL{\Exists{X}\phi}\vdash_{\Gamma\mid\Sigma\mid\Delta}\,}$
&
$\exists_{R}^{\theta}~\proofrule{\,\vdash_{\Gamma\mid\Sigma\mid\Delta}\UL{\Sen^{\FOL_{\alpha}(\INS)}(\theta)(\phi)}\,}{\,\vdash_{\Gamma\mid\Sigma\mid\Delta}\OL{\Exists{X}\phi}\,}\,(\theta\text{ is a }\Sigma^{\Dex}(X)\text{-substitution})$
\\
\end{tabular}
\end{center}
\medskip
\end{definition}

\begin{lemma}[Soundness]\label{lemma:syn-soundness}
For each $\INS\in\DINS$ and $\Sigma\in|\Sig^{\INS}|$,
$\vdash^{\alpha\,\INS}_{\Sigma}\subseteq\vDash^{\FOL^{{\relax}}_{\alpha}(\INS)}_{\Sigma}$.
\end{lemma}

\begin{definition}
Let $\INS\in|\DINS|$.
\begin{itemize}
\item
A subset $\Gamma\subseteq\Sen^{\INS}(\Sigma)$ is called
epi-basic if $\Gamma$ has a $\Sigma$-model $T_{\Sigma,\Gamma}$ (called epi-basic model), which satisfies:
\begin{itemize}
\item $\A\models^{\INS}\Gamma$ iff there is a homomorphism $h:T_{\Sigma,\Gamma}\to\A\in\Mod^{\INS}(\Sigma)$
\item such homomorphism is unique
\end{itemize}
\item
A model $\A\in|\Mod^{\INS}(\Sigma)|$ is reachable if for each $X\in|\Sigma_{\Dex^{\INS}}|$, and $\Sigma^{\Dex^{\INS}}(X)$-expansion $\hat{\A}$,
there is a $\Sigma^{\Dex^{\INS}}(X)$-substitution $\theta$ such that $\Mod^{\INS}(\theta)(\A)=\hat{\A}$.
\end{itemize}
\end{definition}

\begin{definition}
We define the full subcategory $\DINS_{b}$ as follows:
\begin{itemize}
\item[]
$\INS\in|\DINS_{b}|$ iff
\begin{description}
\item[Colimits]
$\Sig^{\INS}$ is co-complete (for directed diagrams).
$\Sen^{\INS}$, $(\cdot)_{\Dex^{\INS}}$, $Sub^{\INS}$ are cocontinuous (for directed colimit) over $\Sig^{\INS}$,
where, $Sub^{\INS}:\Sig^{\INS}\to\Set$ is a functor determined by the sets of substitutions and the translations dealt with in Lemma~\ref{lemma:translation_of_substitution}.
\item[Atomic sentences]
$\Gamma$ has a reachable epi-basic model $T_{\Sigma,\Gamma}$ for each $\Sigma\in|\Sig^{\INS}|$ and $\Gamma\subseteq\Sen^{\INS}(\Sigma)$.
\item[Compact] $\vDash^{\INS}$ is compact.
\end{description}
\end{itemize}
\end{definition}

\begin{lemma}\label{lemma:basic-compactness}
Let $\INS\in|\DINS_{b}|$, the following hold true:
\begin{enumerate}
\item\label{enu:b-compa1} if there is a homomorphism $h:\A\to\B$, and $\A\models^{\INS}\phi$ then $\B\models^{\INS}\phi$
\item\label{enu:b-compa2} $T_{\Sigma,\Gamma}\models^{\INS}\phi$ iff $\Gamma\vDash^{\INS}_{\Sigma}\phi$, (therefore, $\Gamma\vDash^{\INS}_{\Sigma}\Delta$ iff $\Gamma\vDash^{\INS}_{\Sigma}\delta$ for some $\delta\in\Delta$)
\end{enumerate}
\end{lemma}

\begin{theorem}[Compactness]\label{theorem:syn-compactness}
For each $\INS\in|\DINS_{b}|$, $\vdash^{\omega\,\INS}$ is compact.
\end{theorem}

\begin{theorem}[Completeness]\label{theorem:syn-completeness}
Let $\INS\in|\DINS_{b}|$, and assume $\power_{\omega}(\Sigma)=\power_{\omega}(\Sigma^{\Dex^{\INS}}[X])$ holds true for each $\Sigma\in|\Sig^{\INS}|$ and $X\in|\Sigma_{\Dex^{\INS}}|$ where $\power_{\omega}(\Sigma'):=\card(\Sen^{\FOL_{\omega}(\INS)}(\Sigma'))$.
Then $\vdash^{\omega\,\INS}$ is complete i.e. ${\vdash^{\omega\,\INS}_{\Sigma}}={\vDash^{\FOL_{\omega}(\INS)}_{\Sigma}}$ holds true for each $\Sigma\in|\Sig^{\INS}|$.
\end{theorem}

\section{Conclusion and Development}
In conventional institution theory, variable structures are treated axiomatically, and quantified sentences are handled within internal logic.
While this is a sufficient method of expression for many discussions, it also has the drawback of being ad-hoc and complicating the conditions for applying results.
Furthermore, axiomatization using internal logic requires users to define compound sentences themselves, which is somewhat incompatible with the convenience desired by institution theory.
In this paper, we used a generalization of functor categories to directly introduce variable structures as logical components, and showed that basic results can be derived from natural definitions.

The acquisition of the category of $\Dex$-institutions allows for comparison between $\Dex$-institutions.
By adding the concept of modification to create a two-category system, the definition of equivalence also becomes possible.
For example, it becomes possible to formalize discussions such as comparing first-order logic defined in the usual way with first-order logic defined using fanctorial semantics.
Since the introduction of a compound sentence is done using the functor $\FOL$, comparing the atomic formula and the variable structure will automatically provide the answer for the remaining compound part.
In fanctorial semantics, if the class of categories that constitutes the category of signatures changes, the meaning of extensions by variables (often free constructions) also changes accordingly.
Unless we compare the logic including the variable structure, we cannot effectively use the results for the original logic.
However, since we can compare variable structures, even if we add constraints to the model or new representations to an already defined logic $L$ to create a logic $L'$, we can still derive some results from at least the information $L'\in|\DINS/L|$.

Since there are many discussions related to variable structure, it is thought that using this direct introduction will make many institutional-related methods more concise and easier to use.
Consider substitution as an example.
Usually, a conceptual substitution in an institution is often treated as something that includes information about $\Sen(\theta)$ and $\Mod(\theta)$ that satisfy the satisfaction condition.
To put it more strongly, it is like a signature morphism within an extended institution.
Let $\Sigma=\pos{S,F,P}\in\Sig^{\INS_{\FOL_{0}}}$.
If we define $\T(\Sigma)=\pos{S,F',P}$ by $F'_{{\to}s}=T_{\Sigma\,s}$ (set of terms) and $F'_{\BS{s}\to s}=F_{\BS{s}\to s}$, then the term generation operation is, in a sense, an extension of the signature.
This can be canonically defined as a Kleisli category over $\Sig^{\INS_{\FOL_{0}}}$, and the morphisms of the Kleisli category can be considered as generalized signature morphisms.
That is, the generalized signature morphism $\chi:\Sigma\to\Sigma'\in\Sig^{\INS_{\FOL^{\gene}_{0}}}$ is the morphism $\chi:\Sigma\to\T(\Sigma')\in\Sig^{\INS_{\FOL_{0}}}$.
This naturally includes the concept of substitution.
Here, we also consider extending $\Sen$.
Applying $\Sen^{\INS_{\FOL_{0}}}$ normally yields $\Sen^{\INS_{\FOL_{0}}}(\chi):\Sen(\Sigma)\to\Sen(\T(\Sigma'))$.
So to define $\Sen^{\INS_{\FOL^{\gene}_{0}}}(\chi):\Sen(\Sigma)\to\Sen(\Sigma')$, we need to change the target from $\Sen(\T(\Sigma'))$ to $\Sen(\Sigma')$.
\begin{center}
$
\xymatrix@R=30pt@C=60pt{
&\Sen(\T(\Sigma'))\ar[d]|{\epsilon}\\
\ar[ur]|{\Sen^{\INS_{\FOL_{0}}}(\chi)}\Sen(\Sigma)\ar[r]|{\Sen^{\INS_{\FOL^{\gene}_{0}}}(\chi)}&\Sen(\Sigma')
}
$
\end{center}
This suggests defining a natural transformation $\epsilon:\Sen\circ\T\Rightarrow\Sen$.
Note that this is similar to sentence translation between institutions.
Therefore, these concepts can also be understood more structurally, as in the case of $\Dex$-institutions (for example, through a formulation like "monads on $\Dex$-institutions").

\bibliographystyle{plainurl}
\bibliography{dinsb}

@book{Diaconescu2025,
  author = {R{\u{a}}zvan Diaconescu},
  title = {Institution-Independent Model Theory},
  series = {Studies in Universal Logic},
  publisher = {Birkh{\"{a}}user},
  year = {2025},
  _bib2doi_finished = {true},
}

@inproceedings{go-icalp24,
  author       = {Go Hashimoto and
                  Daniel G\u{a}in\u{a} and
                  Ionut \c{T}u\c{t}u},
  editor       = {Karl Bringmann and
                  Martin Grohe and
                  Gabriele Puppis and
                  Ola Svensson},
  title        = {{Forcing, Transition Algebras, and Calculi}},
  booktitle    = {51st International Colloquium on Automata, Languages, and Programming,
                  {ICALP} 2024, July 8-12, 2024, Tallinn, Estonia},
  series       = {LIPIcs},
  volume       = {297},
  pages        = {143:1--143:17},
  publisher    = {Schloss Dagstuhl - Leibniz-Zentrum f{\"{u}}r Informatik},
  year         = {2024},
  doi          = {10.4230/LIPICS.ICALP.2024.143}
}
\newpage
\appendix

\section{Proof of the results presented in Section~\ref{sec:pre}}

\begin{proof}[Proof of well-definedness of Definition~\ref{def:elements}]
We verify that $\newcoprod_{S}:\Cat^{S^{op}}\to\Cat$ is a functor.
\begin{proofcases}
\item[$\newcoprod_{S}M\in|\Cat|$] Let $M\in|\Cat^{S^{op}}|$.
By the axiom of replacement, $\newcoprod_{S}M$ is a set.
We verify that the operations satisfy the axioms.
\begin{itemize}
\item Let $\pos{h,\chi}:\pos{\A,\Sigma}\to\pos{\A',\Sigma'}\in\newcoprod_{S}M$.
\begin{align*}
&\id_{\pos{\A',\Sigma'}}\circ\pos{h,\chi}=\pos{\id_{\A'},\id_{\Sigma'}}\circ\pos{h,\chi}=\pos{M(\chi)(\id_{\A'})\circ h,\id_{\Sigma'}\circ\chi}=\pos{\id_{M(\chi)(\A')}\circ h,\id_{\Sigma'}\circ\chi}=\pos{h,\chi} \\
&\pos{h,\chi}\circ\id_{\pos{\A,\Sigma}}=\pos{h,\chi}\circ\pos{\id_{\A},\id_{\Sigma}}=\pos{M(\id_{\Sigma})(h)\circ\id_{\A},\chi\circ\id_{\Sigma}}=\pos{h\circ\id_{\A},\chi\circ\id_{\Sigma}}=\pos{h,\chi}
\end{align*}
\item Let $\pos{h,\chi}:\pos{\A,\Sigma}\to\pos{\A',\Sigma'},\,\pos{h',\chi'}:\pos{\A',\Sigma'}\to\pos{\A'',\Sigma''},\,\pos{h'',\chi''}:\pos{\A'',\Sigma''}\to\pos{\A''',\Sigma'''}\in\newcoprod_{S}M$.
\begin{align*}
&(\pos{h'',\chi''}\circ\pos{h',\chi'})\circ\pos{h,\chi}=
\pos{M(\chi')(h'')\circ h',\chi''\circ\chi'}\circ\pos{h,\chi}=
\pos{M(\chi)(M(\chi')(h'')\circ h')\circ h,(\chi''\circ\chi')\circ\chi}= \\
&\pos{M(\chi)(M(\chi')(h''))\circ(M(\chi)(h')\circ h),\chi''\circ(\chi'\circ\chi)}=
\pos{M(\chi'\circ\chi)(h'')\circ(M(\chi)(h')\circ h),\chi''\circ(\chi'\circ\chi)}= \\
&\pos{h'',\chi''}\circ\pos{M(\chi)(h')\circ h,\chi'\circ\chi}=
\pos{h'',\chi''}\circ(\pos{h',\chi'}\circ\pos{h,\chi})
\end{align*}
\end{itemize}
\item[$\newcoprod_{S}\eta:\newcoprod_{S}M\to\newcoprod_{S}M'\in\Cat$] Let $\eta:M\to M'\in\Cat^{S^{op}}$.
\begin{itemize}
\item Let $\pos{\A,\Sigma}\in|\newcoprod_{S}M|$.
\begin{align*}
\newcoprod_{S}\eta(\id_{\pos{\A,\Sigma}})
&={\newcoprod}_{S}\eta(\pos{\id_{\A},\id_{\Sigma}})=\pos{\eta_{\Sigma}(\id_{\A}),\id_{\Sigma}} \\
&=\pos{\id_{\eta_{\Sigma}(\A)},\id_{\Sigma}} \\
&=\id_{\pos{\eta_{\Sigma}(\A),\Sigma}}=\id_{\oldcoprod_{S}\eta(\pos{\A,\Sigma})}
\end{align*}
\item Let $\pos{h,\chi}:\pos{\A,\Sigma}\to\pos{\A',\Sigma'},\,\pos{h',\chi'}:\pos{\A',\Sigma'}\to\pos{\A'',\Sigma''}\in\newcoprod_{S}M$.
\begin{align*}
\newcoprod_{S}\eta(\pos{h',\chi'}\circ\pos{h,\chi})
&=\newcoprod_{S}\eta(\pos{M(\chi)(h')\circ h,\chi'\circ\chi})=\pos{\eta_{\Sigma}(M(\chi)(h')\circ h),\chi'\circ\chi} \\
&=\pos{\eta_{\Sigma}(M(\chi)(h'))\circ\eta_{\Sigma}(h),\chi'\circ\chi}(=\pos{M'(\chi)(\eta_{\Sigma'}(h'))\circ\eta_{\Sigma}(h),\chi'\circ\chi}) \\
&=\pos{\eta_{\Sigma'}(h'),\chi'}\circ\pos{\eta_{\Sigma}(h),\chi}=\newcoprod_{S}\eta(\pos{h',\chi'})\circ\newcoprod_{S}\eta(\pos{h,\chi})
\end{align*}
\end{itemize}
\item[$\newcoprod_{S}:\Cat^{S^{op}}\to\Cat$ is a functor]\
\begin{itemize}
\item Let $M\in|\Cat^{S^{op}}|$. Let $\pos{h,\chi}:\pos{\A,\Sigma}\to\pos{\A',\Sigma'}\in\newcoprod_{S}M$.
\begin{align*}
\newcoprod_{S}\id_{M}(\pos{h,\chi}:\pos{\A,\Sigma}\to\pos{\A',\Sigma'})
&=\pos{(\id_{M})_{\Sigma}(h),\chi}:\pos{(\id_{M})_{\Sigma}(\A),\Sigma}\to\pos{(\id_{M})_{\Sigma'}(\A'),\Sigma'} \\
&=\pos{\id_{M(\Sigma)}(h),\chi}:\pos{\id_{M(\Sigma)}(\A),\Sigma}\to\pos{\id_{M(\Sigma')}(\A'),\Sigma'} \\
&=\pos{h,\chi}:\pos{\A,\Sigma}\to\pos{\A',\Sigma'} \\
&=\id_{\oldcoprod_{S}M}(\pos{h,\chi}:\pos{\A,\Sigma}\to\pos{\A',\Sigma'})
\end{align*}
\item Let $\eta:M\to M',\,\eta':M'\to M''\in\Cat^{S^{op}}$. Let $\pos{h,\chi}:\pos{\A,\Sigma}\to\pos{\A',\Sigma'}\in\newcoprod_{S}M$.
\begin{align*}
\newcoprod_{S}(\eta'\circ_{V}\eta)(\pos{h,\chi}:\pos{\A,\Sigma}\to\pos{\A',\Sigma'})
&=\pos{(\eta'\circ_{V}\eta)_{\Sigma}(h),\chi}:\pos{(\eta'\circ_{V}\eta)_{\Sigma}(\A),\Sigma}\to\pos{(\eta'\circ_{V}\eta)_{\Sigma'}(\A'),\Sigma'} \\
&=\pos{(\eta'_{\Sigma}\circ\eta_{\Sigma})(h),\chi}:\pos{(\eta'_{\Sigma}\circ\eta_{\Sigma})(\A),\Sigma}\to\pos{(\eta'_{\Sigma'}\circ\eta_{\Sigma'})(\A'),\Sigma'} \\
&=\pos{\eta'_{\Sigma}(\eta_{\Sigma}(h)),\chi}:\pos{\eta'_{\Sigma}(\eta_{\Sigma}(\A)),\Sigma}\to\pos{\eta'_{\Sigma'}(\eta_{\Sigma'}(\A')),\Sigma'} \\
&=(\newcoprod_{S}\eta'\circ\newcoprod_{S}\eta)(\pos{h,\chi}:\pos{\A,\Sigma}\to\pos{\A',\Sigma'})
\end{align*}
\end{itemize}
\end{proofcases}
\end{proof}

\begin{proof}[Proof of well-definedness of Definition~\ref{def:elements-nat}]
We verify that $\newcoprod_{L}:\newcoprod_{S}\circ\,\Cat^{L^{op}}\to\newcoprod_{S'}$ is a natural transformation.
\begin{proofcases}
\item[$(\newcoprod_{L}M)(\pos{\A,\Sigma})\in|\newcoprod_{S'}M|,\,(\newcoprod_{L}M)(\pos{h,\chi})\in\newcoprod_{S'}M$] Let $\pos{h,\chi}:\pos{\A,\Sigma}\to\pos{\A',\Sigma'}\in\newcoprod_{S}(M\circ L^{op})$.
\begin{itemize}
\item Since $\pos{\A,\Sigma}\in|\newcoprod_{S}M\circ L^{op}|$ means $\A\in|M(L(\Sigma))|$, $(\newcoprod_{L}M)(\pos{\A,\Sigma})=\pos{\A,L(\Sigma)}\in|\newcoprod_{S'}M|$
\item Since $\pos{h,\chi:\Sigma\to\Sigma'}\in\newcoprod_{S}M\circ L^{op}$ means $h\in M(L(\Sigma))$, $(\newcoprod_{L}M)(\pos{h,\chi})=\pos{h,L(\chi)}\in\newcoprod_{S'}M$
\end{itemize}
\item[$\newcoprod_{L}M:\newcoprod_{S}(M\circ L^{op})\to\newcoprod_{S'}M\in\Cat$] Let $M\in|\Cat^{S'^{op}}|$.
\begin{itemize}
\item Let $\pos{\A,\Sigma}\in|\newcoprod_{S}(M\circ L^{op})|$.
\begin{align*}
&\newcoprod_{L}M(\id_{\pos{\A,\Sigma}})=
\newcoprod_{L}M(\pos{\id_{\A},\id_{\Sigma}})=
\pos{\id_{\A},L(\id_{\Sigma})}=
\pos{\id_{\A},\id_{L(\Sigma)}}=
\id_{\pos{\A,L(\Sigma)}}=
\id_{\oldcoprod_{L}M(\pos{\A,\Sigma})}
\end{align*}
\item Let $\pos{h,\chi}:\pos{\A,\Sigma}\to\pos{\A',\Sigma'},\,\pos{h',\chi'}:\pos{\A',\Sigma'}\to\pos{\A'',\Sigma''}\in\newcoprod_{S}(M\circ L^{op})$.
\begin{align*}
&\newcoprod_{L}M(\pos{h',\chi'}\circ\pos{h,\chi})=
\newcoprod_{L}M(\pos{M(L(\chi))(h')\circ h,\chi'\circ\chi})=
\pos{M(L(\chi))(h')\circ h,L(\chi'\circ\chi)}= \\
&\pos{M(L(\chi))(h')\circ h,L(\chi')\circ L(\chi)}=
\pos{h',L(\chi')}\circ\pos{h,L(\chi)}=
\newcoprod_{L}M(\pos{h',\chi'})\circ\newcoprod_{L}M(\pos{h,\chi})
\end{align*}
\end{itemize}
\item[$\newcoprod_{L}:\newcoprod_{S}\circ\,\Cat^{L^{op}}\to\newcoprod_{S'}\in\Cat^{\Cat^{S'^{op}}}$ is a natural transformation.] Let $\eta:M\to M'\in\Cat^{S'^{op}}$.
We verify that the following is commutative:
\begin{center}\small
$\xymatrix@R=40pt@C=70pt{
\newcoprod_{S}(M\circ L^{op}) \ar[r]|{\newcoprod_{S}(\eta\circ L^{op})} \ar[d]|{\newcoprod_{L}M} & \newcoprod_{S}(M'\circ L^{op}) \ar[d]|{\newcoprod_{L}M'} \\
\newcoprod_{S'}M \ar[r]|{\newcoprod_{S'}\eta} & \newcoprod_{S'}M'
}$
\end{center}
Let $\pos{h,\chi}:\pos{\A,\Sigma}\to\pos{\A',\Sigma'}\in\newcoprod_{S}(M\circ L^{op})$.
\begin{align*}
&(\newcoprod_{S'}\eta)\circ(\newcoprod_{L}M)(\pos{h,\chi})=
(\newcoprod_{S'}\eta)(\pos{h,L(\chi):L(\Sigma)\to L(\Sigma')})=
\pos{\eta_{L(\Sigma)}(h),L(\chi)}= \\
&\pos{(\eta\circ L^{op})_{\Sigma}(h),L(\chi)}=
(\newcoprod_{L}M')(\pos{(\eta\circ L^{op})_{\Sigma}(h),\chi})=
(\newcoprod_{L}M')\circ(\newcoprod_{S}(\eta\circ L^{op}))(\pos{h,\chi})
\end{align*}
\end{proofcases}
\end{proof}

\begin{proof}[Proof of Lemma~\ref{lemma:coprod_is_a_list}]\
\begin{proofcases}
\item[$\newcoprod_{\id_{S}}=\id_{\oldcoprod_{S}}$] Let $S\in|\Cat|$.
\begin{itemize}
\item Checking the domain and the codomain.
\begin{align*}
&\newcoprod_{\id_{S}}:\newcoprod_{S}\circ\,\Cat^{\id_{S}^{op}}\to\newcoprod_{S}\,\,(:\newcoprod_{S}\to\newcoprod_{S}) \\
&\id_{\oldcoprod_{S}}:\newcoprod_{S}\to\newcoprod_{S}
\end{align*}
\item Let $M\in|\Cat^{S^{op}}|$ and $\pos{h,\chi}:\pos{\A,\Sigma}\to\pos{\A',\Sigma'}\in\newcoprod_{S}M$.
\begin{align*}
&\newcoprod_{\id_{S}}M(\pos{h,\chi}:\pos{\A,\Sigma}\to\pos{\A',\Sigma'})
=\pos{h,\id_{S}(\chi)}:\pos{\A,\id_{S}(\Sigma)}\to\pos{\A',\id_{S}(\Sigma')}
=\pos{h,\chi}:\pos{\A,\Sigma}\to\pos{\A',\Sigma'} \\
&\id_{\oldcoprod_{S}}M(\pos{h,\chi}:\pos{\A,\Sigma}\to\pos{\A',\Sigma'})
=\id_{\oldcoprod_{S}M}(\pos{h,\chi}:\pos{\A,\Sigma}\to\pos{\A',\Sigma'})
=\pos{h,\chi}:\pos{\A,\Sigma}\to\pos{\A',\Sigma'}
\end{align*}
\end{itemize}
\item[$\newcoprod_{L\circ^{op}L'}=\Cat^{\Cat^{L'^{op}}}(\newcoprod_{L})\circ_{V}^{op}\newcoprod_{L'}$] Let $L:S\to S',\,L':S'\to S''\in\Cat$.
\begin{itemize}
\item Checking the domain and the codomain.
\begin{align*}
&\newcoprod_{L'\circ L}:\newcoprod_{S}\circ\,\Cat^{(L'\circ L)^{op}}\to\newcoprod_{S''} \\
&\newcoprod_{L'}\circ_{V}\,\Cat^{\Cat^{L'^{op}}}(\newcoprod_{L}):\newcoprod_{S}\circ\,\Cat^{(L'\circ L)^{op}}\to\newcoprod_{S''} \text{ since }\\
&-\,\newcoprod_{L'}:\newcoprod_{S'}\circ\,\Cat^{L'^{op}}\to\newcoprod_{S''} \\
&-\,\Cat^{\Cat^{L'^{op}}}(\newcoprod_{L}):\Cat^{\Cat^{L'^{op}}}(\newcoprod_{S}\circ\,\Cat^{L^{op}})\to\Cat^{\Cat^{L'^{op}}}(\newcoprod_{S'})\,\,(:\newcoprod_{S}\circ\,\Cat^{(L'\circ L)^{op}}\to\newcoprod_{S'}\circ\,\Cat^{L'^{op}})
\end{align*}
\item Let $M\in|\Cat^{S''^{op}}|$ and $\pos{h,\chi}:\pos{\A,\Sigma}\to\pos{\A',\Sigma'}\in\newcoprod_{S}M\circ L'^{op}\circ L^{op}$.
\begin{align*}
&\newcoprod_{L'\circ L}M(\pos{h,\chi}:\pos{\A,\Sigma}\to\pos{\A',\Sigma'})
=\pos{h,(L'\circ L)(\chi)}:\pos{\A,(L'\circ L)(\Sigma)}\to\pos{\A',(L'\circ L)(\Sigma')} \\
&(\newcoprod_{L'}\circ_{V}\,\Cat^{\Cat^{L'^{op}}}(\newcoprod_{L}))(M)(\pos{h,\chi}:\pos{\A,\Sigma}\to\pos{\A',\Sigma'})
=(\newcoprod_{L'}M\circ\,(\newcoprod_{L}M\circ L'^{op}))(\pos{h,\chi}:\pos{\A,\Sigma}\to\pos{\A',\Sigma'}) \\
&=\newcoprod_{L'}M(\pos{h,L(\chi)}:\pos{\A,L(\Sigma)}\to\pos{\A',L(\Sigma')})
=\pos{h,(L'\circ L)(\chi)}:\pos{\A,(L'\circ L)(\Sigma)}\to\pos{\A',(L'\circ L)(\Sigma')}
\end{align*}
\end{itemize}
\end{proofcases}
\end{proof}

\begin{proof}[Proof of well-definedness of Definition~\ref{def:lists}]
We verify that $\newprod:\Cat^{S^{op}}\to\Cat$ is a functor.
\begin{proofcases}
\item[$\newprod_{S}M\in|\Cat|$] Let $M\in|\Cat^{S^{op}}|$.
Since $S$ and $\newcoprod_{S}M$ are sets, $\newprod_{S}M\subseteq(\newcoprod_{S}M)^{S}$ is also a set.
To verify that it is a subcategory, it is sufficient to check whether the operations are closed.
\begin{itemize}
\item Let $c\in|\newprod_{S}M|$. $\id_{\pi}\circ_{H}\id_{c}=\id_{\pi\circ c}=\id_{\id_{S}}$.
\item Let $h:c\to c,\,h':c'\to c''\in\newprod M$. $\id_{\pi}\circ_{H}(h'\circ_{V}h)=(\id_{\pi}\circ_{H}h')\circ_{V}(\id_{\pi}\circ_{H}h)=\id_{\id_{S}}\circ_{V}\id_{\id_{S}}=\id_{\id_{S}}$.
\end{itemize}
\item[$\newprod_{S}\eta:\newprod_{S}M\to\newprod_{S}M'\in\Cat$] Let $\eta:M\to M'\in\Cat^{S^{op}}$.
Since $\newcoprod_{S}\eta:\newcoprod_{S}M\to\newcoprod_{S}M'$ is a functor, $(\newcoprod_{S}\eta\circ\cdot):(\newcoprod_{S}M)^{S}\to(\newcoprod_{S}M')^{S}$ is also a functor.
$\newprod_{S}M\subseteq(\newcoprod_{S}M)^{S}$ and $\newprod_{S}M'\subseteq(\newcoprod_{S}M')^{S}$ are subcategories, and
$\newprod_{S}\eta$ is $(\newcoprod_{S}\eta\circ\cdot)$ restricted to $\newprod_{S}M$,
so it is sufficient to verify $(\newcoprod_{S}\eta\circ\cdot)(\newprod_{S}M)\subseteq\newprod_{S}M'$.
By definition, $\newcoprod_{S}\eta$ does not change the right component, so for $c\in|\newprod_{S}M|$,
we can say that $\pi\circ(\newcoprod_{S}\eta\circ c)=\pi\circ c=\id_{S}$, and therefore $\newprod_{S}\eta(c)=\newcoprod_{S}\eta\circ c\in|\newprod_{S}M'|$.
The same is true for $h:c\to c\in\newprod_{S}M$.
\item[$\newprod_{S}:\Cat^{S^{op}}\to\Cat$ is a functor]\
\begin{itemize}
\item Let $M\in|\Cat^{S^{op}}|$. Let $h:c\to c'\in\newprod_{S}M$.
\begin{align*}
\newprod_{S}\id_{M}(h:c\to c')
&=\newcoprod_{S}\id_{M}\circ_{H}h:\newcoprod_{S}\id_{M}\circ c\to\newcoprod_{S}\id_{M}\circ c' \\
&=\id_{\oldcoprod_{S}M}\circ_{H}h:\id_{\oldcoprod_{S}M}\circ c\to\id_{\oldcoprod_{S}M}\circ c' \\
&=h:c\to c' \\
&=\id_{\oldprod_{S}M}(h:c\to c')
\end{align*}
\item Let $\eta:M\to M',\,\eta':M'\to M''\in\Cat^{S^{op}}$. Let $h:c\to c'\in\newprod_{S}M$.
\begin{align*}
\newprod_{S}(\eta'\circ_{V}\eta)(h:c\to c')
&=\newcoprod_{S}(\eta'\circ_{V}\eta)\circ_{H}h:\newcoprod_{S}(\eta'\circ_{V}\eta)\circ c\to\newcoprod_{S}(\eta'\circ_{V}\eta)\circ c' \\
&=(\newcoprod_{S}\eta'\circ\newcoprod_{S}\eta)\circ_{H}h:(\newcoprod_{S}\eta'\circ\newcoprod_{S}\eta)\circ c\to(\newcoprod_{S}\eta'\circ\newcoprod_{S}\eta)\circ c' \\
&=\newcoprod_{S}\eta'\circ_{H}(\newcoprod_{S}\eta\circ_{H}h):\newcoprod_{S}\eta'\circ(\newcoprod_{S}\eta\circ c)\to\newcoprod_{S}\eta'\circ(\newcoprod_{S}\eta\circ c') \\
&=(\newprod_{S}\eta'\circ\newprod_{S}\eta)(h:c\to c')
\end{align*}
\end{itemize}
\end{proofcases}
\end{proof}

\begin{proof}[Proof of well-definedness of Definition~\ref{def:lists-nat}]
We verify that $\newprod_{L}:\newprod_{S'}\to\newprod_{S}\circ\,\Cat^{L^{op}}$ is a natural transformation.
\begin{proofcases}
\item[$(\newprod_{L}M)(c)\in|\newprod_{S}(M\circ L^{op})|,\,(\newprod_{L}M)(h)\in\newprod_{S}(M\circ L^{op})$] Let $h:c\to c'\in\newprod_{S'}M$.
\begin{itemize}
\item We show $(\newprod_{L}M)(c)\in|\newprod_{S}(M\circ L^{op})|$.
\begin{itemize}
\item Let $\Sigma\in|S|$.
\begin{align*}
&(\newprod_{L}M)(c)(\id_{\Sigma})=\{\pos{L(\chi)^{c},\chi}\}_{\chi\in S}(\id_{\Sigma})=\pos{L(\id_{\Sigma})^{c},\id_{\Sigma}}=\pos{\id_{L(\Sigma)}^{c},\id_{\Sigma}} \\
&=\pos{\id_{L(\Sigma)^{c}},\id_{\Sigma}}=\id_{\pos{L(\Sigma)^{c},\Sigma}}=\id_{(\oldprod_{L}M)(c)(\Sigma)}
\end{align*}
\item Let $\chi:\Sigma\to\Sigma',\,\chi':\Sigma'\to\Sigma''\in S$.
\begin{align*}
&(\newprod_{L}M)(c)(\chi'\circ\chi)=\{\pos{L(\chi)^{c},\chi}\}_{\chi\in S}(\chi'\circ\chi)=\pos{L(\chi'\circ\chi)^{c},\chi'\circ\chi}=\pos{(L(\chi')\circ L(\chi))^{c},\chi'\circ\chi} \\
&=\pos{L(\chi')^{c}\circ L(\chi)^{c},\chi'\circ\chi}=\pos{L(\chi')^{c},\chi'}\circ\pos{L(\chi)^{c},\chi}=(\newprod_{L}M)(c)(\chi')\circ(\newprod_{L}M)(c)(\chi)
\end{align*}
\end{itemize}
$\pi\circ(\newprod_{L}M)(c)=\pi\circ\{\pos{L(\chi)^{c},\chi}\}_{\chi\in S}=\{\chi\}_{\chi\in S}=\id_{S}$
Therefore, $(\newprod_{L}M)(c)\in|\newprod_{S}(M\circ L^{op})|$.
\item We show $(\newprod_{L}M)(h)\in\newprod_{S}(M\circ L^{op})$.
First, we check the following commutativity.
\begin{center}\small
$\xymatrix@R=40pt@C=80pt{
(\newprod_{L}M)(c)(\Sigma) \ar[r]|{(\newprod_{L}M)(c)(\chi)} \ar[d]|{(\newprod_{L}M)(h)(\Sigma)} & (\newprod_{L}M)(c)(\Sigma') \ar[d]|{(\newprod_{L}M)(h)(\Sigma')} \\
(\newprod_{L}M)(c')(\Sigma) \ar[r]|{(\newprod_{L}M)(c')(\chi)} & (\newprod_{L}M)(c')(\Sigma')
}$
\end{center}
\begin{align*}
&(\newprod_{L}M)(c')(\chi)\circ(\newprod_{L}M)(h)(\Sigma)=\pos{L(\chi)^{c'},\chi}\circ\pos{L(\Sigma)^{h},\id_{\Sigma}}=\pos{L(\chi)^{c'}\circ L(\Sigma)^{h},\chi} \\
&(\newprod_{L}M)(h)(\Sigma')\circ(\newprod_{L}M)(c)(\chi)=\pos{L(\Sigma')^{h},\id_{\Sigma}}\circ\pos{L(\chi)^{c},\chi}
=\pos{M(L(\chi))(L(\Sigma')^{h})\circ L(\chi)^{c},\chi}
\end{align*}
Now that we know that the right components are equal, we will slightly transform the right components in order to compare the left components.
\begin{align*}
&\pos{L(\chi)^{c'}\circ L(\Sigma)^{h},L(\chi)}=\pos{L(\chi)^{c'},L(\chi)}\circ\pos{L(\Sigma)^{h},\id_{L(\Sigma)}}=c'(\chi)\circ h(\Sigma) \\
&\pos{M(L(\chi))(L(\Sigma')^{h})\circ L(\chi)^{c},L(\chi)}=\pos{L(\Sigma')^{h},\id_{L(\Sigma')}}\circ\pos{L(\chi)^{c},L(\chi)}=h(\Sigma')\circ c(\chi)
\end{align*}
Since, $h:c\Rightarrow c',\,\,c'(\chi)\circ h(\Sigma)=h(\Sigma')\circ c(\chi)$. Therefore, commutativity holds true. \\
$\id_{\pi}\circ(\newprod_{L})(h)=\pi\circ\{\pos{L(\Sigma)^{h},\id_{\Sigma}}\}_{\Sigma\in|S|}=\{\id_{\Sigma}\}_{\Sigma\in|S|}=\id_{\id_{S}}$
Therefore, $(\newprod_{L}M)(h)\in\newprod_{S}(M\circ L^{op})$.
\end{itemize}
\item[$\newprod_{L}M:\newprod_{S'}M\to\newprod_{S}(M\circ L^{op})\in\Cat$] Let $M\in|\Cat^{S'^{op}}|$.
\begin{itemize}
\item Let $c\in|\newprod_{S'}M|$.
\begin{align*}
(\newprod_{L}M)(\id_{c})
=&\{\pos{L(\Sigma)^{\id_{c}},\id_{\Sigma}}\}_{\Sigma\in|S|}:\{\pos{L(\chi)^{c},\chi}\}_{\chi\in S}\Rightarrow\{\pos{L(\chi)^{c},\chi}\}_{\chi\in S} \\
=&\{\pos{\id_{L(\Sigma)^{c}},\id_{\Sigma}}\}_{\Sigma\in|S|}:\{\pos{L(\chi)^{c},\chi}\}_{\chi\in S}\Rightarrow\{\pos{L(\chi)^{c},\chi}\}_{\chi\in S} \\
=&\id_{\{\pos{L(\chi)^{c},\chi}\}_{\chi\in S}}=\id_{(\oldprod_{L}M)(c)}
\end{align*}
\item Let $h:c\to c',\,h':c'\to c''\in\newcoprod_{S'}M$.
\begin{align*}
(\newprod_{L}M)(h'\circ h)
=&\{\pos{L(\Sigma)^{h'\circ h},\id_{\Sigma}}\}_{\Sigma\in|S|}:\{\pos{L(\chi)^{c},\chi}\}_{\chi\in S}\Rightarrow\{\pos{L(\chi)^{c''},\chi}\}_{\chi\in S} \\
=&\{\pos{L(\Sigma)^{h'}\circ L(\Sigma)^{h},\id_{\Sigma}}\}_{\Sigma\in|S|}:\{\pos{L(\chi)^{c},\chi}\}_{\chi\in S}\Rightarrow\{\pos{L(\chi)^{c''},\chi}\}_{\chi\in S} \\
=&\{\pos{L(\Sigma)^{h'},\id_{\Sigma}}\}_{\Sigma\in|S|}\circ_{V}\{\pos{L(\Sigma)^{h},\id_{\Sigma}}\}_{\Sigma\in|S|}
=(\newprod_{L}M)(h')\circ(\newprod_{L}M)(h)
\end{align*}
\end{itemize}
\item[$\newprod_{L}:\newprod_{S'}\to\newprod_{S}\circ\,\Cat^{L^{op}}\in\Cat^{(\Cat^{S'^{op}})}$ is a natural transformation.] Let $\eta:M\to M'\in\Cat^{S'^{op}}$.
We verify that the following is commutative:
\begin{center}\small
$\xymatrix@R=40pt@C=70pt{
\newprod_{S'}M \ar[r]|{\newprod_{S'}\eta} \ar[d]|{\newprod_{L}M} & \newprod_{S'}M' \ar[d]|{\newprod_{L}M'} \\
\newprod_{S}(M\circ L^{op}) \ar[r]|{\newprod_{S}(\eta\circ L^{op})} & \newprod_{S}(M'\circ L^{op})
}$
\end{center}
Let $h:c\to c'\in\newprod_{S'}M$.
\begin{align*}
&(\newprod_{S}(\eta\circ L^{op}))\circ(\newcoprod_{L}M)(h)=(\newprod_{S}(\eta\circ L^{op}))(\{\pos{L(\Sigma)^{h},\id_{\Sigma}}\}_{\Sigma\in|S|}:\{\pos{L(\chi)^{c},\chi}\}_{\chi\in S}\to\{\pos{L(\chi)^{c'},\chi}\}_{\chi\in S}) \\
&=(\newcoprod_{S}(\eta\circ L^{op}))\circ_{H}\{\pos{L(\Sigma)^{h},\id_{\Sigma}}\}_{\Sigma\in|S|}=\{\pos{(\eta\circ L^{op})_{\Sigma}(L(\Sigma)^{h}),\id_{\Sigma}}\}_{\Sigma\in|S|}=\{\pos{\eta_{L(\Sigma)}(L(\Sigma)^{h}),\id_{\Sigma}}\}_{\Sigma\in|S|} \\ 
& \\
&(\newprod_{L}M')\circ(\newprod_{S'}\eta)(h)=(\newprod_{L}M')((\newcoprod_{S'}\eta)\circ_{H}h)=\{\pos{L(\Sigma)^{(\oldcoprod_{S'}\eta)\,\circ_{H}h},\id_{\Sigma}}\}_{\Sigma\in|S|}=\{\pos{L(\Sigma)^{\{\oldcoprod_{S'}\eta(h_{\Sigma'})\}_{\Sigma'\in|S|}},\id_{\Sigma}}\}_{\Sigma\in|S|} \\
&=\{\pos{L(\Sigma)^{\{\oldcoprod_{S'}\eta(\pos{\Sigma'^{h},\id_{\Sigma'}})\}_{\Sigma'\in|S|}},\id_{\Sigma}}\}_{\Sigma\in|S|}=\{\pos{L(\Sigma)^{\{\pos{\eta_{\Sigma'}(\Sigma'^{h}),\id_{\Sigma'}}\}_{\Sigma'\in|S|}},\id_{\Sigma}}\}_{\Sigma\in|S|}=\{\pos{\eta_{L(\Sigma)}(L(\Sigma)^{h}),\id_{\Sigma}}\}_{\Sigma\in|S|}
\end{align*}
Therefore, it is commutative.
\end{proofcases}
\end{proof}

\begin{proof}[Proof of Lemma~\ref{lemma:prod_is_a_list}]\
\begin{proofcases}
\item[$\newprod_{\id_{S}}=\id_{\oldprod_{S}}$] Let $S\in|\Cat|$.
\begin{itemize}
\item Checking the domain and the codomain.
\begin{align*}
&\newprod_{\id_{S}}:\newprod_{S}\to\newprod_{S}\circ\,\Cat^{\id_{S}^{op}}\,\,(:\newprod_{S}\to\newprod_{S}) \\
&\id_{\oldprod_{S}}:\newprod_{S}\to\newprod_{S}
\end{align*}
\item Let $M\in|\Cat^{S^{op}}|$ and $h:c\to c'\in\newprod_{S}M$.
\begin{align*}
&\newprod_{\id_{S}}M(h:c\to c')
=\{\pos{\id_{S}(\Sigma)^{h},\id_{\Sigma}}\}_{\Sigma\in|S|}:\{\pos{\id_{S}(\chi)^{c},\chi}\}_{\chi\in S}\to\{\pos{\id_{S}(\chi)^{c'},\chi}\}_{\chi\in S} \\
&=\{\pos{\Sigma^{h},\id_{\Sigma}}\}_{\Sigma\in|S|}:\{\pos{\chi^{c},\chi}\}_{\chi\in S}\to\{\pos{\chi^{c'},\chi}\}_{\chi\in S}
=h:c\to c' \\
&\id_{\oldprod_{S}}M(h:c\to c')
=\id_{\oldprod_{S}M}(h:c\to c')
=h:c\to c'
\end{align*}
\end{itemize}
\item[$\newprod_{L\circ^{op}L'}=\Cat^{\Cat^{L'^{op}}}(\newprod_{L})\circ_{V}\newprod_{L'}$] Let $L:S\to S',\,L':S'\to S''\in\Cat$.
\begin{itemize}
\item Checking the domain and the codomain.
\begin{align*}
&\newprod_{L'\circ L}:\newprod_{S''}\to\newprod_{S}\circ\,\Cat^{(L'\circ L)^{op}} \\
&\Cat^{\Cat^{L'^{op}}}(\newprod_{L})\circ\newprod_{L'}:\newprod_{S''}\to\newprod_{S}\circ\,\Cat^{(L'\circ L)^{op}} \text{ since }\\
&-\,\newprod_{L'}:\newprod_{S''}\to\newprod_{S'}\circ\,\Cat^{L'^{op}} \\
&-\,\Cat^{\Cat^{L'^{op}}}(\newprod_{L}):\Cat^{\Cat^{L'^{op}}}(\newprod_{S'})\to\Cat^{\Cat^{L'^{op}}}(\newprod_{S}\circ\,\Cat^{L^{op}})\,\,(:\newprod_{S'}\circ\,\Cat^{L'^{op}}\to\newprod_{S}\circ\,\Cat^{(L'\circ L)^{op}})
\end{align*}
\item Let $M\in|\Cat^{S''^{op}}|$ and $h:c\to c'\in\newprod_{S''}M$.
\begin{align*}
&\newprod_{L'\circ L}M(h:c\to c')
=\{\pos{(L'\circ L)(\Sigma)^{h},\id_{\Sigma}}\}_{\Sigma\in|S|}:\{\pos{(L'\circ L)(\chi)^{c},\chi}\}_{\chi\in S}\to\{\pos{(L'\circ L)(\chi)^{c'},\chi}\}_{\chi\in S} \\
&(\Cat^{\Cat^{L'^{op}}}(\newprod_{L})\circ_{V}\newprod_{L'})(M)(h:c\to c')
=((\newprod_{L}M\circ L'^{op})\circ\newprod_{L'}M)(h:c\to c') \\
&=(\newprod_{L}M\circ L'^{op})(\{\pos{L'(\Sigma)^{h},\id_{\Sigma}}\}_{\Sigma\in|S|}:\{\pos{L'(\chi)^{c},\chi}\}_{\chi\in S}\to\{\pos{L'(\chi)^{c'},\chi}\}_{\chi\in S}) \\
&=\{\pos{(L'\circ L)(\Sigma)^{h},\id_{\Sigma}}\}_{\Sigma\in|S|}:\{\pos{(L'\circ L)(\chi)^{c},\chi}\}_{\chi\in S}\to\{\pos{(L'\circ L)(\chi)^{c'},\chi}\}_{\chi\in S}
\end{align*}
\end{itemize}
\end{proofcases}
\end{proof}

\section{Proof of the results presented in Section~\ref{sec:Dq}}

\begin{proof}[Proof of well-definedness of Definition~\ref{def:DINS}]
We will investigate well-definedness.
(Incidentally, the last condition in the definition of $\Dex$-institution and $\Dex$-institution morphism assumes uniqueness of the model extension, but as can be seen from the proof below, they form a category even if we only assume existence.)
{\par}
What needs to be verified is whether the objects constructed by $\id$ and $\circ$ properly form the $\Dex$-institution morphisms, and whether they satisfy the category axioms such as associativity.
The parts other than $\Dex$ are functors and natural transformations determined from the diagram. For example, $\Set^{\INM'\circ\INM}$ is the composition of the following two natural transformations.
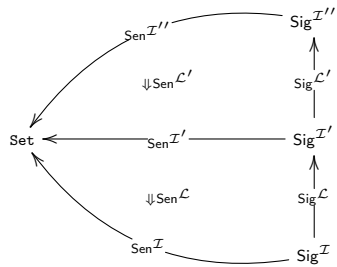
\begin{center}\small
\def\objectstyle{\scriptstyle}
\def\labelstyle{\scriptscriptstyle}
$\xymatrix@R=30pt@C=90pt{
&\Sig^{\INS''} \ar@/_20pt/[ld]|{\Sen^{\INS''}} \\
\Set \ar@{}[ru]|{\Downarrow\Sen^{\INM'}} \ar@{}[rd]|{\Downarrow\Sen^{\INM}} &\Sig^{\INS'} \ar[l]|{\Sen^{\INS'}} \ar[u]|{\Sig^{\INM'}} \\
&\Sig^{\INS} \ar@/^20pt/[lu]|{\Sen^{\INS}} \ar[u]|{\Sig^{\INM}}
}$
\end{center}
Since they cannot be composed directly, we compose them by matching their domains using “$\circ_{H}\Sig^{\INM}$”.
$\Cat$ is a strict 2-category, and generally, such computations of natural transformations are fixed regardless of the order of interpretation of the diagram.
Therefore, it is clear that these definitions satisfy axioms such as associativity.
Thus, here we only verify the satisfaction condition and the part concerning $\Dex^{\INM'\circ\INM}$.
\begin{proofcases}
\item[confirmation of institution morphisms]\
\begin{proofcases}
\item[$\id_{\INS}:\INS\to\INS\in\DINS$] Let $\INS\in|\DINS|$.
\begin{proofcases}
\item[$\A\models^{\INS}\Sen^{\id_{\INS}}(\Sigma)(\gamma')$ iff $\Mod^{\id_{\INS}}(\Sigma)(\A)\models^{\INS}\gamma'$]
$\A\models^{\INS}\Sen^{\id_{\INS}}(\Sigma)(\gamma')\,(=\gamma')$ iff $(\A=)\,\Mod^{\id_{\INS}}(\Sigma)(\A)\models^{\INS}\gamma'$
\item[$\Dex^{\id_{\INS}}(\cdot):\Dex^{\cod(\id_{\INS})}_{\id_{\INS}}(\cdot)\Leftarrow\Dex^{\dom(\id_{\INS})}_{\id_{\INS}}(\cdot)\in\newprod_{\cdot\in\Sig^{\INS}}(\Sig^{\id_{\INS}}(\cdot)/\Sig^{\INS})^{\Sig^{\id_{\INS}}(\cdot)_{\Dex^{\INS}}}$]
\
\begin{itemize}
\item $\Dex^{\cod(\id_{\INS})}_{\id_{\INS}}:=\newprod_{\Sig^{\INS}}((\cdot/\Sig^{\id_{\INS}})^{(\cdot)_{\Dex^{\id_{\INS}}}})(\Dex^{\INS})
=\newprod_{\Sig^{\INS}}((\cdot/\id_{\Sig^{\INS}})^{\id_{(\cdot)_{\Dex^{\INS}}}})(\Dex^{\INS})$
\item[] $=\newprod_{\Sig^{\INS}}(\id_{(\cdot/\Sig^{\INS})^{(\cdot)_{\Dex^{\INS}}}})(\Dex^{\INS})
=\id_{\oldprod_{\Sig^{\INS}}(\cdot/\Sig^{\INS})^{(\cdot)_{\Dex^{\INS}}}}(\Dex^{\INS})=\Dex^{\INS}$
\item $\Dex^{\dom(\id_{\INS})}_{\id_{\INS}}:=\newprod_{\Sig^{\id_{\INS}}}((\cdot/\Sig^{\INS})^{(\cdot)_{\Dex^{\INS}}})(\Dex^{\INS})
=\newprod_{\id_{\Sig^{\INS}}}((\cdot/\Sig^{\INS})^{(\cdot)_{\Dex^{\INS}}})(\Dex^{\INS})$
\item[] $=\id_{\oldprod_{\Sig^{\INS}}(\cdot/\Sig^{\INS})^{(\cdot)_{\Dex^{\INS}}}}(\Dex^{\INS})=\Dex^{\INS}$
\end{itemize}
$\Dex^{\id_{\INS}}(\cdot)=\id_{\Dex^{\INS}}:\Dex^{\cod(\id_{\INS})}_{\id_{\INS}}(\cdot)\Leftarrow\Dex^{\dom(\id_{\INS})}_{\id_{\INS}}(\cdot)
\in\newprod_{\Sig^{\INS}}(\cdot/\Sig^{\INS})^{(\cdot)_{\Dex^{\INS}}}=
\newprod_{\cdot\in\Sig^{\INS}}(\Sig^{\id_{\INS}}(\cdot)/\Sig^{\INS})^{\Sig^{\id_{\INS}}(\cdot)_{\Dex^{\INS}}}$
\item[universal morphism]
Let $\Sigma\in|\Sig^{\INS}|$, $X\in|\Sigma_{\Dex^{\INS}}|$.
For any $\chi\in|\Sigma/\Sig^{\INS}|$ and $\theta:\Sig^{\id_{\INS}}(\Sigma)^{\Dex^{\INS}}(X)\to(\Sigma/\Sig^{\id_{\INS}})(\chi)$,
we need to show that there exists a unique morphism $\theta'\in|\Sigma/\Sig^{\INS}|$
such that $(\Sigma/\Sig^{\id_{\INS}})(\theta')\circ\Sigma^{\Dex^{\INM}}(X)=\theta$.
However, by definition, $(\Sigma/\Sig^{\id_{\INS}})$ and $\Sigma^{\Dex^{\INM}}(X)$ are "$\id$",
so $\theta':=\theta$ is a unique morphism that satisfies the condition.
\item[translatability of expanded models]
Let take $\Sigma\in|\Sig^{\INS}|$, $X\in|\Sig^{\id_{\INS}}(\Sigma)_{\Dex^{\INS}}|$, $\A\in|\Mod^{\INS}(\Sigma)|$, and
$\Sig^{\id_{\INS}}(\Sigma)^{\Dex^{\INS}}(X)$-expansion $\hat{\A}'$ of $\A':=\Mod^{\id_{\INS}}(\Sigma)(\A)$.
\begin{itemize}
\item
Since $(\Sigma^{\Dex^{\id_{\INS}}}*\Mod^{\id_{\INS}})(X)=(\id_{\Mod^{\INS}}\circ_{H}(\Sigma^{\Dex^{\id_{\INS}}}[X])^{op})\circ_{V}(\Mod^{\id_{\INS}}\circ_{H}\id_{(\Sigma^{\Dex^{\INS}}[\Sigma_{\Dex^{\id_{\INS}}}(X)])^{op}})$\\
$=(\id_{\Mod^{\INS}}\circ_{H}\id_{(\Sig^{\id_{\INS}})^{op}\circ(\Sigma^{\Dex^{\INS}}[\Sigma_{\Dex^{\id_{\INS}}}(X)])^{op}}\circ_{V}((\id_{\Mod^{\INS}\circ(\Sig^{\id_{\INS}})^{op}}\circ_{H}\id_{(\Sigma^{\Dex^{\INS}}[\Sigma_{\Dex^{\id_{\INS}}}(X)])^{op}})$
$=\id_{\Mod^{\INS}(\Sigma^{\Dex^{\INS}}[X])}$,
\item
there is (unique) $\Sigma^{\Dex^{\INS}}(\Sigma_{\Dex^{\id_{\INS}}}(X))$-expansion $\hat{\A}:=\hat{\A}'$ of $\A$ such that
$(\Sigma^{\Dex^{\id_{\INS}}}*\Mod^{\id_{\INS}})(X)(\hat{\A})\,(=\id_{\Mod^{\INS}(\Sigma^{\Dex^{\INS}}[X])}(\hat{\A})=\hat{\A})\,=\hat{\A}'$.
\end{itemize}
\end{proofcases}
\item[$\INM'\circ\INM:\INS\to\INS''\in\DINS$] Let $\INM:\INS\to\INS',\,\INM':\INS'\to\INS''\in\DINS$.
\begin{proofcases}
\item[$\A\models^{\INS}\Sen^{\INM'\circ\INM}(\Sigma)(\gamma')$ iff $\Mod^{\INM'\circ\INM}(\Sigma)(\A)\models^{\INS''}\gamma'$]
$\A\models^{\INS}\Sen^{\INM'\circ\INM}(\Sigma)(\gamma')$
iff $\A\models^{\INS}\Sen^{\INM}(\Sigma)\circ\Sen^{\INM'}(\Sig^{\INM}(\Sigma))(\gamma')$
iff $\Mod^{\INM'}(\Sig^{\INM}(\Sigma))\circ\Mod^{\INM}(\Sigma)(\A)\models^{\INS}\gamma'$
iff $\Mod^{\INM'\circ\INM}(\Sigma)(\A)\models^{\INS''}\gamma'$
\item[$\Dex^{\INM'\circ\INM}(\cdot):\Dex^{\INS}_{\INM'\circ\INM}(\cdot)\Leftarrow\Dex^{\INS''}_{\INM'\circ\INM}(\cdot)\in\newprod_{\cdot\in\Sig^{\INS}}(\Sig^{\INM'\circ\INM}(\cdot)/\Sig^{\INS''})^{\Sig^{\INM'\circ\INM}(\cdot)_{\Dex^{\INS''}}}$]
\
\begin{itemize}
\item $\Dex^{\INS}_{\INM'\circ\INM}:=\newprod_{\Sig^{\INS}}((\cdot/\Sig^{\INM'\circ\INM})^{(\cdot)_{\Dex^{\INM'\circ\INM}}})(\Dex^{\INS})
=\newprod_{\Sig^{\INS}}((\cdot/(\Sig^{\INM'}\circ\Sig^{\INM}))^{(\cdot)_{\Dex^{\INM}}\circ_{V}((\cdot)_{\Dex^{\INM'}}\circ_{H}\Sig^{\INM})})(\Dex^{\INS})$
\item[] $=\newprod_{\Sig^{\INS}}(((\Sig^{\INM}(\cdot)/\Sig^{\INM'})\circ(\cdot/\Sig^{\INM}))^{(\cdot)_{\Dex^{\INM}}\circ_{V}((\cdot)_{\Dex^{\INM'}}\circ_{H}\Sig^{\INM})})(\Dex^{\INS})$
\item[] $=\newprod_{\Sig^{\INS}}((\Sig^{\INM}(\cdot)/\Sig^{\INM'})^{\Sig^{\INM}(\cdot)_{\Dex^{\INM'}}}\circ(\cdot/\Sig^{\INM})^{(\cdot)_{\Dex^{\INM}}})(\Dex^{\INS})$
\item[] $=(\newprod_{\Sig^{\INS}}((\Sig^{\INM}(\cdot)/\Sig^{\INM'})^{\Sig^{\INM}(\cdot)_{\Dex^{\INM'}}})\circ\newprod_{\Sig^{\INS}}((\cdot/\Sig^{\INM})^{(\cdot)_{\Dex^{\INM}}}))(\Dex^{\INS})$
\item[] $=\newprod_{\Sig^{\INS}}((\Sig^{\INM}(\cdot)/\Sig^{\INM'})^{\Sig^{\INM}(\cdot)_{\Dex^{\INM'}}})(\Dex^{\INS}_{\INM})$
\item $\Dex^{\INS''}_{\INM'\circ\INM}:=\newprod_{\Sig^{\INM'\circ\INM}}((\cdot/\Sig^{\INS''})^{(\cdot)_{\Dex^{\INS''}}})(\Dex^{\INS''})
=\newprod_{\Sig^{\INM'}\circ\Sig^{\INM}}((\cdot/\Sig^{\INS''})^{(\cdot)_{\Dex^{\INS''}}})(\Dex^{\INS''})$
\item[] $=(\newprod_{\Sig^{\INM}}((\cdot/\Sig^{\INS''})^{(\cdot)_{\Dex^{\INS''}}}\circ(\Sig^{\INM'})^{op})\circ\newprod_{\Sig^{\INM'}}((\cdot/\Sig^{\INS''})^{(\cdot)_{\Dex^{\INS''}}}))(\Dex^{\INS''})$ (Lemma~\ref{lemma:prod_is_a_list})
\item[] $=(\newprod_{\Sig^{\INM}}((\Sig^{\INM'}(\cdot)/\Sig^{\INS''})^{\Sig^{\INM'}(\cdot)_{\Dex^{\INS''}}})\circ\newprod_{\Sig^{\INM'}}((\cdot/\Sig^{\INS''})^{(\cdot)_{\Dex^{\INS''}}}))(\Dex^{\INS''})$
\item[] $=\newprod_{\Sig^{\INM}}((\Sig^{\INM'}(\cdot)/\Sig^{\INS''})^{\Sig^{\INM'}(\cdot)_{\Dex^{\INS''}}})(\Dex^{\INS''}_{\INM'})$
\item By $\Dex^{\INM}:\Dex^{\INS}_{\INM}\Leftarrow\Dex^{\INS'}_{\INM}$, $\Dex^{\INM'}:\Dex^{\INS'}_{\INM'}\Leftarrow\Dex^{\INS''}_{\INM'}$, and above two, 
$\Dex^{\INM'\circ\INM}(\cdot)
=\newprod_{\Sig^{\INS}}((\Sig^{\INM}(\cdot)/\Sig^{\INM'})^{\Sig^{\INM}(\cdot)_{\Dex^{\INM'}}})(\Dex^{\INM})
\circ_{V}
\newprod_{\Sig^{\INM}}((\Sig^{\INM'}(\cdot)/\Sig^{\INS''})^{\Sig^{\INM'}(\cdot)_{\Dex^{\INS''}}})(\Dex^{\INM'})
:\Dex^{\INS}_{\INM'\circ\INM}(\cdot)\Leftarrow\Dex^{\INS''}_{\INM'\circ\INM}(\cdot)\in\newprod_{\cdot\in\Sig^{\INS}}(\Sig^{\INM'\circ\INM}(\cdot)/\Sig^{\INS''})^{\Sig^{\INM'\circ\INM}(\cdot)_{\Dex^{\INS''}}}$
holds true.
\end{itemize}
\item[universal morphism]
Let $\Sigma\in|\Sig^{\INS}|$ and $X\in|\Sig^{\INM'\circ\INM}(\Sigma)_{\Dex^{\INS''}}|$.
For any $\chi\in|\Sigma/\Sig^{\INS}|$ and $\theta:\Sig^{\INM'\circ\INM}(\Sigma)^{\Dex^{\INS}}(X)\to(\Sigma/\Sig^{\INM'\circ\INM})(\chi)$,
we must prove that there exists a unique $\theta''\in|\Sigma/\Sig^{\INS}|$ such that $(\Sigma/\Sig^{\INM'\circ\INM})(\theta'')\circ\Sigma^{\Dex^{\INM'\circ\INM}}(X)=\theta$.
Note that $\Sigma^{\Dex^{\INM'\circ\INM}}(X)=(\Sig^{\INM}(\Sigma)/\Sig^{\INM'})(\Sigma^{\Dex^{\INM}}(\Sig^{\INM}(\Sigma)_{\Dex^{\INM'}}(X)))\circ\Sig^{\INM}(\Sigma)^{\Dex^{\INM'}}(X)$ by definition, so the condition can be rephrased as follows:
\begin{align*}
\theta&=(\Sigma/\Sig^{\INM'\circ\INM})(\theta'')\circ\Sigma^{\Dex^{\INM'\circ\INM}}(X)\\
&=((\Sig^{\INM}(\Sigma)/\Sig^{\INM'})\circ(\Sigma/\Sig^{\INM}))(\theta'')\circ(\Sig^{\INM}(\Sigma)/\Sig^{\INM'})(\Sigma^{\Dex^{\INM}}(\Sig^{\INM}(\Sigma)_{\Dex^{\INM'}}(X)))
\circ\Sig^{\INM}(\Sigma)^{\Dex^{\INM'}}(X)\\
&=(\Sig^{\INM}(\Sigma)/\Sig^{\INM'})((\Sigma/\Sig^{\INM})(\theta'')\circ\Sigma^{\Dex^{\INM}}(\Sig^{\INM}(\Sigma)_{\Dex^{\INM'}}(X)))
\circ\Sig^{\INM}(\Sigma)^{\Dex^{\INM'}}(X)
\end{align*}
Due to the properties of universal morphisms with respect to $\INM'$, there exists a unique $\theta'$ such that
$\theta=(\Sig^{\INM}(\Sigma)/\Sig^{\INM'})(\theta')\circ\Sig^{\INM}(\Sigma)^{\Dex^{\INM'}}(X)$
and due to the properties of universal morphisms with respect to $\INM$, there exists a unique $\theta''$ such that
$\theta'=(\Sigma/\Sig^{\INM})(\theta'')\circ\Sigma^{\Dex^{\INM}}(\Sig^{\INM}(\Sigma)_{\Dex^{\INM'}}(X))$.
\item[translatability of expanded models]
Let take $\Sigma\in|\Sig^{\INS}|$, $X''\in|\Sig^{\INM'\circ\INM}(\Sigma)_{\Dex^{\INS''}}|$, $\A\in|\Mod^{\INS}(\Sigma)|$, and
$\Sig^{\INM'\circ\INM}(\Sigma)^{\Dex^{\INS''}}(X'')$-expansion $\hat{\A}''$ of $\A'':=\Mod^{\INM'\circ\INM}(\Sigma)(\A)$.
\begin{itemize}
\item Let $\Sigma':=\Sig^{\INM}(\Sigma)\in|\Sig^{\INS'}|$, $\Sigma'':=\Sig^{\INM'}(\Sigma')\in|\Sig^{\INS''}|$, $X':=\Sigma'_{\Dex^{\INM'}}(X'')\in|\Sigma'_{\Dex^{\INS'}}|$. Note that $X''\in|\Sigma''_{\Dex^{\INS''}}|$.
\item
$(\Sigma^{\Dex^{\INM'\circ\INM}}*\Mod^{\INM'\circ\INM})(X'')=\Mod^{\INS''}(\Sigma^{\Dex^{\INM'\circ\INM}}[X''])\circ\Mod^{\INM'\circ\INM}(\Sigma^{\Dex^{\INS}}[\Sigma_{\Dex^{\INM'\circ\INM}}(X'')])$
\begin{itemize}
\item
$\Mod^{\INS''}(\Sigma^{\Dex^{\INM'\circ\INM}}[X''])$
$=\Mod^{\INS''}(\Sig^{\INM'}(\Sigma^{\Dex^{\INM}}[\Sig^{\INM}(\Sigma)_{\Dex^{\INM'}}(X'')])\circ\Sig^{\INM}(\Sigma)^{\Dex^{\INM'}}[X''])$ \\
$=\Mod^{\INS''}(\Sig^{\INM}(\Sigma)^{\Dex^{\INM'}}[X''])\circ\Mod^{\INS''}(\Sig^{\INM'}(\Sigma^{\Dex^{\INM}}[\Sig^{\INM}(\Sigma)_{\Dex^{\INM'}}(X'')]))$ \\
$=\Mod^{\INS''}(\Sigma'^{\Dex^{\INM'}}[X''])\circ\Mod^{\INS''}(\Sig^{\INM'}(\Sigma^{\Dex^{\INM}}[X']))$
\item
$\Mod^{\INM'\circ\INM}(\Sigma^{\Dex^{\INS}}[\Sigma_{\Dex^{\INM'\circ\INM}}(X'')])$
$=
((\Mod^{\INM'}\circ_{H}\Sig^{\INM\,op})\circ_{V}\Mod^{\INM})
(\Sigma^{\Dex^{\INS}}[\Sigma_{\Dex^{\INM}}(\Sig^{\INM}(\Sigma)_{\Dex^{\INM'}}(X''))])
$
\\
$=
\Mod^{\INM'}(\Sig^{\INM}(\Sigma^{\Dex^{\INS}}[\Sigma_{\Dex^{\INM}}(\Sig^{\INM}(\Sigma)_{\Dex^{\INM'}}(X''))]))
\circ
\Mod^{\INM}(\Sigma^{\Dex^{\INS}}[\Sigma_{\Dex^{\INM}}(\Sig^{\INM}(\Sigma)_{\Dex^{\INM'}}(X''))])
$
\\
$=
\Mod^{\INM'}(\Sig^{\INM}(\Sigma^{\Dex^{\INS}}[\Sigma_{\Dex^{\INM}}(X')]))
\circ
\Mod^{\INM}(\Sigma^{\Dex^{\INS}}[\Sigma_{\Dex^{\INM}}(X')])
$
\\
$=
\Mod^{\INM'}(
\Sigma^{\Dex^{\INS}_{\INM}}[X']
)
\circ
\Mod^{\INM}(\Sigma^{\Dex^{\INS}}[\Sigma_{\Dex^{\INM}}(X')])
$
\item[] Since $\Sigma^{\Dex^{\INM}}[X']:\Sigma^{\Dex^{\INS'}_{\INM}}[X']\to\Sigma^{\Dex^{\INS}_{\INM}}[X']$ and $\Mod^{\INM'}$ is a natural transformation, the following is commutative.
\begin{center}
$\xymatrix@R=30pt@C=120pt{
\Mod^{\INS'}(\Sigma^{\Dex^{\INS'}_{\INM}}[X']) \ar[d]|{\Mod^{\INM'}(\Sigma^{\Dex^{\INS'}_{\INM}}[X'])}
&\Mod^{\INS'}(\Sigma^{\Dex^{\INS}_{\INM}}[X']) \ar[l]|{\Mod^{\INS'}(\Sigma^{\Dex^{\INM}}[X'])} \ar[d]|{\Mod^{\INM'}(\Sigma^{\Dex^{\INS}_{\INM}}[X'])} \\
\Mod^{\INS''}(\Sig^{\INM'}(\Sigma^{\Dex^{\INS'}_{\INM}}[X']))
&\Mod^{\INS''}(\Sig^{\INM'}(\Sigma^{\Dex^{\INS}_{\INM}}[X'])) \ar[l]|{\Mod^{\INS''}(\Sig^{\INM'}(\Sigma^{\Dex^{\INM}}[X']))}
}$
\end{center}
\item[]
$\Mod^{\INS''}(\Sig^{\INM'}(\Sigma^{\Dex^{\INM}}[X']))\circ\Mod^{\INM'}(\Sigma^{\Dex^{\INS}_{\INM}}[X'])$
$=\Mod^{\INM'}(\Sigma^{\Dex^{\INS'}_{\INM}}[X'])\circ\Mod^{\INS'}(\Sigma^{\Dex^{\INM}}[X'])$ \\
$=\Mod^{\INM'}(
\Sig^{\INM}(\Sigma)^{\Dex^{\INS'}}[\Sig^{\INM}(\Sigma)_{\Dex^{\INM'}}(X'')]
)\circ\Mod^{\INS'}(\Sigma^{\Dex^{\INM}}[X'])$
$=\Mod^{\INM'}(
\Sigma'^{\Dex^{\INS'}}[\Sigma'_{\Dex^{\INM'}}(X'')]
)\circ\Mod^{\INS'}(\Sigma^{\Dex^{\INM}}[X'])$
\item[] Therefore, \\
$(\Sigma^{\Dex^{\INM'\circ\INM}}*\Mod^{\INM'\circ\INM})(X'')$ \\
$=
\Mod^{\INS''}(\Sigma'^{\Dex^{\INM'}}[X''])
\circ
\Mod^{\INM'}(\Sigma'^{\Dex^{\INS'}}[\Sigma'_{\Dex^{\INM'}}(X'')])
\circ
\Mod^{\INS'}(\Sigma^{\Dex^{\INM}}[X'])
\circ
\Mod^{\INM}(\Sigma^{\Dex^{\INS}}[\Sigma_{\Dex^{\INM}}(X')])
$
\\
$=
(\Sigma'^{\Dex^{\INM'}}*\Mod^{\INM'})(X'')
\circ
(\Sigma^{\Dex^{\INM}}*\Mod^{\INM})(X')
$
\end{itemize}
\item
From the above result, this operation can be split into two parts.
Since $\INM$ and $\INM'$ are $\Dex$-institution morphisms,
we can perform model translation for each of
$(\Sigma'^{\Dex^{\INM'}}*\Mod^{\INM'})(X'')$ and $(\Sigma^{\Dex^{\INM}}*\Mod^{\INM})(X')$.
Therefore, we can say that a model satisfying the conditions also exists for
$(\Sigma^{\Dex^{\INM'\circ\INM}}*\Mod^{\INM'\circ\INM})(X'')$.
(If both possess the uniqueness property, it exists uniquely.)
\end{itemize}
\end{proofcases}
\end{proofcases}
\item[confirmation of axioms]\
\begin{proofcases}
\item[identity element] Let $\INM:\INS\to\INS'\in\DINS$.
\begin{itemize}
\item
$\Dex^{\id_{\INS'}\circ\INM}(\cdot)
=
\newprod_{\Sig^{\INS}}((\Sig^{\INM}(\cdot)/\Sig^{\id_{\INS'}})^{\Sig^{\INM}(\cdot)_{\Dex^{\id_{\INS'}}}})(\Dex^{\INM})
\circ_{V}
\newprod_{\Sig^{\INM}}((\Sig^{\id_{\INS'}}(\cdot)/\Sig^{\INS'})^{\Sig^{\id_{\INS'}}(\cdot)_{\Dex^{\INS'}}})(\Dex^{\id_{\INS'}})
$
\\
$
=\Dex^{\INM}(\cdot)\circ_{V}\Dex^{\id_{\INS'}}(\cdot)
=\Dex^{\INM}(\cdot)
$
\item
$
\Dex^{\INM\circ\id_{\INS}}
=
\newprod_{\Sig^{\INS}}((\Sig^{\id_{\INS}}(\cdot)/\Sig^{\INM})^{\Sig^{\id_{\INS}}(\cdot)_{\Dex^{\INM}}})(\Dex^{\id_{\INS}})
\circ_{V}
\newprod_{\Sig^{\id_{\INS}}}((\Sig^{\INM}(\cdot)/\Sig^{\INS'})^{\Sig^{\INM}(\cdot)_{\Dex^{\INS'}}})(\Dex^{\INM})
$
\\
$
=\Dex^{\id_{\INS}}(\cdot)\circ_{V}\Dex^{\INM}(\cdot)
=\Dex^{\INM}(\cdot)
$
\end{itemize}
Therefore, $\Dex^{\id_{\INS'}\circ\INM}=\Dex^{\INM\circ\id_{\INS}}$.
\item[associative] Let $\INM:\INS\to\INS',\,\INM:\INS'\to\INS'',\,\INM:\INS''\to\INS'''\in\DINS$.
\begin{itemize}
\item Let $\Sigma\in|\Sig^{\INS}|$, $\Sigma':=\Sig^{\INM}(\Sigma)\in|\Sig^{\INS'}|$, $\Sigma'':=\Sig^{\INM'}(\Sigma')\in|\Sig^{\INS''}|$, $\Sigma''':=\Sig^{\INM''}(\Sigma'')\in|\Sig^{\INS'''}|$.
\item Let $X'''\in|\Sigma'''_{\Dex^{\INS'''}}|$, $X'':=\Sigma''_{\Dex^{\INM''}}(X''')\in|\Sigma''_{\Dex^{\INS''}}|$, $X':=\Sigma'_{\Dex^{\INM'}}(X'')\in|\Sigma'_{\Dex^{\INS'}}|$, $X:=\Sigma_{\Dex^{\INM}}(X')\in|\Sigma_{\Dex^{\INS}}|$.
\end{itemize}
\begin{itemize}
\item $\Dex^{(\INM''\circ\INM')\circ\INM}(\Sigma)=\pos{\Sigma^{\Dex^{(\INM''\circ\INM')\circ\INM}},\id_{\Sigma}}$
\item[]
$\Sigma^{\Dex^{(\INM''\circ\INM')\circ\INM}}[X''']$
$=\Sig^{\INM''\circ\INM'}(\Sigma^{\Dex^{\INM}}[\Sig^{\INM}(\Sigma)_{\Dex^{\INM''\circ\INM'}}(X''')])\circ\Sig^{\INM}(\Sigma)^{\Dex^{\INM''\circ\INM'}}[X''']$ \\
$=\Sig^{\INM''\circ\INM'}(\Sigma^{\Dex^{\INM}}[\Sigma'_{\Dex^{\INM'\circ\INM}}(X''')])\circ\Sigma'^{\Dex^{\INM''\circ\INM'}}[X''']$ \\
$=\Sig^{\INM''}(\Sig^{\INM'}(\Sigma^{\Dex^{\INM}}[\Sigma'_{\Dex^{\INM'}}(\Sig^{\INM'}(\Sigma')_{\Dex^{\INM''}}(X'''))]))\circ\Sig^{\INM''}(\Sigma'^{\Dex^{\INM'}}[\Sig^{\INM'}(\Sigma')_{\Dex^{\INM''}}(X''')])\circ\Sig^{\INM'}(\Sigma')^{\Dex^{\INM''}}[X''']$ \\
$=\Sig^{\INM''}(\Sig^{\INM'}(\Sigma^{\Dex^{\INM}}[X']))\circ\Sig^{\INM''}(\Sigma'^{\Dex^{\INM'}}[X''])\circ\Sigma''^{\Dex^{\INM''}}[X''']$ \\
$=\Sigma^{\Dex^{\INM}}[X']''\circ\Sigma'^{\Dex^{\INM'}}[X'']'\circ\Sigma''^{\Dex^{\INM''}}[X''']$ \\
where $\Sigma^{\Dex^{\INM}}[X']'':=\Sig^{\INM''}(\Sig^{\INM'}(\Sigma^{\Dex^{\INM}}[X']))$
and $\Sigma'^{\Dex^{\INM'}}[X'']':=\Sig^{\INM''}(\Sigma'^{\Dex^{\INM'}}[X''])$
\item $\Dex^{\INM''\circ(\INM'\circ\INM)}(\Sigma)=\pos{\Sigma^{\Dex^{\INM''\circ(\INM'\circ\INM)}},\id_{\Sigma}}$
\item[]
$\Sigma^{\Dex^{\INM''\circ(\INM'\circ\INM)}}[X''']$
$=\Sig^{\INM''}(\Sigma^{\Dex^{\INM'\circ\INM}}[\Sig^{\INM'\circ\INM}(\Sigma)_{\Dex^{\INM''}}(X''')])\circ\Sig^{\INM'\circ\INM}(\Sigma)^{\Dex^{\INM''}}[X''']$ \\
$=\Sig^{\INM''}(\Sigma^{\Dex^{\INM'\circ\INM}}[X''])\circ\Sigma''^{\Dex^{\INM''}}[X''']$ \\
$=\Sig^{\INM''}(\Sig^{\INM'}(\Sigma^{\Dex^{\INM}}[\Sig^{\INM}(\Sigma)_{\Dex^{\INM'}}(X'')])\circ\Sig^{\INM}(\Sigma)^{\Dex^{\INM'}}[X''])\circ\Sigma''^{\Dex^{\INM''}}[X''']$ \\
$=\Sig^{\INM''}(\Sig^{\INM'}(\Sigma^{\Dex^{\INM}}[X'])\circ\Sigma'^{\Dex^{\INM'}}[X''])\circ\Sigma''^{\Dex^{\INM''}}[X''']$ \\
$=\Sig^{\INM''}(\Sig^{\INM'}(\Sigma^{\Dex^{\INM}}[X']))\circ\Sig^{\INM''}(\Sigma'^{\Dex^{\INM'}}[X''])\circ\Sigma''^{\Dex^{\INM''}}[X''']$ \\
$=\Sigma^{\Dex^{\INM}}[X']''\circ\Sigma'^{\Dex^{\INM'}}[X'']'\circ\Sigma''^{\Dex^{\INM''}}[X''']$ \\
where $\Sigma^{\Dex^{\INM}}[X']'':=\Sig^{\INM''}(\Sig^{\INM'}(\Sigma^{\Dex^{\INM}}[X']))$
and $\Sigma'^{\Dex^{\INM'}}[X'']':=\Sig^{\INM''}(\Sigma'^{\Dex^{\INM'}}[X''])$
\end{itemize}
Therefore, $\Dex^{(\INM''\circ\INM')\circ\INM}=\Dex^{\INM''\circ(\INM'\circ\INM)}$.
\end{proofcases}
\end{proofcases}
\end{proof}

\begin{proof}[Proof of Lemma~\ref{lemma:S*M}]
$\Mod$ has already been discussed in the proof above, so we will now look at the part about $\Sen$.
\begin{proofcases}
\item[$(\Sen^{\id_{\INS}}*\Sigma^{\Dex^{\id_{\INS}}})(X)=\id_{\Sen^{\INS}(\Sigma^{\Dex^{\INS}}[X])}$]\
\begin{itemize}
\item[]$(\Sen^{\id_{\INS}}*\Sigma^{\Dex^{\id_{\INS}}})(X)$
\item[]$=\Sen^{\id_{\INS}}(\Sigma^{\Dex^{\INS}}[\Sigma_{\Dex^{\id_{\INS}}}(X)])\circ\Sen^{\INS}(\Sigma^{\Dex^{\id_{\INS}}}[X])$
\item[]$=\Sen^{\id_{\INS}}(\Sigma^{\Dex^{\INS}}[\id_{\Sigma_{\Dex^{\INS}}}(X)])\circ\Sen^{\INS}(\Sigma^{\id_{\Dex^{\INS}}}[X])$
\item[]$=\Sen^{\id_{\INS}}(\Sigma^{\Dex^{\INS}}[X])\circ\Sen^{\INS}(\id_{\Sigma^{\Dex^{\INS}}[X]})$
\item[]$=\id_{\Sen^{\INS}(\Sigma^{\Dex^{\INS}}[X])}\circ\id_{\Sen^{\INS}(\Sigma^{\Dex^{\INS}}[X])}$
\item[]$=\id_{\Sen^{\INS}(\Sigma^{\Dex^{\INS}}[X])}$
\end{itemize}
\item[$(\Sen^{\INM'\circ\INM}*\Sigma^{\Dex^{\INM'\circ\INM}})(X'')=(\Sen^{\INM}*\Sigma^{\Dex^{\INM}})(X')\circ(\Sen^{\INM'}*\Sigma'^{\Dex^{\INM'}})(X'')$]\ \\ \ \\
$(\Sen^{\INM'\circ\INM}*\Sigma^{\Dex^{\INM'\circ\INM}})(X'')
=\Sen^{\INM'\circ\INM}(\Sigma^{\Dex^{\INS}}[\Sigma_{\Dex^{\INM'\circ\INM}}(X'')])
\circ\Sen^{\INS''}(\Sigma^{\Dex^{\INM'\circ\INM}}[X''])$
\begin{itemize}
\item
$\Sen^{\INS''}(\Sigma^{\Dex^{\INM'\circ\INM}}[X''])$
$=\Sen^{\INS''}(\Sig^{\INM'}(\Sigma^{\Dex^{\INM}}[\Sig^{\INM}(\Sigma)_{\Dex^{\INM'}}(X'')])\circ\Sig^{\INM}(\Sigma)^{\Dex^{\INM'}}[X''])$ \\
$=\Sen^{\INS''}(\Sig^{\INM'}(\Sigma^{\Dex^{\INM}}[\Sig^{\INM}(\Sigma)_{\Dex^{\INM'}}(X'')]))
\circ\Sen^{\INS''}(\Sig^{\INM}(\Sigma)^{\Dex^{\INM'}}[X''])$\\
$=\Sen^{\INS''}(\Sig^{\INM'}(\Sigma^{\Dex^{\INM}}[X']))
\circ\Sen^{\INS''}(\Sigma'^{\Dex^{\INM'}}[X''])$
\item
$\Sen^{\INM'\circ\INM}(\Sigma^{\Dex^{\INS}}[\Sigma_{\Dex^{\INM'\circ\INM}}(X'')])$
$=
(\Sen^{\INM}\circ_{V}(\Sen^{\INM'}\circ_{H}\Sig^{\INM}))
(\Sigma^{\Dex^{\INS}}[\Sigma_{\Dex^{\INM}}(\Sig^{\INM}(\Sigma)_{\Dex^{\INM'}}(X''))])
$
\\
$=
\Sen^{\INM}(\Sigma^{\Dex^{\INS}}[\Sigma_{\Dex^{\INM}}(\Sig^{\INM}(\Sigma)_{\Dex^{\INM'}}(X''))])
\circ
\Sen^{\INM'}(\Sig^{\INM}(\Sigma^{\Dex^{\INS}}[\Sigma_{\Dex^{\INM}}(\Sig^{\INM}(\Sigma)_{\Dex^{\INM'}}(X''))]))
$
\\
$=
\Sen^{\INM}(\Sigma^{\Dex^{\INS}}[\Sigma_{\Dex^{\INM}}(X')])
\circ
\Sen^{\INM'}(\Sig^{\INM}(\Sigma^{\Dex^{\INS}}[\Sigma_{\Dex^{\INM}}(X')]))
$
\\
$=
\Sen^{\INM}(\Sigma^{\Dex^{\INS}}[\Sigma_{\Dex^{\INM}}(X')])
\circ
\Sen^{\INM'}(\Sigma^{\Dex^{\INS}_{\INM}}[X'])
$
\item[] Since $\Sigma^{\Dex^{\INM}}[X']:\Sigma^{\Dex^{\INS'}_{\INM}}[X']\to\Sigma^{\Dex^{\INS}_{\INM}}[X']$ and $\Sen^{\INM'}$ is a natural transformation, the following is commutative.
\begin{center}
$\xymatrix@R=30pt@C=120pt{
\Sen^{\INS'}(\Sigma^{\Dex^{\INS'}_{\INM}}[X']) \ar@{<-}[d]|{\Sen^{\INM'}(\Sigma^{\Dex^{\INS'}_{\INM}}[X'])}
&\Sen^{\INS'}(\Sigma^{\Dex^{\INS}_{\INM}}[X']) \ar@{<-}[l]|{\Sen^{\INS'}(\Sigma^{\Dex^{\INM}}[X'])} \ar@{<-}[d]|{\Sen^{\INM'}(\Sigma^{\Dex^{\INS}_{\INM}}[X'])} \\
\Sen^{\INS''}(\Sig^{\INM'}(\Sigma^{\Dex^{\INS'}_{\INM}}[X']))
&\Sen^{\INS''}(\Sig^{\INM'}(\Sigma^{\Dex^{\INS}_{\INM}}[X'])) \ar@{<-}[l]|{\Sen^{\INS''}(\Sig^{\INM'}(\Sigma^{\Dex^{\INM}}[X']))}
}$
\end{center}
\item[]
$\Sen^{\INM'}(\Sigma^{\Dex^{\INS}_{\INM}}[X'])\circ\Sen^{\INS''}(\Sig^{\INM'}(\Sigma^{\Dex^{\INM}}[X']))$
$=\Sen^{\INS'}(\Sigma^{\Dex^{\INM}}[X'])\circ\Sen^{\INM'}(\Sigma^{\Dex^{\INS'}_{\INM}}[X'])$ \\
$=\Sen^{\INS'}(\Sigma^{\Dex^{\INM}}[X'])
\circ\Sen^{\INM'}(\Sig^{\INM}(\Sigma)^{\Dex^{\INS'}}[\Sig^{\INM}(\Sigma)_{\Dex^{\INM'}}(X'')])$
$=\Sen^{\INS'}(\Sigma^{\Dex^{\INM}}[X'])
\circ\Sen^{\INM'}(\Sigma'^{\Dex^{\INS'}}[\Sigma'_{\Dex^{\INM'}}(X'')])$
\item[] Therefore, \\
$(\Sen^{\INM'\circ\INM}*\Sigma^{\Dex^{\INM'\circ\INM}})(X'')$ \\
$=
\Sen^{\INM}(\Sigma^{\Dex^{\INS}}[\Sigma_{\Dex^{\INM}}(X')])
\circ
\Sen^{\INS'}(\Sigma^{\Dex^{\INM}}[X'])
\circ
\Sen^{\INM'}(\Sigma'^{\Dex^{\INS'}}[\Sigma'_{\Dex^{\INM'}}(X'')])
\circ
\Sen^{\INS''}(\Sigma'^{\Dex^{\INM'}}[X''])
$
\\
$=
(\Sen^{\INM}*\Sigma^{\Dex^{\INM}})(X')
\circ
(\Sen^{\INM'}*\Sigma'^{\Dex^{\INM'}})(X'')
$
\end{itemize}
\end{proofcases}
\end{proof}

\begin{proof}[Proof of well-definedness of Definition~\ref{def:FOL:}]
We will investigate well-definedness.
(Incidentally, as can be seen from the proof below, the uniqueness of the extension, which is the last condition in the definition of $\Dex$-institution and $\Dex$-institution morphism, is not required for introducing the quantifier, but if uniqueness is assumed, we can also introduce $\exists^{\kappa}$ (there exist $\kappa$ instances).)
\begin{proofcases}
\item[well-defined]
Since no other parts have been changed, we only need to check the $\Sen$ and $\models$ parts.
\begin{proofcases}
\item[$\FOL^{{\relax}}_{\alpha}(\INS)\in|\DINS|$]\
\begin{proofcases}
\item[$\Sen^{\FOL^{{\relax}}_{\alpha}(\INS)}:\Sig^{\FOL^{{\relax}}_{\alpha}(\INS)}\to\Set\in\CAT$] Let $\INS\in|\DINS|$.
\begin{proofcases}
\item[well-defined]
Unlike induction, recursive definitions require the set-like property of the binary relation used for recursion.
For example it is generally not possible to define $\Sen_{n}$ recursively using the “size of a sentence” (number of symbols used).
This is because the definition of $\Sen^{\FOL^{{\relax}}_{\alpha}(\INS)}(\Sigma)$ moves the elements of $\Sig^{\INS}$, so even if we collect only atomic formulas, $\Sen_{0}$ will still be a proper class.

However, in practice, only the elements of $\cup\{\Sen^{\FOL^{{\relax}}_{\alpha}(\INS)}(\Sigma[X])\mid X\in|\Sigma_{\Dex^{\INS}}|\}$ are referenced to determine whether a quantifier formula is a $\Sigma$-sentence.
So, if we first fix $\Sigma$ and then define the entire class of signatures obtained by finitely extending it with variables as $\Dex^{\INS\,*}_{\Sigma}$,
then since this part is a set, it allowing a recursive definition of the characteristic function $1_{\{v\in\Sen(s)\mid v,s\}}:\pos{V,{\in^{+}}}\times\Dex^{\INS\,*}_ {\Sigma}\to2$.
($V$ is the class of all sets. Note that $a,b\in^{2}\{\{a\},\{a,b\}\}=:\pos{a,b}$.)
By gluing these functions together, we obtain a similar result.
Note that a collection of sentences is a set. Let the class of sentences that appear in $\Dex^{\INS\,*}_ {\Sigma}$ be $\Sen^{\FOL^{{\relax}}_{\alpha}(\INS)}(\Dex^{\INS\,*}_ {\Sigma})$.
If we divide this into sets by sentence height, we can write $\Sen^{\FOL^{{\relax}}_{\alpha}(\INS)}(\Dex^{\INS\,*}_ {\Sigma})=\cup_{\beta\in ON}\Sen^{\FOL^{{\relax}}_{\alpha}(\INS)}_{\beta}(\Dex^{\INS\,*}_ {\Sigma})$.
This seems like it would be a proper class, but in reality, since $\alpha^{+}$ is a regular cardinal, the expansion stops halfway and becomes $\cup_{\beta\in\alpha^{+}}\Sen^{\FOL^{{\relax}}_{\alpha}(\INS)}_{\beta}(\Dex^{\INS\,*}_ {\Sigma})$.
The same holds true for $\Sen^{\FOL^{{\relax}}_{\alpha}(\INS)}(\chi)$.
\item[preserve operations]
Ensure that operations are preserved.
\begin{proofcases}
\item[$\Sen^{\FOL^{{\relax}}_{\alpha}(\INS)}(\id_{\Sigma})=\id_{\Sen^{\FOL^{{\relax}}_{\alpha}(\INS)}(\Sigma)}$]
We show that the following $P(\gamma)$ holds true for each sentence $\gamma$ by induction based on sentence complexity.
($P(\gamma)\Leftrightarrow\Sen^{\FOL^{{\relax}}_{\alpha}(\INS)}(\id_{\Sigma})(\gamma)=\id_{\Sen^{\FOL^{{\relax}}_{\alpha}(\INS)}(\Sigma)}(\gamma)$ holds true for each $\Sigma\in|\Sig^{\INS}|$ such that $\gamma\in\Sen^{\FOL^{{\relax}}_{\alpha}(\INS)}(\Sigma)$)
\footnote{
Sentence complexity refers to “partial tree orders on tree structures consisting solely of symbols.”
Note that by deliberately swapping $\gamma$ and $\Sigma$ here, we demonstrate that it is permissible to alter the signature under the induction hypothesis (IH) being referenced.
}
The same applies to the other cases, so we will only check three.
\begin{proofcases}
\item[$\gamma=\pos{\phi,atomic}$]
\begin{align*}
&\Sen^{\FOL^{{\relax}}_{\alpha}(\INS)}(\id_{\Sigma})(\pos{\phi,atomic})
=\pos{\Sen^{\INS}(\id_{\Sigma})(\phi),atomic}
=\pos{\id_{\Sen^{\INS}(\Sigma)}(\phi),atomic} \\
&=\pos{\phi,atomic}
=\id_{\Sen^{\FOL^{{\relax}}_{\alpha}(\INS)}(\Sigma)}(\pos{\phi,atomic})
\end{align*}
\item[$\gamma=\neg\phi$]
\begin{align*}
&\Sen^{\FOL^{{\relax}}_{\alpha}(\INS)}(\id_{\Sigma})(\neg\phi)
=\neg\Sen^{\FOL^{{\relax}}_{\alpha}(\INS)}(\id_{\Sigma})(\phi)
\stackrel{IH}{=}\neg\id_{\Sen^{\FOL^{{\relax}}_{\alpha}(\INS)}(\Sigma)}(\phi) \\
&\neg\phi
=\id_{\Sen^{\FOL^{{\relax}}_{\alpha}(\INS)}(\Sigma)}(\neg\phi)
\end{align*}
\item[$\gamma=\Exists{X}\phi$]
\begin{align*}
&\Sen^{\FOL^{{\relax}}_{\alpha}(\INS)}(\id_{\Sigma})(\Exists{X}\phi)
=\Exists{\id_{\Sigma\,\Dex^{\INS}}(X)}\Sen^{\FOL^{{\relax}}_{\alpha}(\INS)}(\id_{\Sigma}^{\Dex^{\INS}}[X])(\phi) \\
&\stackrel{\ref{eq:IDex-3}}{=}\Exists{\id_{\Sigma\,\Dex^{\INS}}(X)}\Sen^{\FOL^{{\relax}}_{\alpha}(\INS)}(\id_{\Sigma^{\Dex^{\INS}}[X]})(\phi)
\stackrel{IH}{=}\Exists{\id_{\Sigma\,\Dex^{\INS}}(X)}\id_{\Sen^{\FOL^{{\relax}}_{\alpha}(\INS)}(\Sigma^{\Dex^{\INS}}[X])}(\phi) \\
&=\Exists{\id_{\Sigma^{\Dex^{\INS}}(X)}}\id_{\Sen^{\FOL^{{\relax}}_{\alpha}(\INS)}(\Sigma^{\Dex^{\INS}}[X])}(\phi)
=\Exists{X}\phi
=\id_{\Sen^{\FOL^{{\relax}}_{\alpha}(\INS)}(\Sigma)}(\Exists{X}\phi)
\end{align*}
\end{proofcases}
\item[$\Sen^{\FOL^{{\relax}}_{\alpha}(\INS)}(\chi'\circ\chi)=\Sen^{\FOL^{{\relax}}_{\alpha}(\INS)}(\chi')\circ\Sen^{\FOL^{{\relax}}_{\alpha}(\INS)}(\chi)$]
We show that the following $P(\gamma)$ holds true for each sentence $\gamma$ by induction based on sentence complexity.
($P(\gamma)\Leftrightarrow\Sen^{\FOL^{{\relax}}_{\alpha}(\INS)}(\chi'\circ\chi)(\gamma)=\Sen^{\FOL^{{\relax}}_{\alpha}(\INS)}(\chi')\circ\Sen^{\FOL^{{\relax}}_{\alpha}(\INS)}(\chi)(\gamma)$ holds true for each $\chi:\Sigma\to\Sigma'\in\Sig^{\INS}$ such that $\gamma\in\Sen^{\FOL^{{\relax}}_{\alpha}(\INS)}(\Sigma)$)
The same applies to the other cases, so we will only check three.
\begin{proofcases}
\item[$\gamma=\pos{\phi,atomic}$]
\begin{align*}
&\Sen^{\FOL^{{\relax}}_{\alpha}(\INS)}(\chi'\circ\chi)(\pos{\phi,atomic})
=\pos{\Sen^{\INS}(\chi'\circ\chi)(\phi),atomic}
=\pos{\Sen^{\INS}(\chi')\circ\Sen^{\INS}(\chi)(\phi),atomic} \\
&=\pos{\Sen^{\INS}(\chi')(\Sen^{\INS}(\chi)(\phi)),atomic}
=\Sen^{\FOL^{{\relax}}_{\alpha}(\INS)}(\chi')\circ\Sen^{\FOL^{{\relax}}_{\alpha}(\INS)}(\chi)(\pos{\phi,atomic})
\end{align*}
\item[$\gamma=\neg\phi$]
\begin{align*}
&\Sen^{\FOL^{{\relax}}_{\alpha}(\INS)}(\chi'\circ\chi)(\neg\phi)
=\neg\Sen^{\FOL^{{\relax}}_{\alpha}(\INS)}(\chi'\circ\chi)(\phi)
\stackrel{IH}{=}\neg\Sen^{\FOL^{{\relax}}_{\alpha}(\INS)}(\chi')\circ\Sen^{\FOL^{{\relax}}_{\alpha}(\INS)}(\chi)(\phi) \\
&=\neg\Sen^{\FOL^{{\relax}}_{\alpha}(\INS)}(\chi')(\Sen^{\FOL^{{\relax}}_{\alpha}(\INS)}(\chi)(\phi))
=\Sen^{\FOL^{{\relax}}_{\alpha}(\INS)}(\chi')\circ\Sen^{\FOL^{{\relax}}_{\alpha}(\INS)}(\chi)(\neg\phi)
\end{align*}
\item[$\gamma=\Exists{X}\phi$]
\begin{align*}
&\Sen^{\FOL^{{\relax}}_{\alpha}(\INS)}(\chi'\circ\chi)(\Exists{X}\phi)
=\Exists{(\chi'\circ\chi)_{\Dex^{\INS}}(X)}\Sen^{\FOL^{{\relax}}_{\alpha}(\INS)}((\chi'\circ\chi)^{\Dex^{\INS}}[X])(\phi) \\
&\stackrel{\ref{eq:IDex-3}}{=}\Exists{(\chi'\circ\chi)_{\Dex^{\INS}}(X)}\Sen^{\FOL^{{\relax}}_{\alpha}(\INS)}(\chi'^{\Dex^{\INS}}[\chi_{\Dex^{\INS}}(X)]\circ\chi^{\Dex^{\INS}}[X])(\phi) \\
&\stackrel{IH}{=}\Exists{(\chi'\circ\chi)_{\Dex^{\INS}}(X)}\Sen^{\FOL^{{\relax}}_{\alpha}(\INS)}(\chi'^{\Dex^{\INS}}[\chi_{\Dex^{\INS}}(X)])\circ\Sen^{\FOL^{{\relax}}_{\alpha}(\INS)}(\chi^{\Dex^{\INS}}[X])(\phi) \\
&=\Exists{\chi'_{\Dex^{\INS}}\circ\chi_{\Dex^{\INS}}(X)}\Sen^{\FOL^{{\relax}}_{\alpha}(\INS)}(\chi'^{\Dex^{\INS}}[\chi_{\Dex^{\INS}}(X)])\circ\Sen^{\FOL^{{\relax}}_{\alpha}(\INS)}(\chi^{\Dex^{\INS}}[X])(\phi) \\
&=\Exists{\chi'_{\Dex^{\INS}}(\chi_{\Dex^{\INS}}(X))}\Sen^{\FOL^{{\relax}}_{\alpha}(\INS)}(\chi'^{\Dex^{\INS}}[\chi_{\Dex^{\INS}}(X)])(\Sen^{\FOL^{{\relax}}_{\alpha}(\INS)}(\chi^{\Dex^{\INS}}[X])(\phi)) \\
&=\Sen^{\FOL^{{\relax}}_{\alpha}(\INS)}(\chi')(\Exists{\chi_{\Dex^{\INS}}(X)}\Sen^{\FOL^{{\relax}}_{\alpha}(\INS)}(\chi^{\Dex^{\INS}}[X])(\phi)) \\
&=\Sen^{\FOL^{{\relax}}_{\alpha}(\INS)}(\chi')(\Sen^{\FOL^{{\relax}}_{\alpha}(\INS)}(\chi)(\Exists{X}\phi))
=\Sen^{\FOL^{{\relax}}_{\alpha}(\INS)}(\chi')\circ\Sen^{\FOL^{{\relax}}_{\alpha}(\INS)}(\chi)(\Exists{X}\phi)
\end{align*}
\end{proofcases}
\end{proofcases}
\end{proofcases}
\item[$\Mod^{\FOL^{{\relax}}_{\alpha}(\INS)}(\chi)(\A')\models^{\FOL^{{\relax}}_{\alpha}(\INS)}\gamma$ iff $\A'\models^{\FOL^{{\relax}}_{\alpha}(\INS)}\Sen^{\FOL^{{\relax}}_{\alpha}(\INS)}(\chi)(\gamma)$]
We show that the following $P(\gamma)$ holds true for each sentence $\gamma$ by induction based on sentence complexity.
($P(\gamma)\Leftrightarrow$
$\Mod^{\FOL^{{\relax}}_{\alpha}(\INS)}(\chi)(\A')\models^{\FOL^{{\relax}}_{\alpha}(\INS)}\gamma$ iff $\A'\models^{\FOL^{{\relax}}_{\alpha}(\INS)}\Sen^{\FOL^{{\relax}}_{\alpha}(\INS)}(\chi)(\gamma)$
holds true for each $\chi:\Sigma\to\Sigma'\in\Sig^{\INS}$ and $\A'\in|\Mod^{\FOL^{{\relax}}_{\alpha}(\INS)}(\Sigma')|$
such that $\gamma\in\Sen^{\FOL^{{\relax}}_{\alpha}(\INS)}(\Sigma)$)
The same applies to the other cases, so we will only check three.
\begin{proofcases}
\item[$\gamma=\pos{\phi,atomic}$]
\begin{align*}
&\Mod^{\FOL^{{\relax}}_{\alpha}(\INS)}(\chi)(\A')\models^{\FOL^{{\relax}}_{\alpha}(\INS)}\pos{\phi,atomic} \\
&\Leftrightarrow\Mod^{\INS}(\chi)(\A')\models^{\INS}\phi
\Leftrightarrow\A'\models^{\INS}\Sen^{\INS}(\chi)(\phi) \\
&\Leftrightarrow\A'\models^{\FOL^{{\relax}}_{\alpha}(\INS)}\Sen^{\FOL^{{\relax}}_{\alpha}(\INS)}(\chi)(\pos{\phi,atomic})
\end{align*}
\item[$\gamma=\neg\phi$]
\begin{align*}
&\Mod^{\FOL^{{\relax}}_{\alpha}(\INS)}(\chi)(\A')\models^{\FOL^{{\relax}}_{\alpha}(\INS)}\neg\phi
\Leftrightarrow\Mod^{\FOL^{{\relax}}_{\alpha}(\INS)}(\chi)(\A')\not\models^{\FOL^{{\relax}}_{\alpha}(\INS)}\phi \\
&\stackrel{IH}{\Leftrightarrow}\A'\not\models^{\FOL^{{\relax}}_{\alpha}(\INS)}\Sen^{\FOL^{{\relax}}_{\alpha}(\INS)}(\chi)(\phi)
\Leftrightarrow\A'\models^{\FOL^{{\relax}}_{\alpha}(\INS)}\Sen^{\FOL^{{\relax}}_{\alpha}(\INS)}(\chi)(\neg\phi)
\end{align*}
\item[$\gamma=\Exists{X}\phi$]
\begin{align*}
&\Mod^{\FOL^{{\relax}}_{\alpha}(\INS)}(\chi)(\A')\models^{\FOL^{{\relax}}_{\alpha}(\INS)}\Exists{X}\phi \\
&\Leftrightarrow\hat{\A}\models^{\FOL^{{\relax}}_{\alpha}(\INS)}\phi \text{ for some } \Sigma^{\Dex^{\INS}}(X)\text{-expansion } \hat{\A} \text{ of } \A:=\Mod^{\FOL^{{\relax}}_{\alpha}(\INS)}(\chi)(\A')=\Mod^{\INS}(\chi)(\A') \\
&\A'\models^{\FOL^{{\relax}}_{\alpha}(\INS)}\Sen^{\FOL^{{\relax}}_{\alpha}(\INS)}(\chi)(\Exists{X}\phi) \\
&\Leftrightarrow\A'\models^{\FOL^{{\relax}}_{\alpha}(\INS)}\Exists{\chi_{\Dex^{\INS}}(X)}\Sen^{\FOL^{{\relax}}_{\alpha}(\INS)}(\chi^{\Dex^{\INS}}[X])(\phi) \\
&\Leftrightarrow\hat{\A}'\models^{\FOL^{{\relax}}_{\alpha}(\INS)}\Sen^{\FOL^{{\relax}}_{\alpha}(\INS)}(\chi^{\Dex^{\INS}}[X])(\phi) \text{ for some } \Sigma'^{\Dex^{\INS}}(\chi_{\Dex^{\INS}}(X))\text{-expansion } \hat{\A}' \text{ of } \A' \\
&\stackrel{IH}{\Leftrightarrow}\Mod^{\FOL^{{\relax}}_{\alpha}(\INS)}(\chi^{\Dex^{\INS}}[X])(\hat{\A}')\models^{\FOL^{{\relax}}_{\alpha}(\INS)}\phi \text{ for some } \Sigma'^{\Dex^{\INS}}(\chi_{\Dex^{\INS}}(X))\text{-expansion } \hat{\A}' \text{ of } \A' \\
&\Leftrightarrow\hat{\A}\models^{\FOL^{{\relax}}_{\alpha}(\INS)}\phi  \text{ for some } \hat{\A} \\ 
&\,\,\,\,\,\text{ such that }\hat{\A}=\Mod^{\INS}(\chi^{\Dex^{\INS}}[X])(\hat{\A}') \text{ for some } \Sigma'^{\Dex^{\INS}}(\chi_{\Dex^{\INS}}(X))\text{-expansion } \hat{\A}' \text{ of } \A' \\
&\stackrel{*}{\Leftrightarrow}\hat{\A}\models^{\FOL^{{\relax}}_{\alpha}(\INS)}\phi \text{ for some } \Sigma^{\Dex^{\INS}}(X)\text{-expansion } \hat{\A} \text{ of } \A:=\Mod^{\INS}(\chi)(\A') \\
\end{align*}
* It is clear from the last condition of $\Dex$-institution that the part after "such that" is rewritten in this way.
\begin{center}
$\xymatrix@R=40pt@C=40pt{
\hat{\A} \ar@{.}[r] \ar@{.>}[d]|{\red} & \Mod^{\INS}(\Sigma^{\Dex^{\INS}}[X])
\ar@{<-}[rrr]|{\Mod^{\INS}(\chi^{\Dex^{\INS}}[X])} \ar[d]|{\Mod^{\INS}(\Sigma^{\Dex^{\INS}}(X))} &&& \Mod^{\INS}(\Sigma'^{\Dex^{\INS}}[\chi_{\Dex^{\INS}}(X)]) \ar[d]|{\Mod^{\INS}(\Sigma'^{\Dex^{\INS}}(\chi_{\Dex^{\INS}}(X)))} & \ar@{.}[l] \hat{\A}' \ar@{.>}[d]|{\red} \\
\A \ar@{.}[r] & \Mod^{\INS}(\Sigma) \ar@{<-}[rrr]|{\Mod^{\INS}(\chi)} &&& \Mod^{\INS}(\Sigma') & \ar@{.}[l] \A'
}$
\end{center}
\item[$\gamma=\exists^{\kappa}{X}\cdot\phi$ (extra)]
We will show that if $\INS$ has the uniqueness property, the satisfaction condition holds true even if we add $\exists^{\kappa}{X}\cdot\phi$ (there are at least $\kappa$ instances).
In a similar way to above, we can show the following:
\begin{align*}
&\Mod^{\FOL^{{\relax}}_{\alpha}(\INS)}(\chi)(\A')\models^{\FOL^{{\relax}}_{\alpha}(\INS)}\exists^{\kappa}{X}\cdot\phi \\
&\Leftrightarrow\hat{\A}\models^{\FOL^{{\relax}}_{\alpha}(\INS)}\phi \text{ for } \kappa\ \Sigma^{\Dex^{\INS}}(X)\text{-expansions } \hat{\A} \text{ of } \A:=\Mod^{\FOL^{{\relax}}_{\alpha}(\INS)}(\chi)(\A')=\Mod^{\INS}(\chi)(\A') \\
&\Leftrightarrow\card\{\hat{\A}\mid\hat{\A}\models^{\FOL^{{\relax}}_{\alpha}(\INS)}\phi\text{ and }\hat{\A} \text{ is a } \Sigma^{\Dex^{\INS}}(X)\text{-expansion} \text{ of } \A:=\Mod^{\INS}(\chi)(\A')\}\geq\kappa\\
&\A'\models^{\FOL^{{\relax}}_{\alpha}(\INS)}\Sen^{\FOL^{{\relax}}_{\alpha}(\INS)}(\chi)(\exists^{\kappa}{X}\cdot\phi) \\
&\stackrel{IH}{\Leftrightarrow}\Mod^{\FOL^{{\relax}}_{\alpha}(\INS)}(\chi^{\Dex^{\INS}}[X])(\hat{\A}')\models^{\FOL^{{\relax}}_{\alpha}(\INS)}\phi \text{ for } \kappa\ \Sigma'^{\Dex^{\INS}}(\chi_{\Dex^{\INS}}(X))\text{-expansions } \hat{\A}' \text{ of } \A' \\
&\Leftrightarrow\card\{\hat{\A}'\mid\Mod^{\INS}(\chi^{\Dex^{\INS}}[X])(\hat{\A}')\models^{\FOL^{{\relax}}_{\alpha}(\INS)}\phi\text{ and } \hat{\A}' \text{ is a } \Sigma'^{\Dex^{\INS}}(\chi_{\Dex^{\INS}}(X))\text{-expansion} \text{ of } \A'\}\geq\kappa
\end{align*}
From the unique existence condition, $\Mod^{\INS}(\chi^{\Dex^{\INS}}[X])$ is a bijection between these two sets, so the two conditions coincide.
\end{proofcases}
\end{proofcases}
\item[$\FOL^{{\relax}}_{\alpha}(\INM:\INS\to\INS'):\FOL^{{\relax}}_{\alpha}(\INS)\to\FOL^{{\relax}}_{\alpha}(\INS')\in\DINS$] Let $\INM:\INS\to\INS'\in\DINS$.
\begin{proofcases}
\item[$\Sen^{\FOL^{{\relax}}_{\alpha}(\INM)}:\Sen^{\FOL^{{\relax}}_{\alpha}(\INS)}(\cdot)\Leftarrow\Sen^{\FOL^{{\relax}}_{\alpha}(\INS')}(\Sig^{\FOL^{{\relax}}_{\alpha}(\INM)}(\cdot))\in\Set^{\Sig^{\FOL^{{\relax}}_{\alpha}(\INS)}}$]\
\begin{proofcases}
\item[well-defined]
First, notice that $\Sigma^{\Dex^{\INM}}[X]\in\Sig^{\INS'}$.
$
\Sen^{\FOL^{{\relax}}_{\alpha}(\INS')}(\Sigma^{\Dex^{\INM}}[X])
$
does not change the complexity of the sentences.
Next, pay attention to $\Sigma_{\Dex^{\INM}}(X)\in\Sigma_{\Dex^{\INS}}$. In other words, what is referenced in the recursive definition is what is obtained by repeating the extension of $\Sigma$ by variables a finite number of times.
So, in reality, the entire $\Sen^{\FOL^{{\relax}}_{\alpha}(\INM)}$ shown in the diagram has not yet been defined,
but if we restrict it in the same way as with $\Sen^{\INS}$,
we can consider the part used in the recursion to be already defined.
\item[commutativity]
We show that the following $P(\gamma)$ holds true for each sentence $\gamma$ by induction based on sentence complexity.
($P(\gamma)\Leftrightarrow$ when $\gamma$ is substituted, the following commutes:
\begin{center}
$\xymatrix@R=40pt@C=20pt{
& \Sen^{\FOL^{{\relax}}_{\alpha}(\INS')}(\Sig^{\INM}(\Sigma)) \ar[rrrr]|{\Sen^{\FOL^{{\relax}}_{\alpha}(\INS')}(\Sig^{\INM}(\chi))} \ar[d]|{\Sen^{\FOL^{{\relax}}_{\alpha}(\INM)}(\Sigma)} &&&& \Sen^{\FOL^{{\relax}}_{\alpha}(\INS')}(\Sig^{\INM}(\Sigma')) \ar[d]|{\Sen^{\FOL^{{\relax}}_{\alpha}(\INM)}(\Sigma')} & \\
& \Sen^{\FOL^{{\relax}}_{\alpha}(\INS)}(\Sigma) \ar[rrrr]|{\Sen^{\FOL^{{\relax}}_{\alpha}(\INS)}(\chi)} &&&& \Sen^{\FOL^{{\relax}}_{\alpha}(\INS)}(\Sigma') &
}$
\end{center}
for each $\chi:\Sigma\to\Sigma'\in\Sig^{\INS}$ such that $\gamma\in\Sen^{\FOL^{{\relax}}_{\alpha}(\INS')}(\Sig^{\INM}(\Sigma))$.)
The same applies to the other cases, so we will only check three.
\begin{proofcases}
\item[$\gamma=\pos{\phi,atomic}$]
\begin{align*}
&(\Sen^{\FOL^{{\relax}}_{\alpha}(\INS)}(\chi)\circ\Sen^{\FOL^{{\relax}}_{\alpha}(\INM)}(\Sigma))(\pos{\phi,atomic})
=\Sen^{\FOL^{{\relax}}_{\alpha}(\INS)}(\chi)(\pos{\Sen^{\INM}(\Sigma)(\phi),atomic}) \\
&=\pos{\Sen^{\INS}(\chi)(\Sen^{\INM}(\Sigma)(\phi)),atomic}
=\pos{(\Sen^{\INS}(\chi)\circ\Sen^{\INM}(\Sigma))(\phi),atomic} \\
&=\pos{(\Sen^{\INM}(\Sigma')\circ\Sen^{\INS'}(\Sig^{\INM}(\chi)))(\phi),atomic}
=\pos{\Sen^{\INM}(\Sigma')(\Sen^{\INS'}(\Sig^{\INM}(\chi))(\phi)),atomic} \\
&=\Sen^{\FOL^{{\relax}}_{\alpha}(\INM)}(\Sigma')(\pos{\Sen^{\INS'}(\Sig^{\INM}(\chi))(\phi),atomic})
=(\Sen^{\FOL^{{\relax}}_{\alpha}(\INM)}(\Sigma')\circ\Sen^{\FOL^{{\relax}}_{\alpha}(\INS')}(\Sig^{\INM}(\chi)))(\pos{\phi,atomic})
\end{align*}
\item[$\gamma=\neg\phi$]
\begin{align*}
&(\Sen^{\FOL^{{\relax}}_{\alpha}(\INS)}(\chi)\circ\Sen^{\FOL^{{\relax}}_{\alpha}(\INM)}(\Sigma))(\neg\phi)
=\Sen^{\FOL^{{\relax}}_{\alpha}(\INS)}(\chi)(\neg\Sen^{\FOL^{{\relax}}_{\alpha}(\INM)}(\Sigma)(\phi)) \\
&=\neg\Sen^{\FOL^{{\relax}}_{\alpha}(\INS)}(\chi)(\Sen^{\FOL^{{\relax}}_{\alpha}(\INM)}(\Sigma)(\phi))
=\neg(\Sen^{\FOL^{{\relax}}_{\alpha}(\INS)}(\chi)\circ\Sen^{\FOL^{{\relax}}_{\alpha}(\INM)}(\Sigma))(\phi) \\
&\stackrel{IH}{=}\neg(\Sen^{\FOL^{{\relax}}_{\alpha}(\INM)}(\Sigma')\circ\Sen^{\FOL^{{\relax}}_{\alpha}(\INS')}(\Sig^{\INM}(\chi)))(\phi)
=\neg\Sen^{\FOL^{{\relax}}_{\alpha}(\INM)}(\Sigma')(\Sen^{\FOL^{{\relax}}_{\alpha}(\INS')}(\Sig^{\INM}(\chi))(\phi)) \\
&=\Sen^{\FOL^{{\relax}}_{\alpha}(\INM)}(\Sigma')(\neg\Sen^{\FOL^{{\relax}}_{\alpha}(\INS')}(\Sig^{\INM}(\chi))(\phi))
=(\Sen^{\FOL^{{\relax}}_{\alpha}(\INM)}(\Sigma')\circ\Sen^{\FOL^{{\relax}}_{\alpha}(\INS')}(\Sig^{\INM}(\chi)))(\neg\phi)
\end{align*}
\item[$\gamma=\Exists{X'}\phi$]
\begin{align*}
&(\Sen^{\FOL^{{\relax}}_{\alpha}(\INS)}(\chi)
\circ\Sen^{\FOL^{{\relax}}_{\alpha}(\INM)}(\Sigma))
(\Exists{X'}\phi)\\
&=\Sen^{\FOL^{{\relax}}_{\alpha}(\INS)}(\chi)
(\Exists{\Sigma_{\Dex^{\INM}}(X')}(\Sen^{\FOL^{{\relax}}_{\alpha}(\INM)}*\Sigma^{\Dex^{\INM}})(X')(\phi))\\
&=\Exists{\chi_{\Dex^{\INS}}(\Sigma_{\Dex^{\INM}}(X'))}
\Sen^{\FOL^{{\relax}}_{\alpha}(\INS)}(\chi^{\Dex^{\INS}}[\Sigma_{\Dex^{\INM}}(X')])
((\Sen^{\FOL^{{\relax}}_{\alpha}(\INM)}*\Sigma^{\Dex^{\INM}})(X')(\phi))\\
&=\Exists{(\chi_{\Dex^{\INS}}\circ\Sigma_{\Dex^{\INM}})(X')}
(\Sen^{\FOL^{{\relax}}_{\alpha}(\INS)}(\chi^{\Dex^{\INS}}[\Sigma_{\Dex^{\INM}}(X')])
\circ(\Sen^{\FOL^{{\relax}}_{\alpha}(\INM)}*\Sigma^{\Dex^{\INM}})(X'))(\phi)\\
&=\Exists{f(X')}F(\phi)\\
&\text{ where }f=\chi_{\Dex^{\INS}}\circ\Sigma_{\Dex^{\INM}}
\text{ and }F=\Sen^{\FOL^{{\relax}}_{\alpha}(\INS)}(\chi^{\Dex^{\INS}}[\Sigma_{\Dex^{\INM}}(X')])
\circ(\Sen^{\FOL^{{\relax}}_{\alpha}(\INM)}*\Sigma^{\Dex^{\INM}})(X')
&\\
&(\Sen^{\FOL^{{\relax}}_{\alpha}(\INM)}(\Sigma')
\circ\Sen^{\FOL^{{\relax}}_{\alpha}(\INS')}(\Sig^{\INM}(\chi)))
(\Exists{X'}\phi)\\
&=\Sen^{\FOL^{{\relax}}_{\alpha}(\INM)}(\Sigma')
(
\Exists{\Sig^{\INM}(\chi)_{\Dex^{\INS'}}(X')}
\Sen^{\FOL^{{\relax}}_{\alpha}(\INS')}(\Sig^{\INM}(\chi)^{\Dex^{\INS'}}[X'])
(\phi)
)\\
&=
\Exists{\Sigma'_{\Dex^{\INM}}(\Sig^{\INM}(\chi)_{\Dex^{\INS'}}(X'))}
(\Sen^{\FOL^{{\relax}}_{\alpha}(\INM)}*\Sigma'^{\Dex^{\INM}})(\Sig^{\INM}(\chi)_{\Dex^{\INS'}}(X'))
(
\Sen^{\FOL^{{\relax}}_{\alpha}(\INS')}(\Sig^{\INM}(\chi)^{\Dex^{\INS'}}[X'])
(\phi)
)\\
&=
\Exists{(\Sigma'_{\Dex^{\INM}}\circ\Sig^{\INM}(\chi)_{\Dex^{\INS'}})(X')}
((\Sen^{\FOL^{{\relax}}_{\alpha}(\INM)}*\Sigma'^{\Dex^{\INM}})(\Sig^{\INM}(\chi)_{\Dex^{\INS'}}(X'))
\circ\Sen^{\FOL^{{\relax}}_{\alpha}(\INS')}(\Sig^{\INM}(\chi)^{\Dex^{\INS'}}[X']))
(\phi)\\
&=\Exists{g(X')}G(\phi)\\
&\text{ where }g=\Sigma'_{\Dex^{\INM}}\circ\Sig^{\INM}(\chi)_{\Dex^{\INS'}}
\text{ and }G=(\Sen^{\FOL^{{\relax}}_{\alpha}(\INM)}*\Sigma'^{\Dex^{\INM}})(\Sig^{\INM}(\chi)_{\Dex^{\INS'}}(X'))
\circ\Sen^{\FOL^{{\relax}}_{\alpha}(\INS')}(\Sig^{\INM}(\chi)^{\Dex^{\INS'}}[X'])
\end{align*}
By the definition of $(\cdot)_{\Dex^{\INM}}$, $f=g$.
So we verify that $F(\phi)=G(\phi)$.
Let
$X:=\Sigma_{\Dex^{\INM}}(X')$,
$(\overline{\chi}:\overline{\Sigma}\to\overline{\Sigma'}):=(\chi^{\Dex^{\INS}}[X]:\Sigma^{\Dex^{\INS}}[X]\to\Sigma'^{\Dex^{\INS}}[\chi_{\Dex^{\INS}}(X)])$,
and $\overline{\phi}:=\Sen^{\FOL^{{\relax}}_{\alpha}(\INS')}(\Sigma^{\Dex^{\INM}}[X'])(\phi)$.
\begin{align*}
F(\phi)
&=(\Sen^{\FOL^{{\relax}}_{\alpha}(\INS)}(\chi^{\Dex^{\INS}}[\Sigma_{\Dex^{\INM}}(X')])
\circ(\Sen^{\FOL^{{\relax}}_{\alpha}(\INM)}*\Sigma^{\Dex^{\INM}})(X'))
(\phi)\\
&=(\Sen^{\FOL^{{\relax}}_{\alpha}(\INS)}(\chi^{\Dex^{\INS}}[\Sigma_{\Dex^{\INM}}(X')])
\circ\Sen^{\FOL^{{\relax}}_{\alpha}(\INM)}(\Sigma^{\Dex^{\INS}}[\Sigma_{\Dex^{\INM}}(X')])
\circ\Sen^{\FOL^{{\relax}}_{\alpha}(\INS')}(\Sigma^{\Dex^{\INM}}[X']))
(\phi)\\
&=(\Sen^{\FOL^{{\relax}}_{\alpha}(\INS)}(\overline{\chi})
\circ\Sen^{\FOL^{{\relax}}_{\alpha}(\INM)}(\overline{\Sigma}))
(\overline{\phi})
\stackrel{IH}{=}(\Sen^{\FOL^{{\relax}}_{\alpha}(\INM)}(\overline{\Sigma'})
\circ\Sen^{\FOL^{{\relax}}_{\alpha}(\INS')}(\Sig^{\INM}(\overline{\chi})))
(\overline{\phi})\\
&\\
G(\phi)
&=((\Sen^{\FOL^{{\relax}}_{\alpha}(\INM)}*\Sigma'^{\Dex^{\INM}})(\Sig^{\INM}(\chi)_{\Dex^{\INS'}}(X'))
\circ\Sen^{\FOL^{{\relax}}_{\alpha}(\INS')}(\Sig^{\INM}(\chi)^{\Dex^{\INS'}}[X']))
(\phi)\\
&=(\Sen^{\FOL^{{\relax}}_{\alpha}(\INM)}(\Sigma'^{\Dex^{\INS}}[\Sigma'_{\Dex^{\INM}}(\Sig^{\INM}(\chi)_{\Dex^{\INS'}}(X'))])
\circ\Sen^{\FOL^{{\relax}}_{\alpha}(\INS')}(\Sigma'^{\Dex^{\INM}}[\Sig^{\INM}(\chi)_{\Dex^{\INS'}}(X')])\\
&\circ\Sen^{\FOL^{{\relax}}_{\alpha}(\INS')}(\Sig^{\INM}(\chi)^{\Dex^{\INS'}}[X']))
(\phi)\\
&=(\Sen^{\FOL^{{\relax}}_{\alpha}(\INM)}(\Sigma'^{\Dex^{\INS}}[\Sigma'_{\Dex^{\INM}}(\Sig^{\INM}(\chi)_{\Dex^{\INS'}}(X'))])\\
&\circ\Sen^{\FOL^{{\relax}}_{\alpha}(\INS')}
(\Sigma'^{\Dex^{\INM}}[\Sig^{\INM}(\chi)_{\Dex^{\INS'}}(X')]
\circ\Sig^{\INM}(\chi)^{\Dex^{\INS'}}[X']))
(\phi)\\
&\stackrel{\ref{eq:MDex-3}}{=}(\Sen^{\FOL^{{\relax}}_{\alpha}(\INM)}(\Sigma'^{\Dex^{\INS}}[\Sigma'_{\Dex^{\INM}}(\Sig^{\INM}(\chi)_{\Dex^{\INS'}}(X'))])
\circ\Sen^{\FOL^{{\relax}}_{\alpha}(\INS')}(\Sig^{\INM}(\chi^{\Dex^{\INS}}[\Sigma_{\Dex^{\INM}}(X')])\circ\Sigma^{\Dex^{\INM}}[X']))
(\phi)\\
&\text{since }(\cdot)_{\Dex^{\INM}}\text{ natural,}\\
&=(
\Sen^{\FOL^{{\relax}}_{\alpha}(\INM)}(
\Sigma'^{\Dex^{\INS}}[
\chi_{\Dex^{\INS}}(\Sigma_{\Dex^{\INM}}(X'))
])
\circ\Sen^{\FOL^{{\relax}}_{\alpha}(\INS')}(\Sig^{\INM}(\chi^{\Dex^{\INS}}[\Sigma_{\Dex^{\INM}}(X')])\circ\Sigma^{\Dex^{\INM}}[X']))
(\phi)\\
&=(\Sen^{\FOL^{{\relax}}_{\alpha}(\INM)}(\overline{\Sigma'})
\circ\Sen^{\FOL^{{\relax}}_{\alpha}(\INS')}(\Sig^{\INM}(\overline{\chi})))
(\overline{\phi})
\end{align*}
Therefore, $F(\phi)=G(\phi)$.
\end{proofcases}
\end{proofcases}
\item[$\A\models^{\FOL^{{\relax}}_{\alpha}(\INS)}\Sen^{\FOL^{{\relax}}_{\alpha}(\INM)}(\Sigma)(\gamma')$ iff $\Mod^{\FOL^{{\relax}}_{\alpha}(\INM)}(\Sigma)(\A)\models^{\FOL^{{\relax}}_{\alpha}(\INS')}\gamma'$]
We show that the following $P(\gamma')$ holds true for each sentence $\gamma'$ by induction based on sentence complexity.
($P(\gamma')\Leftrightarrow$
$\A\models^{\FOL^{{\relax}}_{\alpha}(\INS)}\Sen^{\FOL^{{\relax}}_{\alpha}(\INM)}(\Sigma)(\gamma')$ iff $\Mod^{\FOL^{{\relax}}_{\alpha}(\INM)}(\Sigma)(\A)\models^{\FOL^{{\relax}}_{\alpha}(\INS')}\gamma'$
holds true for each $\Sigma\in|\Sig^{\INS}|$ and $\A\in|\Mod^{\FOL^{{\relax}}_{\alpha}(\INS)}(\Sigma)|$
such that $\gamma'\in\Sen^{\FOL^{{\relax}}_{\alpha}(\INS')}(\Sig^{\INM}(\Sigma))$)
The same applies to the other cases, so we will only check three.
\begin{proofcases}
\item[$\gamma'=\pos{\phi,atomic}$]
\begin{align*}
&\A\models^{\FOL^{{\relax}}_{\alpha}(\INS)}\Sen^{\FOL^{{\relax}}_{\alpha}(\INM)}(\Sigma)(\pos{\phi,atomic})
\Leftrightarrow\A\models^{\INS}\Sen^{\INM}(\Sigma)(\phi)\\
&\Leftrightarrow\Mod^{\INM}(\Sigma)(\A)\models^{\INS'}\phi
\Leftrightarrow\Mod^{\FOL^{{\relax}}_{\alpha}(\INM)}(\Sigma)(\A)\models^{\FOL^{{\relax}}_{\alpha}(\INS')}\pos{\phi,atomic}
\end{align*}
\item[$\gamma'=\neg\phi$]
\begin{align*}
&\A\models^{\FOL^{{\relax}}_{\alpha}(\INS)}\Sen^{\FOL^{{\relax}}_{\alpha}(\INM)}(\Sigma)(\neg\phi)
\Leftrightarrow\A\not\models^{\FOL^{{\relax}}_{\alpha}(\INS)}\Sen^{\FOL^{{\relax}}_{\alpha}(\INM)}(\Sigma)(\phi)\\
&\stackrel{IH}{\Leftrightarrow}\Mod^{\FOL^{{\relax}}_{\alpha}(\INM)}(\Sigma)(\A)\not\models^{\FOL^{{\relax}}_{\alpha}(\INS')}\phi
\Leftrightarrow\Mod^{\FOL^{{\relax}}_{\alpha}(\INM)}(\Sigma)(\A)\models^{\FOL^{{\relax}}_{\alpha}(\INS')}\neg\phi
\end{align*}
\item[$\gamma'=\Exists{X'}\phi$]
\begin{align*}
&\A\models^{\FOL^{{\relax}}_{\alpha}(\INS)}\Sen^{\FOL^{{\relax}}_{\alpha}(\INM)}(\Sigma)(\Exists{X'}\phi)\\
&\Leftrightarrow\A\models^{\FOL^{{\relax}}_{\alpha}(\INS)}
\Exists{\Sigma_{\Dex^{\INM}}(X')}(\Sen^{\FOL^{{\relax}}_{\alpha}(\INM)}*\Sigma^{\Dex^{\INM}})(X')(\phi)\\
&\Leftrightarrow\hat{\A}\models^{\FOL^{{\relax}}_{\alpha}(\INS)}
(\Sen^{\FOL^{{\relax}}_{\alpha}(\INM)}*\Sigma^{\Dex^{\INM}})(X')(\phi)
\text{ for some }\Sigma^{\Dex^{\INS}}(\Sigma_{\Dex^{\INM}}(X'))\text{-expansion }\hat{\A}\text{ of }\A\\
&\Space\Leftrightarrow\hat{\A}\models^{\FOL^{{\relax}}_{\alpha}(\INS)}
\Sen^{\FOL^{{\relax}}_{\alpha}(\INM)}(\Sigma^{\Dex^{\INS}}[\Sigma_{\Dex^{\INM}}(X')])
(\Sen^{\FOL^{{\relax}}_{\alpha}(\INS')}(\Sigma^{\Dex^{\INM}}[X'])
(\phi))
\text{ for some... }\\
&\Space\stackrel{IH}{\Leftrightarrow}
\Mod^{\FOL^{{\relax}}_{\alpha}(\INM)}(\Sigma^{\Dex^{\INS}}[\Sigma_{\Dex^{\INM}}(X')])
(\hat{\A})
\models^{\FOL^{{\relax}}_{\alpha}(\INS)}
\Sen^{\FOL^{{\relax}}_{\alpha}(\INS')}(\Sigma^{\Dex^{\INM}}[X'])
(\phi)
\text{ for some... }\\
&\Space\Leftrightarrow
\Mod^{\FOL^{{\relax}}_{\alpha}(\INS')}(\Sigma^{\Dex^{\INM}}[X'])
(\Mod^{\FOL^{{\relax}}_{\alpha}(\INM)}(\Sigma^{\Dex^{\INS}}[\Sigma_{\Dex^{\INM}}(X')])
(\hat{\A}))
\models^{\FOL^{{\relax}}_{\alpha}(\INS)}
\phi
\text{ for some... }\\
&\Leftrightarrow(\Sigma^{\Dex^{\INM}}*\Mod^{\FOL^{{\relax}}_{\alpha}(\INM)})(X')(\hat{\A})\models^{\FOL^{{\relax}}_{\alpha}(\INS')}\phi
\text{ for some }\Sigma^{\Dex^{\INS}}(\Sigma_{\Dex^{\INM}}(X'))\text{-expansion }\hat{\A}\text{ of }\A\\
&\\
&\Mod^{\FOL^{{\relax}}_{\alpha}(\INM)}(\Sigma)(\A)\models^{\FOL^{{\relax}}_{\alpha}(\INS')}\Exists{X'}\phi\\
&\Leftrightarrow\hat{\A}'\models^{\FOL^{{\relax}}_{\alpha}(\INS')}\phi
\text{ for some }\Sig^{\INM}(\Sigma)^{\Dex^{\INS'}}(X')\text{-expansion }\hat{\A}'\text{ of }\A':=\Mod^{\INM}(\Sigma)(\A)\\
&\stackrel{*}{\Leftrightarrow}(\Sigma^{\Dex^{\INM}}*\Mod^{\FOL^{{\relax}}_{\alpha}(\INM)})(X')(\hat{\A})\models^{\FOL^{{\relax}}_{\alpha}(\INS')}\phi
\text{ for some }\Sigma^{\Dex^{\INS}}(\Sigma_{\Dex^{\INM}}(X'))\text{-expansion }\hat{\A}\text{ of }\A
\end{align*}
*This equivalence is clear from the last condition of the definition of $\Dex$-institution morphism and the commutativity of its diagrams.
\item[$\gamma'=\exists^{\kappa}{X'}\cdot\phi$ (extra)]
We will show that if $\INM$ has the uniqueness property, the satisfaction condition holds true even if we add $\exists^{\kappa}{X'}\cdot\phi$ (there are at least $\kappa$ instances).
In a similar way to above, we can show the following:
\begin{align*}
&\A\models^{\FOL^{{\relax}}_{\alpha}(\INS)}\Sen^{\FOL^{{\relax}}_{\alpha}(\INM)}(\Sigma)(\exists^{\kappa}{X'}\cdot\phi)\\
&\Leftrightarrow(\Sigma^{\Dex^{\INM}}*\Mod^{\FOL^{{\relax}}_{\alpha}(\INM)})(X')(\hat{\A})\models^{\FOL^{{\relax}}_{\alpha}(\INS')}\phi
\text{ for }\kappa\ \Sigma^{\Dex^{\INS}}(\Sigma_{\Dex^{\INM}}(X'))\text{-expansions }\hat{\A}\text{ of }\A\\
&\Leftrightarrow\card\{\hat{\A}\mid(\Sigma^{\Dex^{\INM}}*\Mod^{\INM})(X')(\hat{\A})\models^{\FOL^{{\relax}}_{\alpha}(\INS')}\phi\text{ and }\hat{\A}\text{ is a }\Sigma^{\Dex^{\INS}}(\Sigma_{\Dex^{\INM}}(X'))\text{-expansion}\text{ of }\A\}\geq\kappa
\\%
&\Mod^{\FOL^{{\relax}}_{\alpha}(\INM)}(\Sigma)(\A)\models^{\FOL^{{\relax}}_{\alpha}(\INS')}\exists^{\kappa}{X'}\cdot\phi\\
&\Leftrightarrow\hat{\A}'\models^{\FOL^{{\relax}}_{\alpha}(\INS')}\phi
\text{ for }\kappa\ \Sig^{\INM}(\Sigma)^{\Dex^{\INS'}}(X')\text{-expansions }\hat{\A}'\text{ of }\A':=\Mod^{\INM}(\Sigma)(\A)\\
&\Leftrightarrow\card\{\hat{\A}'\mid\hat{\A}'\models^{\FOL^{{\relax}}_{\alpha}(\INS')}\phi\text{ and }\hat{\A}'\text{ is a }\Sig^{\INM}(\Sigma)^{\Dex^{\INS'}}(X')\text{-expansion}\text{ of }\A':=\Mod^{\INM}(\Sigma)(\A)\}\geq\kappa
\end{align*}
From the unique existence condition, $(\Sigma^{\Dex^{\INM}}*\Mod^{\INM})(X')$ is a bijection between these two sets, so the two conditions coincide.
\end{proofcases}
\end{proofcases}
\end{proofcases}
\item[preserve operations]
Since no other parts have been changed, we only need to check the $\Sen$ parts.
\begin{proofcases}
\item[$\Sen^{\FOL^{{\relax}}_{\alpha}(\id_{\INS})}=\Sen^{\id_{\FOL^{{\relax}}_{\alpha}(\INS)}}$]
We show that the following $P(\gamma)$ holds true for each sentence $\gamma$ by induction based on sentence complexity.
($P(\gamma)\Leftrightarrow
\Sen^{\FOL^{{\relax}}_{\alpha}(\id_{\INS})}(\Sigma)(\gamma)=\Sen^{\id_{\FOL^{{\relax}}_{\alpha}(\INS)}}(\Sigma)(\gamma)
$
holds true for each $\Sigma\in|\Sig^{\INS}|$ such that $\gamma\in\Sen^{\FOL^{{\relax}}_{\alpha}(\INS)}(\Sig^{\id_{\INS}}(\Sigma))$.)
The same applies to the other cases, so we will only check three.
\begin{proofcases}
\item[$\gamma=\pos{\phi,atomic}$]
\begin{align*}
&\Sen^{\FOL^{{\relax}}_{\alpha}(\id_{\INS})}(\Sigma)(\pos{\phi,atomic})\\
&=\pos{\Sen^{\id_{\INS}}(\Sigma)(\phi),atomic}
=\pos{\id_{\Sen^{\INS}}(\Sigma)(\phi),atomic}
=\pos{\id_{\Sen^{\INS}(\Sigma)}(\phi),atomic}\\
&=\pos{\phi,atomic}\\
&=\id_{\Sen^{\FOL^{{\relax}}_{\alpha}(\INS)}(\Sigma)}(\pos{\phi,atomic})
=\id_{\Sen^{\FOL^{{\relax}}_{\alpha}(\INS)}}(\Sigma)(\pos{\phi,atomic})
=\Sen^{\id_{\FOL^{{\relax}}_{\alpha}(\INS)}}(\Sigma)(\pos{\phi,atomic})
\end{align*}
\item[$\gamma=\neg\phi$]
\begin{align*}
&\Sen^{\FOL^{{\relax}}_{\alpha}(\id_{\INS})}(\Sigma)(\neg\phi)
=\neg\Sen^{\FOL^{{\relax}}_{\alpha}(\id_{\INS})}(\Sigma)(\phi)\\
&\stackrel{IH}{=}\neg\Sen^{\id_{\FOL^{{\relax}}_{\alpha}(\INS)}}(\Sigma)(\phi)
=\neg\id_{\Sen^{\FOL^{{\relax}}_{\alpha}(\INS)}}(\Sigma)(\phi)
=\neg\id_{\Sen^{\FOL^{{\relax}}_{\alpha}(\INS)}(\Sigma)}(\phi)\\
&=\neg\phi\\
&=\id_{\Sen^{\FOL^{{\relax}}_{\alpha}(\INS)}(\Sigma)}(\neg\phi)
=\id_{\Sen^{\FOL^{{\relax}}_{\alpha}(\INS)}}(\Sigma)(\neg\phi)
=\Sen^{\id_{\FOL^{{\relax}}_{\alpha}(\INS)}}(\Sigma)(\neg\phi)
\end{align*}
\item[$\gamma=\Exists{X}\phi$]
\begin{align*}
&\Sen^{\FOL^{{\relax}}_{\alpha}(\id_{\INS})}(\Sigma)(\Exists{X}\phi)
=\Exists{\Sigma_{\Dex^{\FOL^{{\relax}}_{\alpha}(\id_{\INS})}}(X)}(\Sen^{\FOL^{{\relax}}_{\alpha}(\id_{\INS})}*\Sigma^{\Dex^{\FOL^{{\relax}}_{\alpha}(\id_{\INS})}})(X)(\phi)\\
&\stackrel{IH}{=}\Exists{\Sigma_{\Dex^{\FOL^{{\relax}}_{\alpha}(\id_{\INS})}}(X)}
(\Sen^{\id_{\FOL^{{\relax}}_{\alpha}(\INS)}}*\Sigma^{\Dex^{\FOL^{{\relax}}_{\alpha}(\id_{\INS})}})(X)
(\phi)\\
&\stackrel{Lemma~\ref{lemma:S*M}}{=}\Exists{\Sigma_{\Dex^{\FOL^{{\relax}}_{\alpha}(\id_{\INS})}}(X)}
\id_{\Sen^{\FOL^{{\relax}}_{\alpha}(\INS)}(\Sigma^{\Dex^{\INS}}[X])}
(\phi)\\
&=\Exists{\Sigma_{\Dex^{\id_{\INS}}}(X)}
\id_{\Sen^{\FOL^{{\relax}}_{\alpha}(\INS)}(\Sigma^{\Dex^{\INS}}[X])}
(\phi)
=\Exists{\id_{\Sigma_{\Dex^{\INS}}}(X)}
\id_{\Sen^{\FOL^{{\relax}}_{\alpha}(\INS)}(\Sigma^{\Dex^{\INS}}[X])}
(\phi)\\
&=\Exists{X}\phi\\
&=\id_{\Sen^{\FOL^{{\relax}}_{\alpha}(\INS)}(\Sigma)}(\Exists{X}\phi)
=\id_{\Sen^{\FOL^{{\relax}}_{\alpha}(\INS)}}(\Sigma)(\Exists{X}\phi)
=\Sen^{\id_{\FOL^{{\relax}}_{\alpha}(\INS)}}(\Sigma)(\Exists{X}\phi)
\end{align*}
\end{proofcases}
\item[$\Sen^{\FOL^{{\relax}}_{\alpha}(\INM'\circ\INM)}=\Sen^{\FOL^{{\relax}}_{\alpha}(\INM')\circ\FOL^{{\relax}}_{\alpha}(\INM)}$]
We show that the following $P(\gamma'')$ holds true for each sentence $\gamma''$ by induction based on sentence complexity.
($P(\gamma'')\Leftrightarrow
\Sen^{\FOL^{{\relax}}_{\alpha}(\INM'\circ\INM)}(\Sigma)(\gamma'')=\Sen^{\FOL^{{\relax}}_{\alpha}(\INM')\circ\FOL^{{\relax}}_{\alpha}(\INM)}(\Sigma)(\gamma'')
$
holds true for each $\Sigma\in|\Sig^{\INS}|$ such that $\gamma''\in\Sen^{\FOL^{{\relax}}_{\alpha}(\INS)}(\Sig^{\id_{\INS}}(\Sigma))$.)
The same applies to the other cases, so we will only check three.
\begin{proofcases}
\item[$\gamma''=\pos{\phi,atomic}$]
\begin{align*}
&\Sen^{\FOL^{{\relax}}_{\alpha}(\INM'\circ\INM)}(\Sigma)(\pos{\phi,atomic})\\
&=\pos{\Sen^{\INM'\circ\INM}(\Sigma)(\phi),atomic}\\
&\Space=\pos{
(\Sen^{\INM}\circ_{V}(\Sen^{\INM'}\circ_{H}\Sig^{\INM}))
(\Sigma)(\phi),atomic}\\
&\Space=\pos{
\Sen^{\INM}(\Sigma)
\circ
(\Sen^{\INM'}(\Sig^{\INM}(\Sigma)))
(\phi),atomic}\\
&=
\pos{\Sen^{\INM}(\Sigma)
(\Sen^{\INM'}(\Sig^{\INM}(\Sigma))
(\phi)),atomic}\\
&=
\Sen^{\INM}(\Sigma)
(\Sen^{\INM'}(\Sig^{\INM}(\Sigma))
(\pos{\phi,atomic}))\\
&\Space=
(\Sen^{\FOL^{{\relax}}_{\alpha}(\INM)}(\Sigma)
\circ\Sen^{\FOL^{{\relax}}_{\alpha}(\INM')}(\Sig^{\INM}(\Sigma)))
(\pos{\phi,atomic})\\
&\Space=
(\Sen^{\FOL^{{\relax}}_{\alpha}(\INM)}\circ_{V}(\Sen^{\FOL^{{\relax}}_{\alpha}(\INM')}\circ_{H}\Sig^{\INM}))(\Sigma)(\pos{\phi,atomic})\\
&=\Sen^{\FOL^{{\relax}}_{\alpha}(\INM')\circ\FOL^{{\relax}}_{\alpha}(\INM)}(\Sigma)(\pos{\phi,atomic})
\end{align*}
\item[$\gamma''=\neg\phi$]
\begin{align*}
&\Sen^{\FOL^{{\relax}}_{\alpha}(\INM'\circ\INM)}(\Sigma)(\neg\phi)
=\neg\Sen^{\FOL^{{\relax}}_{\alpha}(\INM'\circ\INM)}(\Sigma)(\phi)\\
&\stackrel{IH}{=}\neg\Sen^{\FOL^{{\relax}}_{\alpha}(\INM')\circ\FOL^{{\relax}}_{\alpha}(\INM)}(\Sigma)(\phi)\\
&\Space=\neg
(\Sen^{\FOL^{{\relax}}_{\alpha}(\INM)}\circ_{V}(\Sen^{\FOL^{{\relax}}_{\alpha}(\INM')}\circ_{H}\Sig^{\FOL^{{\relax}}_{\alpha}(\INM)}))
(\Sigma)(\phi)\\
&\Space=\neg
\Sen^{\FOL^{{\relax}}_{\alpha}(\INM)}(\Sigma)
\circ
(\Sen^{\FOL^{{\relax}}_{\alpha}(\INM')}(\Sig^{\FOL^{{\relax}}_{\alpha}(\INM)}(\Sigma)))
(\phi)\\
&=
\neg\Sen^{\FOL^{{\relax}}_{\alpha}(\INM)}(\Sigma)
(\Sen^{\FOL^{{\relax}}_{\alpha}(\INM')}(\Sig^{\FOL^{{\relax}}_{\alpha}(\INM)}(\Sigma))
(\phi))\\
&=
\Sen^{\FOL^{{\relax}}_{\alpha}(\INM)}(\Sigma)
(\Sen^{\FOL^{{\relax}}_{\alpha}(\INM')}(\Sig^{\FOL^{{\relax}}_{\alpha}(\INM)}(\Sigma))
(\neg\phi))\\
&\Space=
(\Sen^{\FOL^{{\relax}}_{\alpha}(\INM)}(\Sigma)
\circ\Sen^{\FOL^{{\relax}}_{\alpha}(\INM')}(\Sig^{\INM}(\Sigma)))
(\neg\phi)\\
&\Space=
(\Sen^{\FOL^{{\relax}}_{\alpha}(\INM)}\circ_{V}(\Sen^{\FOL^{{\relax}}_{\alpha}(\INM')}\circ_{H}\Sig^{\INM}))(\Sigma)(\neg\phi)\\
&=\Sen^{\FOL^{{\relax}}_{\alpha}(\INM')\circ\FOL^{{\relax}}_{\alpha}(\INM)}(\Sigma)(\neg\phi)
\end{align*}
\item[$\gamma''=\Exists{X''}\phi$]
\begin{align*}
&\Sen^{\FOL^{{\relax}}_{\alpha}(\INM'\circ\INM)}(\Sigma)(\Exists{X''}\phi)
=\Exists{\Sigma_{\Dex^{\FOL^{{\relax}}_{\alpha}(\INM'\circ\INM)}}(X'')}
(\Sen^{\FOL^{{\relax}}_{\alpha}(\INM'\circ\INM)}*\Sigma^{\Dex^{\FOL^{{\relax}}_{\alpha}(\INM'\circ\INM)}})(X'')
(\phi)\\
&=\Exists{\Sigma_{\Dex^{\FOL^{{\relax}}_{\alpha}(\INM'\circ\INM)}}(X'')}
(\Sen^{\FOL^{{\relax}}_{\alpha}(\INM'\circ\INM)}*\Sigma^{\Dex^{\FOL^{{\relax}}_{\alpha}(\INM')\circ\FOL^{{\relax}}_{\alpha}(\INM)}})(X'')
(\phi)\\
&\stackrel{IH}{=}\Exists{\Sigma_{\Dex^{\FOL^{{\relax}}_{\alpha}(\INM'\circ\INM)}}(X'')}
(\Sen^{\FOL^{{\relax}}_{\alpha}(\INM')\circ\FOL^{{\relax}}_{\alpha}(\INM)}*\Sigma^{\Dex^{\FOL^{{\relax}}_{\alpha}(\INM'\circ\INM)}})(X'')
(\phi)\\
&\stackrel{Lemma~\ref{lemma:S*M}}{=}\Exists{\Sigma_{\Dex^{\FOL^{{\relax}}_{\alpha}(\INM'\circ\INM)}}(X'')}
((\Sen^{\FOL^{{\relax}}_{\alpha}(\INM)}*\Sigma^{\Dex^{\INM}})(X')\circ(\Sen^{\FOL^{{\relax}}_{\alpha}(\INM')}*\Sigma'^{\Dex^{\INM'}})(X''))
(\phi)\\
&=\Exists{\Sigma_{\Dex^{\INM}}(X')}
((\Sen^{\FOL^{{\relax}}_{\alpha}(\INM)}*\Sigma^{\Dex^{\INM}})(X')\circ(\Sen^{\FOL^{{\relax}}_{\alpha}(\INM')}*\Sigma'^{\Dex^{\INM'}})(X''))
(\phi)\\
&=\Sen^{\FOL^{{\relax}}_{\alpha}(\INM)}(\Sigma)
(\Exists{\Sigma'_{\Dex^{\INM'}}(X'')}(\Sen^{\FOL^{{\relax}}_{\alpha}(\INM')}*\Sigma'^{\Dex^{\INM'}})(X'')
(\phi))\\
&=\Sen^{\FOL^{{\relax}}_{\alpha}(\INM)}(\Sigma)
(\Sen^{\FOL^{{\relax}}_{\alpha}(\INM')}(\Sigma')
(\Exists{X''}\phi))\\
&=\Sen^{\FOL^{{\relax}}_{\alpha}(\INM')\circ\FOL^{{\relax}}_{\alpha}(\INM)}(\Sigma)
(\Exists{X''}\phi)
\end{align*}
\end{proofcases}
\end{proofcases}
\end{proofcases}
\end{proof}

\begin{proof}[Proof of Lemma~\ref{lemma:translation_of_substitution}]\
\begin{proofcases}
\item[$\INS$]
Since $\theta$ is a substitution
$(\chi\circ\theta)\circ\Sigma^{\Dex^{\INS}}(X)=\chi\circ(\theta\circ\Sigma^{\Dex^{\INS}}(X))=\chi=(\chi/\Sig^{\INS})(\id_{\Sigma'})$.
Therefore $(\chi\circ\theta):\Sigma^{\Dex^{\INS}}(X)\to(\chi/\Sig^{\INS})(\id_{\Sigma'})\in|\Sigma/\Sig^{\INS}|$.
By the properties of universal morphisms, there exists a unique $\theta':\Sigma'^{\Dex^{\INS}}(\chi_{\Dex^{\INS}}(X))\to\id_{\Sigma'}$ such that $(\chi/\Sig^{\INS})(\theta')\circ\chi^{\Dex^{\INS}}(X)=\chi\circ\theta$.
In other words, there exists a unique $\Sigma'^{\Dex^{\INS}}(\chi_{\Dex^{\INS}}(X))$-substitution $\theta'$ such that $\theta'\circ\chi^{\Dex^{\INS}}[X]=\chi\circ\theta$.
\item[$\INM$]
Since $\theta$ is a substitution
$\theta\circ\Sig^{\INM}(\Sigma)^{\Dex^{\INS'}}(X)=\id_{\Sig^{\INM}(\Sigma)}=(\Sigma/\Sig^{\INM})(\id_{\Sigma})$.
Therefore $\theta:\Sig^{\INM}(\Sigma)^{\Dex^{\INS'}}(X)\to(\Sigma/\Sig^{\INM})(\id_{\Sigma})\in|\Sig^{\INM}(\Sigma)/\Sig^{\INS'}|$.
By the properties of universal morphisms, there exists a unique $\theta':\Sigma^{\Dex^{\INS}}(\Sigma_{\Dex^{\INM}}(X))\to\id_{\Sigma}$ such that $(\Sigma/\Sig^{\INM})(\theta')\circ\Sigma^{\Dex^{\INM}}(X)=\theta$.
In other words, there exists a unique $\Sigma^{\Dex^{\INS}}(\Sigma_{\Dex^{\INM}}(X))$-substitution $\theta'$ such that $\Sig^{\INM}(\theta')\circ\Sigma^{\Dex^{\INM}}[X]=\theta$.
\end{proofcases}
\end{proof}

\section{Proof of the results presented in Section~\ref{sec:Pc}}

\begin{proof}[Proof of well-definedness of Definition~\ref{def:Semantic-sequent}]
We confirm that this satisfies the structural rules.
\begin{proofcases}
\item[$\Modify$] Let $\A'\in|\Mod^{\INS}(\Sigma')|$.
Assuming the premise holds true, $\Mod^{\INS}(\chi)(\A')$ cannot satisfies $\Gamma\cup\{\neg\delta\mid\delta\in\Delta\}$.
Also by satisfaction condition, $\A'$ cannot satisfies $\Sen^{\INS}(\chi)(\Gamma)\cup\{\neg\delta\mid\delta\in\Sen^{\INS}(\chi)(\Delta)\}$.
So $\A'$ cannot satisfies $\Gamma'\cup\{\neg\delta\mid\delta\in\Delta'\}$.
Since the choice of $\A'$ was arbitrary, this means that the conclusion holds true.
\item[$\Init$] Because $\A\models\phi$ and $\A\not\models\phi$ cannot be true at the same time.
\item[$\Cut$] We show that if the conclusion is not true, then one of the premises must not be true.
From the assumption, there exists a model $\A$ that satisfies $\Gamma\cup\Gamma'$ and completely negates $\Delta\cup\Delta'$.
This satisfies either $\A\models\psi$ or $\A\not\models\psi$, so one of the premises cannot be held.
\end{proofcases}
\end{proof}

\begin{proof}[Proof of well-definedness of Definition~\ref{def:Syntactic-sequent}]
Strictly speaking, proof trees are defined recursively in the form
\[\pos{\text{Rule name}^{\{\text{Information about choice}\}},\proofrule{\text{Sequence of premise proof trees}}{\text{Root}},\text{Applied part}}.\]
\begin{itemize}
\item In the rule display, only the root portion of "Sequence of premises proof trees" is shown.
\item
The "Applied part" is $\proofrule{\pos{\Gamma'_{i}\vdash_{\Sigma_{i}}\Delta'_{i}}_{i\in I}}{\Gamma'\vdash_{\Sigma}\Delta'}$
for rules written in the form $\proofrule{\pos{\UL{\Gamma'_{i}}\vdash_{\Gamma_{i}\mid\Sigma_{i}\mid\Delta_{i}}\UL{\Delta'_{i}}}_{i\in I}}{\OL{\Gamma'}\vdash_{\Gamma\mid\Sigma\mid\Delta}\OL{\Delta'}}$.
\\
Note that in this paper, formally, sequent is defined using a set rather than a sequence of sentences, so it may be difficult to determine what the rule was applied to from the surrounding structure.
This is redundant in notation, so it will generally be omitted.
\item
"Information about choice"
If it exists, it is written as a subscript to the upper right of the rule.
This is information to avoid non-deterministic choices in recursive operations on proof trees.
(For example, note that if $\chi(\phi)=\chi(\phi')=\psi$, then the translation of $\phi\vee\phi'$ is $\psi\vee\psi$.)
\end{itemize}
{\par}
We prove that the proof is well-defined (by showing that it satisfies three axioms).
We verify that the proof can be constructed recursively.
To avoid complexity, $\Sen(\chi)$ is sometimes simply written as $\chi$.
\begin{proofcases}
\item[$\Modify$]
For proof tree $T$,
signature morphism $\chi:\Sigma\to\Sigma'$ whose domain is the root signature of $T$, and
sentence sets $\hat{\Gamma}',\hat{\Delta}'\subseteq\Sen(\Sigma')$,
such that $\hat{\Gamma}'\supseteq\chi(\hat{\Gamma})$, $\hat{\Delta}'\supseteq\chi(\hat{\Delta})$ hold for root $\hat{\Gamma}\vdash_{\Sigma}\hat{\Delta}$ of $T$,
we define new proof tree $\Modify(\chi,T,\hat{\Gamma}',\hat{\Delta}')$ for $\hat{\Gamma}'\vdash_{\Sigma'}\hat{\Delta}'$ recursively with respect to the structure of the proofs.%
\footnote{
If we represent all signatures obtained by repeatedly adding variables to $\Sigma$ as $\Dex_{\Sigma}$, this becomes a set.
Therefore, all proof trees whose root signature is an element of $\Dex_{\Sigma}$ also become a set.
Strictly speaking, recursive definitions are only possible if the order satisfies set-likeness, so using a function that assigns a structure to perform recursion is not allowed.
However, in practice, we only use proof information that has $\Dex_{\Sigma}$ as its root, so we can consider set-likeness to be guaranteed.
}
{\par}
Generally, proof rules take the following form:
\[name^{choice}~\proofrule
{\pos{\UL{\Phi_{i}}\vdash_{\iota_{i}(\Gamma_{i})\mid\Sigma_{i}\mid\iota_{i}(\Delta_{i})}\UL{\Psi_{i}}}_{i\in I}}
{\OL{\Phi}\vdash_{\cup_{i\in I}\Gamma_{i}\mid\Sigma\mid\cup_{i\in I}\Delta_{i}}\OL{\Psi}}\]
where $\hat{\Gamma}=\cup_{i\in I}\Gamma_{i}\cup\Phi$, $\hat{\Delta}=\cup_{i\in I}\Delta_{i}\cup\Psi$.
\begin{align*}\hspace*{-1cm}
\Modify(\chi,
  &\pos{name^{choice},
  \proofrule{\pos{T_{i}}_{i\in I}}{\hat{\Gamma}\vdash_{\Sigma}\hat{\Delta}},
  \proofrule{\pos{\Phi_{i}\vdash_{\Sigma_{i}}\Psi_{i}}_{i\in I}}{\Phi\vdash_{\Sigma}\Psi}},
\hat{\Gamma}',\hat{\Delta}')&\\
:=&\pos{name^{choice'},
\proofrule{\pos{\Modify(\chi_{i},T_{i},\hat{\Gamma}'_{i},\hat{\Delta}'_{i})}_{i\in I}}{\hat{\Gamma}'\vdash_{\Sigma'}\hat{\Delta}'},
\proofrule{\pos{\chi_{i}(\Phi_{i})\vdash_{\Sigma'_{i}}\chi_{i}(\Psi_{i})}_{i\in I}}{\chi(\Phi)\vdash_{\Sigma'}\chi(\Psi)}}&\\
&\chi_{i}:\Sigma_{i}\to\Sigma'_{i}:=
\begin{cases}
\chi^{\Dex}[X]:\Sigma^{\Dex}[X]\to\Sigma'^{\Dex}[X']\,(X':=\chi_{\Dex}(X))&(\text{when the rule has variable adding})\\
\chi:\Sigma\to\Sigma'&(\text{else})
\end{cases}
&\\
&\pos{\hat{\Gamma}'_{i},\hat{\Delta}'_{i}}:=
\begin{cases}
\pos{\chi_{i}(\Phi_{i})\,\iota'_{i}(\hat{\Gamma}'),\iota'_{i}(\hat{\Delta}')\,\chi_{i}(\Psi_{i})}\,(\iota'_{i}:=\Sigma'^{\Dex}(\chi_{\Dex}(X)))&(\text{when the rule has variable adding})\\
\pos{\chi_{i}(\Phi_{i})\,\hat{\Gamma}',\hat{\Delta}'\,\chi_{i}(\Psi_{i})}&(\text{else})
\end{cases}
&\\
&choice':=
\begin{cases}
n&(choice=n)\\
\theta'&(choice=\theta:\Sigma^{\Dex}(X)\text{-substitution})
\end{cases}
&\\
&(\theta':\text{ unique }\Sigma'^{\Dex}(\chi_{\Dex}(X))\text{-substitution which satisfies }\chi\circ\theta=\theta'\circ\chi_{\Dex}[X])& 
\end{align*}
The process is translating the proof diagram and adding extra sentences to each node. 
Since the rule structure is preserved, the result obtained by applying $\Modify$ is also a proof tree.
\item[$\Init$]
We will see that the function $\Init$ can be recursively defined 
which constructs a new proof tree $\Init(\Gamma,\Delta,\Sigma,\psi)$ for a sentence $\psi$, a signature $\Sigma$ such that $\psi\in\Sen(\Sigma)$ holds true, and a sentence set $\Gamma,\Delta\subseteq\Sen(\Sigma)$.
We will consider different cases based on the sentence's form. Since others are similar, we will only discuss some of them.
\begin{proofcases}
\item[$\psi$ is atomic]
Since $\psi\vDash^{\INS}_{\Sigma}\psi$, it can be defined by the following diagram.
\begin{align*}
\Atom~\proofrule{}{\,\psi\vdash_{\Gamma\mid\Sigma\mid\Delta}\psi\,}
\end{align*}
\item[$\psi=\neg\psi'$]
We draw it as a diagram.
\[
{\neg}_{R}~\proofrule
{
{\neg}_{L}~\proofrule{\proofrulet{:}{\,\psi'\vdash_{\Gamma\mid\Sigma\mid\Delta}\psi'\,}}
{\,\neg\psi'\,\psi'\vdash_{\Gamma\mid\Sigma\mid\Delta}\,}
}
{\,\neg\psi'\vdash_{\Gamma\mid\Sigma\mid\Delta}\neg\psi'\,}
\]
The part above $\psi'\vdash_{\Gamma\mid\Sigma\mid\Delta}\psi'$ is already defined by $\Init(\Gamma,\Delta,\Sigma,\psi')$.
\item[$\psi=\vee\psi'$]
We draw it as a diagram.
\[
{\vee}_{L}~\proofrule
{\pos{{\vee}_{R}^{n}~\proofrule{\proofrulet{:}{\,\psi'_{n}\vdash_{\Gamma\mid\Sigma\mid\Delta}\psi'_{n}\,}}
{\,\psi'_{n}\vdash_{\Gamma\mid\Sigma\mid\Delta}\vee\psi'\,}}_{n\in N}}
{\,\vee\psi'\vdash_{\Gamma\mid\Sigma\mid\Delta}\vee\psi'\,}
\]
The part above $\psi'\,\Gamma\vdash_{\Sigma}\Delta\,\psi'$ is already defined by $\Init(\Gamma,\Delta,\Sigma,\psi')$.
This argument is independent of the size of $\psi'$. 
\item[$\psi=\Exists{X}\psi'$]
Let $X'$ be the result of translating the contents of $X$ using $\Sigma^{\Dex}(X):\Sigma\to\Sigma^{\Dex}[X]$.
We have $\Sigma^{\Dex}(X)^{\Dex}[X]:\Sigma^{\Dex}[X]\to\Sigma^{\Dex}[X]^{\Dex}[X']$. 
Let $\theta$ be the substitution that "maps $X'$ to $X$", which we dealt with in Example~\ref{exa:substitution}.
Note that $\theta(\Sigma^{\Dex}(X)^{\Dex}[X](\psi'))=\psi'$,
\[
\exists_{L}~\proofrule
{\proofrule
{\exists_{R}^{\theta}~\proofrule
 {\proofrulet{:}{\psi'\vdash_{\Sigma^{\Dex}(X)(\Gamma)\mid\Sigma^{\Dex}[X]\mid\Sigma^{\Dex}(X)(\Delta)}\psi'\,(=\theta(\Sigma^{\Dex}(X)^{\Dex}[X](\psi')))}}
 {\psi'\vdash_{\Sigma^{\Dex}(X)(\Gamma)\mid\Sigma^{\Dex}[X]\mid\Sigma^{\Dex}(X)(\Delta)}\Exists{X'}\Sigma^{\Dex}(X)^{\Dex}[X](\psi')}}
{\psi'\vdash_{\Sigma^{\Dex}(X)(\Gamma)\mid\Sigma^{\Dex}[X]\mid\Sigma^{\Dex}(X)(\Delta)}\Sigma^{\Dex}(X)(\Exists{X}\psi')}}
{\Exists{X}\psi'\vdash_{\Gamma\mid\Sigma\mid\Delta}\Exists{X}\psi'}
\]
\normalsize
The part above $\psi'\,\Sigma(X)(\Gamma)\vdash_{\Sigma[X]}\Sigma(X)(\Delta)\,\psi'$ is already defined by $\Init(\Gamma,\Delta,\Sigma,\psi')$.
\end{proofcases}
\item[$\Cut$]
We will see that for two proof trees $T$ and $T'$ whose roots are $\Gamma\vdash_{\Sigma}\Delta\,\psi$ and $\psi\,\Gamma'\vdash_{\Sigma}\Delta'$ respectively,
a new proof tree $\Cut(\Gamma,\Gamma',\Delta,\Delta',\psi,T,T')$ whose root is $\Gamma\,\Gamma'\vdash_{\Sigma}\Delta\,\Delta'$
can be recursively defined for the complexity of $\psi$.
\begin{align*}
&\Cut^{\psi}~\proofrule
{\,\proofrulet{T}{\vdash_{\Gamma\mid\Sigma\mid\Delta}\UL{\psi}}\,\Space\,\proofrulet{T'}{\UL{\psi}\vdash_{\Gamma'\mid\Sigma\mid\Delta'}}\,}
{\,\vdash_{\Gamma\,\Gamma'\mid\Sigma\mid\Delta\,\Delta'}\,}
\end{align*}
We will divide this into several cases.
In the following, the earlier cases will take precedence, and the lower cases will be those where the upper cases do not hold true.
\begin{proofcases}
\item[$\psi\in\Delta\cup\Delta'\cup\Gamma\cup\Gamma'$]
Since both are similar, we consider the case $\psi\in\Delta\cup\Delta'$.
The conclusion obtained by $(\Cut)$ is the conclusion of $T$ with an extra sentence added.
Therefore, we can obtain it by ignoring $T'$ and applying $\Modify$ to $T$.
\item[if $\psi$ is not the applied part of the rule]
Since we have already considered the above case, we can think of it as $\psi\not\in\Delta\cup\Delta'\cup\Gamma\cup\Gamma'$.
Since both are similar, we consider the case where the root rule of $T'$ is unrelated to $\psi$.
Generally, $T'$ takes the following form:
\[name^{choice}~\proofrule
{\pos{\proofrulet{T'_{i}}{\UL{\Phi'_{i}}\vdash_{\iota_{i}(\Gamma'_{i})\mid\Sigma_{i}\mid\iota_{i}(\Delta'_{i})}\UL{\Psi'_{i}}}}_{i\in I}}
{\OL{\Phi'}\vdash_{\cup_{i\in I}\Gamma'_{i}\mid\Sigma\mid\cup_{i\in I}\Delta'_{i}}\OL{\Psi'}}\]
where $\Gamma'\cup\{\psi\}=\cup_{i\in I}\Gamma'_{i}\cup\Phi'$, $\Delta'=\cup_{i\in I}\Delta'_{i}\cup\Psi'$.
\begin{proofcases}
\item[if $I=\emptyset$]
In other words, if $\psi$ is introduced as an extra sentence by $(\Atom)$ or "left introduction of $\bot$",
then instead of introducing $\psi$, we can add an extra sentence to $T'$ so that the result is the same as what we get with $(\Cut)$.
We ignore $T$.
\item[if $I\neq\emptyset$] There is $i\in I$ such that $\psi\in\Gamma'_{i}$.
We define a new proof diagram as follows:
\[name^{choice}~\proofrule
{\pos{\proofrulet{T''_{i}}{\UL{\Phi'_{i}}\vdash_{\Gamma''_{i}\mid\Sigma_{i}\mid\Delta''_{i}}\UL{\Psi'_{i}}}}_{i\in I}}
{\OL{\Phi'}\vdash_{\Gamma\cup(\cup_{i\in I}\Gamma'_{i}\setminus\{\psi\})\mid\Sigma\mid\Delta\cup(\cup_{i\in I}\Delta'_{i})}\OL{\Psi'}}\]
where
\begin{proofcases}
\item[if $\psi\in\Gamma'_{i}$] $\pos{\Gamma''_{i},\Delta''_{i}}:=\pos{\iota_{i}(\Gamma)\cup\iota_{i}(\Gamma'_{i}\setminus\{\psi\}),\iota_{i}(\Delta)\cup\iota_{i}(\Delta'_{i})}$, and
\\ \ \\
$
T''_{i}
:=
\Cut^{\iota_{i}(\psi)}~\proofrule{\proofrulet{\Modify(\iota_{i},T,\iota_{i}(\Gamma),\iota_{i}(\Delta)\cup\iota_{i}(\psi))}{\,\vdash_{\iota_{i}(\Gamma)\mid\Sigma_{i}\mid\iota_{i}(\Delta)}\UL{\iota_{i}(\psi)}}\Space\proofrulet{T'_{i}}{\UL{\iota_{i}(\psi)}\,\Phi'_{i}\vdash_{\iota_{i}(\Gamma'_{i}\setminus\{\psi\})\mid\Sigma_{i}\mid\iota_{i}(\Delta'_{i})}\Psi'_{i}}}
{\UL{\Phi'_{i}}\vdash_{\iota_{i}(\Gamma)\cup\iota_{i}(\Gamma'_{i}\setminus\{\psi\})\mid\Sigma_{i}\mid\iota_{i}(\Delta)\cup\iota_{i}(\Delta'_{i})}\UL{\Psi'_{i}}}
$
\\
\item[if $\psi\not\in\Gamma'_{i}$] $\pos{\Gamma''_{i},\Delta''_{i}}:=\pos{\iota_{i}(\Gamma'_{i}),\iota_{i}(\Delta'_{i})}$, and $T''_{i}:=T'_{i}$.
\end{proofcases}
\end{proofcases}
\item[If $\psi$ is the applied part in both rules]
Since we have already considered the above case, we can think of it as $\psi\not\in\Delta\cup\Delta'\cup\Gamma\cup\Gamma'$.
\begin{proofcases}
\item[$\psi$ is atomic]
The $\Atom$ rule involves adding extra compound sentences to what hold true in $\vDash^{\INS}$.
Since $\vDash^{\INS}$ satisfies $(\Cut)$, we can apply $(\Cut)$ first and then add the extra sentence.
\item[$\psi=\neg\phi$]
The root parts of $T$ and $T'$ are in the form
$\neg_{R}~\proofrule
{\phi\,\Gamma\vdash_{\Sigma}\Delta}
{\Gamma\vdash_{\Sigma}\Delta\,\neg\phi}$ and
$\neg_{L}~\proofrule
{\Gamma'\vdash_{\Sigma}\Delta'\,\phi}
{\neg\phi\,\Gamma'\vdash_{\Sigma}\Delta'}$, respectively.
Let the proofs of the premises be $T_{0}$ and $T'_{0}$ respectively.
We can define it by $\Cut(\Gamma,\Gamma',\Delta,\Delta',\psi,T,T')
:=\Cut(\Gamma',\Gamma,\Delta',\Delta,\phi,T'_{0},T_{0})$.
\item[$\psi=\vee\phi$]
Since $\bot=\vee\pos{}$ cannot be introduced by ${\vee}_{R}$, we can assume that both are not empty sequences.
The root parts of $T$ and $T'$ are in the form
${\vee}_{R}^{n}~\proofrule
{\Gamma\vdash_{\Sigma}\Delta\,\phi_{n}}
{\Gamma\vdash_{\Sigma}\Delta\,\vee\phi}$ and
${\vee}_{L}~\proofrule
{\pos{\phi_{n}\,\Gamma'\vdash_{\Sigma}\Delta'}_{n\in N}}
{\vee\phi\,\Gamma'\vdash_{\Sigma}\Delta'}$ respectively.
Let the proofs of the premises be $T_{n}$, $T'_{i}\,(i\in N)$ respectively.
We can define it by $\Cut(\Gamma,\Gamma',\Delta,\Delta',\psi,T,T')
:=\Cut(\Gamma,\Gamma',\Delta,\Delta',\phi_{n},T_{n},T'_{n})$.
\item[$\psi=\Exists{X}\phi$]
The root parts of $T$ and $T'$ are in the form
$\exists_{R}^{\theta}~\proofrule
{\Gamma\vdash_{\Sigma}\Delta\,\theta(\phi)}
{\Gamma\vdash_{\Sigma}\Delta\,\Exists{X}\phi}$ and
$\exists_{L}~\proofrule
{\phi\,\Sigma^{\Dex}(X)(\Gamma')\vdash_{\Sigma^{\Dex}[X]}\Sigma^{\Dex}(X)(\Delta')}
{\Exists{X}\phi\,\Gamma'\vdash_{\Sigma}\Delta'}$ respectively.
Translating the premise $\exists_{L}$ in terms of $\theta$ gives
$\theta(\phi)\,\Gamma'\vdash_{\Sigma}\Delta'$.
Let the proofs of the premises be $T_{0}$, $T'_{0}$ respectively.
We can define it by
\[
\Cut(\Gamma,\Gamma',\Delta,\Delta',\psi,T,T')
:=\Cut(\Gamma,\Gamma',\Delta,\Delta',\theta(\phi),T_{0},\Modify(\theta,T'_{0},\theta(\phi)\,\Gamma',\Delta')).
\]
\end{proofcases}
\end{proofcases}
\end{proofcases}
\end{proof}

\begin{proof}[Proof of Lemma~\ref{lemma:syn-soundness}]
We just need to show that the proof rules also hold true for $\vDash^{\FOL^{\relax}_{\alpha}(\INS)}$.
\begin{proofcases}
\item[$\Atom^{{\relax}}$]
Assume $\Gamma_{b}\vDash^{\INS}_{\Sigma}\Delta_{b}$.
From the definition,
$\Gamma_{b}\vDash^{\FOL^{\relax}_{\alpha}(\INS)}_{\Sigma}\Delta_{b}$, and
$\Gamma_{b}\cup\Gamma\vDash^{\FOL^{\relax}_{\alpha}(\INS)}_{\Sigma}\Delta\cup\Delta_{b}$ hold true.
Therefore, the rule holds true.
\item[$\neg_L$]
Assume $\Gamma\vDash^{\FOL^{{\relax}}_{\alpha}(\INS)}_{\Sigma}\Delta\cup\{\phi\}$ and
$\A\models^{\FOL^{\relax}_{\alpha}(\INS)}\gamma$ for all $\gamma\in\{\neg\phi\}\cup\Gamma$.
We show that $\A\models^{\FOL^{\relax}_{\alpha}(\INS)}\gamma$ holds true for some $\gamma\in\Delta$.
By the assumption, $\A$ satisfies at least one of $\Delta\cup\{\phi\}$, but it is not $\phi$, so
at least one of $\Delta$ must be satisfied.
\item[$\neg_R$]
Assume $\{\phi\}\cup\Gamma\vDash^{\FOL^{{\relax}}_{\alpha}(\INS)}_{\Sigma}\Delta$ and
$\A\models^{\FOL^{\relax}_{\alpha}(\INS)}\gamma$ for all $\gamma\in\Gamma$.
We show that $\A\models^{\FOL^{\relax}_{\alpha}(\INS)}\gamma$ holds true for some $\gamma\in\Delta\cup\{\neg\phi\}$.
By the assumption, if $\A$ satisfies $\phi$, then it satisfies at least one of the conditions in $\Delta$.
Including the case where $\A$ does not satisfy $\phi$, $\A$ satisfies at least one of $\Delta\cup\{\neg\phi\}$.
\item[$\vee_L$]
Assume $\{\phi_{i}\}\cup\Gamma\vDash^{\FOL^{\relax}_{\alpha}(\INS)}_{\Sigma}\Delta\text{ for each }i\in\Dom(\phi)$ and
$\A\models^{\FOL^{\relax}_{\alpha}(\INS)}\gamma$ for all $\gamma\in\{\vee\phi\}\cup\Gamma$.
We show that $\A\models^{\FOL^{\relax}_{\alpha}(\INS)}\gamma$ holds true for some $\gamma\in\Delta$.
By the assumption, we get $\A\models^{\FOL^{\relax}_{\alpha}(\INS)}\gamma$ for all $\gamma\in\Gamma$, and
$\A\models^{\FOL^{\relax}_{\alpha}(\INS)}\phi_{i}$ for some $i\in\Dom(\phi)$.
Therefore, from the first assumption, $\A\models^{\FOL^{\relax}_{\alpha}(\INS)}\gamma$ holds true for some $\gamma\in\Delta$.
\item[$\vee_R$]
Assume $\Gamma\vDash^{\FOL^{\relax}_{\alpha}(\INS)}_{\Sigma}\Delta\cup\{\phi_{i}\}\text{ for some }i\in\Dom(\phi)$ and
$\A\models^{\FOL^{\relax}_{\alpha}(\INS)}\gamma$ for all $\gamma\in\Gamma$.
We show that $\A\models^{\FOL^{\relax}_{\alpha}(\INS)}\gamma$ holds true for some $\gamma\in\Delta\cup\{\vee\phi\}$.
By the assumption$\A\models^{\FOL^{\relax}_{\alpha}(\INS)}\gamma$ holds true for some $\gamma\in\Delta\cup\{\phi_{i}\mid i\in\Dom(\phi)\}$.
In the case of $\gamma\in\{\phi_{i}\mid i\in\Dom(\phi)\}$, $\A\models^{\FOL^{\relax}_{\alpha}(\INS)}\vee\phi$ holds true.
Therefore, $\A\models^{\FOL^{\relax}_{\alpha}(\INS)}\gamma$ holds true for some $\gamma\in\Delta\cup\{\vee\phi\}$.
\item[${\exists}_{L}$]
We show that if the conclusion is not true, then the premise must also be false.
Let $\{\Exists{X}\phi\}\cup\Gamma\not\vDash^{\FOL^{\relax}_{\alpha}(\INS)}_{\Sigma}\Delta$.
There is a model $\A$ such that
$\A\models^{\FOL^{\relax}_{\alpha}(\INS)}\gamma$ for all $\gamma\in\{\Exists{X}\phi\}\cup\Gamma$, and
$\A\not\models^{\FOL^{\relax}_{\alpha}(\INS)}\gamma$ for all $\gamma\in\Delta$.
Since $\A\models^{\FOL^{\relax}_{\alpha}(\INS)}\Exists{X}\phi$,
there is an expansion $\hat{\A}$ that satisfies $\hat{\A}\models^{\FOL^{\relax}_{\alpha}(\INS)}\phi$.
By satisfaction condition,
$\hat{\A}\models^{\FOL^{\relax}_{\alpha}(\INS)}\gamma$ holds true for each $\gamma\in\Sen^{\FOL^{{\relax}}_{\alpha}(\INS)}(\Sigma^{\Dex^{\INS}}(X))(\Gamma)$ and
$\hat{\A}\not\models^{\FOL^{\relax}_{\alpha}(\INS)}\gamma$ holds true for each $\gamma\in\Sen^{\FOL^{{\relax}}_{\alpha}(\INS)}(\Sigma^{\Dex^{\INS}}(X))(\Delta)$.
Therefore, $\{\phi\}\cup\Sen^{\FOL^{{\relax}}_{\alpha}(\INS)}(\Sigma^{\Dex^{\INS}}(X))(\Gamma)\not\vDash^{\FOL^{\relax}_{\alpha}(\INS)}_{\Sigma^{\Dex^{\INS}}[X]}\Sen^{\FOL^{{\relax}}_{\alpha}(\INS)}(\Sigma^{\Dex^{\INS}}(X))(\Delta)$ holds true.
\item[${\exists}_{R}$]
Assume $\Gamma\vDash^{\FOL^{\relax}_{\alpha}(\INS)}_{\Sigma}\Delta\cup\{\Sen^{\FOL^{{\relax}}_{\alpha}(\INS)}(\theta)(\phi)\}$ and
$\A\models^{\FOL^{\relax}_{\alpha}(\INS)}\gamma$ for all $\gamma\in\Gamma$.
We show that $\A\models^{\FOL^{\relax}_{\alpha}(\INS)}\gamma$ holds true for some $\gamma\in\Delta\cup\{\Exists{X}\phi\}$.
By the assumption, $\A\models^{\FOL^{\relax}_{\alpha}(\INS)}\gamma$ holds true for some $\gamma\in\Delta\cup\{\Sen^{\FOL^{{\relax}}_{\alpha}(\INS)}(\theta)(\phi)\}$.
In the case of $\gamma=\Sen^{\FOL^{{\relax}}_{\alpha}(\INS)}(\theta)(\phi)$, from satisfaction condition, $\Mod^{\INS}(\theta)(\A)\models^{\FOL^{\relax}_{\alpha}(\INS)}\phi$.
Since $\Mod^{\INS}(\theta)(\A)$ is a $\Sigma^{\Dex^{\INS}}(X)$-expansion of $\A$, so $\A\models^{\FOL^{\relax}_{\alpha}(\INS)}\Exists{X}\phi$ holds true.
Therefore, $\A\models^{\FOL^{\relax}_{\alpha}(\INS)}\gamma$ holds true for some $\gamma\in\Delta\cup\{\Exists{X}\phi\}$.
\end{proofcases}
\end{proof}

\begin{proof}[Proof of Lemma~\ref{lemma:basic-compactness}]\
\begin{proofcases}
\item[\ref{enu:b-compa1}]
Since $\A\models^{\INS}\phi$, there is $f:T_{\Sigma,\Gamma}\to\A$.
Since $h:\A\to\B$, $h\circ f:T_{\Sigma,\Gamma}\to\B$.
Therefore, $\B\models^{\INS}\phi$.
\item[\ref{enu:b-compa2}]
(Forward implication)
Let $T_{\Sigma,\Gamma}\models^{\INS}\phi$.
Suppose $\A\models^{\INS}\gamma$ holds true for each $\gamma\in\Gamma$.
Then there is $f:T_{\Sigma,\Gamma}\to\A$.
By $T_{\Sigma,\Gamma}\models^{\INS}\phi$ and (\ref{enu:b-compa1}), $\A\models^{\INS}\phi$ holds true.
{\par}
(Backward implication)
Conversely, let $\Gamma\vDash^{\INS}_{\Sigma}\phi$.
Since $T_{\Sigma,\Gamma}\models^{\INS}\gamma$ for each $\gamma\in\Gamma$, $T_{\Sigma,\Gamma}\models^{\INS}\phi$.
\end{proofcases}
\end{proof}

\begin{proof}[Proof of Theorem~\ref{theorem:syn-compactness}]
We will use induction on the complexity of proof diagrams.
Let $\pos{\pos{\Gamma_{i},\Sigma_{i},\Delta_{i}}\in\Modify(\INS)}_{i\in\pos{J,\leq}}$
with a co-limit $\{\chi_{i}:\pos{\Gamma_{i},\Sigma_{i},\Delta_{i}}\to\pos{\Gamma,\Sigma,\Delta}\}_{i\in J}$,
and let $T$ is a proof of $\Gamma\vdash_{\Sigma}\Delta$.
We show the existence of $i\in J$ that satisfies $\Gamma_{i}\vdash_{\Sigma_{i}}\Delta_{i}$.
\begin{proofcases}
\item[$\Atom$] By compactness of $\vDash^{\INS}$.
\item[else]
Since the others are roughly the same, consider the case of $({\neg}_{L}$.
The last part will take the following form.
\[\neg_{L}~\proofrule{\,\vdash_{\Gamma\mid\Sigma\mid\Delta}\phi\,}{\,\neg\phi\vdash_{\Gamma\mid\Sigma\mid\Delta}\,}\]
As we saw in the proof of $\Modify$, we can consider the premises of a proof to include the sentences that appear in the conclusion.
Therefore, $\Gamma$ can remain as is.
From the conditions, $\Sen^{\FOL_{\omega}(\INS)}$ is co-continuous.
So there are $i\in J$ and $\phi_{i}\in J$ such that $\Sen^{\FOL_{\omega}(\INS)}(\chi_{i})(\phi_{i})=\phi$.
Let $\pos{\Gamma'_{n},\Sigma'_{n},\Delta'_{n}}_{n\in\pos{J,{\leq}}}$ be obtained by adding $\Sigma_{i\leq j}(\phi_{i})$ to each $\Delta_{j}$.
By $(IH)$, there is $k\geq i$ such that $\Gamma'_{k}\vdash_{\Sigma'_{k}}\Delta'_{k}$.
Applying $({\neg}_{L})$ to this gives $\{\neg\Sigma_{i\leq k}(\phi_{i})\}\cup\Gamma_{k}\vdash_{\Sigma_{k}}\Delta_{k}$.
Since $\neg\phi\in\Gamma$, there is $l\geq k$ such that $\neg\Sigma_{i\leq l}(\phi_{i})\in\Gamma_{l}$.
Therefore, $\Gamma_{l}\vdash_{\Sigma_{l}}\Delta_{l}$ holds true.
\end{proofcases}
\end{proof}

\begin{proof}[Proof of Theorem~\ref{theorem:syn-completeness}]
Since we already checked soundness, now we confirm ${\vdash^{\omega\,\INS}_{\Sigma}}\supseteq{\vDash^{\FOL_{\omega}(\INS)}_{\Sigma}}$.
Let $\Sigma\in|\Sig^{\INS}|$ and let $\Gamma,\Delta\subseteq\Sen^{\FOL_{\omega}(\INS)}(\Sigma)$ such that $\Gamma\not\vdash^{\omega\,\INS}_{\Sigma}\Delta$.
We will show $\Gamma\not\vDash^{\FOL_{\omega}(\INS)}_{\Sigma}\Delta$.
Let $\alpha:=\power_{\omega}(\Sigma)$.
Let $\Dex^{\alpha}_{\Sigma}$ be the smallest set that is closed under the operations of extending by a variable and taking the limit of sequences less than $\alpha$, starting from $\Sigma$.
Since an expansion by a variable does not change the power, $\power_{\omega}(\Sigma')=\alpha$ for all $\Sigma'\in\Dex^{\alpha}_{\Sigma}$.
Prepare a choice function of bijections $\pos{\phi_{\Sigma'}:\alpha\to\Sen^{\FOL_{\omega}(\INS)}(\Sigma')}_{\Sigma'\in\Dex^{\alpha}_{\Sigma}}$.
Prepare a bijective $pair:\alpha\times\alpha\to\alpha$ that satisfies $pair(m,n)\geq m$.
We define $\pos{\BS\Gamma_{n},\BS\Sigma_{n},\BS\Delta_{n}}_{n\in\pos{\alpha,{\leq}}}$ that satisfies $\BS\Gamma_{n}\not\vdash^{\alpha\,\INS}_{\BS\Sigma_{n}}\BS\Delta_{n}\,(\text{for each }n\in\alpha)$ as follows:
\begin{proofcases}
\item[$n=0$] $\pos{\BS\Gamma_{n},\BS\Sigma_{n},\BS\Delta_{n}}:=\pos{\Gamma,\Sigma,\Delta}$.
\item[$n=n'+1$] Let $pair(l,m)=n'$.
\begin{proofcases}
\item[$\phi_{\BS\Sigma_{l}}(m)\,\BS\Gamma_{n'}\not\vdash^{\alpha\,\INS}_{\BS\Sigma_{n'}}\BS\Delta_{n'}$]\
\begin{proofcases}
\item[$\phi_{\Sigma_{l}}(m)=\Exists{X}\psi$]\
\begin{itemize}
\item $\pos{\BS\Gamma_{n},\BS\Sigma_{n},\BS\Delta_{n}}:=\pos{\{\psi\}\cup\BS\Sigma_{n'}^{\Dex^{\INS}}(X)(\{\phi_{\BS\Sigma_{l}}(m)\}\cup\BS\Gamma_{n'}),\BS\Sigma_{n'}^{\Dex^{\INS}}[X],\BS\Sigma^{\Dex^{\INS}}(X)(\BS\Delta_{n'})}$.
$\BS\Gamma_{n}\not\vdash^{\alpha\,\INS}_{\BS\Sigma_{n}}\BS\Delta_{n}$ holds true from $(\exists_{L})$.
\item $\pos{\BS\Gamma_{n'\leq n},\BS\Sigma_{n'\leq n},\BS\Delta_{n'\leq n}}:=\pos{{\subseteq},\BS\Sigma_{n'}^{\Dex^{\INS}}(X),{\subseteq}}$.
\end{itemize}
\item[else]\
\begin{itemize}
\item $\pos{\BS\Gamma_{n},\BS\Sigma_{n},\BS\Delta_{n}}:=\pos{\{\phi_{\BS\Sigma_{l}}(m)\}\cup\BS\Gamma_{n'},\BS\Sigma_{n'},\BS\Delta_{n'}}$.
\item $\pos{\BS\Gamma_{n'\leq n},\BS\Sigma_{n'\leq n},\BS\Delta_{n'\leq n}}:=\pos{{\subseteq},\id_{\Sigma},{\subseteq}}$.
\end{itemize}
\end{proofcases}
\item[$\phi_{\BS\Sigma_{l}}(m)\,\BS\Gamma_{n'}\vdash^{\alpha\,\INS}_{\BS\Sigma_{n'}}\BS\Delta_{n'}$]
From $(\Cut)$ and
$\BS\Gamma_{n'}\not\vdash^{\alpha\,\INS}_{\BS\Sigma_{n'}}\BS\Delta_{n'}$,
$\BS\Gamma_{n'}\not\vdash^{\alpha\,\INS}_{\BS\Sigma_{n'}}\BS\Delta_{n'}\,\phi_{\BS\Sigma_{l}}(m)$.
\begin{itemize}
\item $\pos{\BS\Gamma_{n},\BS\Sigma_{n},\BS\Delta_{n}}:=\pos{\BS\Gamma_{n'},\BS\Sigma_{n'},\BS\Delta_{n'}\cup\{\phi_{\Sigma_{l}}(m)\}}$
\item $\pos{\BS\Gamma_{n'\leq n},\BS\Sigma_{n'\leq n},\BS\Delta_{n'\leq n}}:=\pos{{\subseteq},\id_{\Sigma},{\subseteq}}$.
\end{itemize}
\end{proofcases}
\item[$n$ is limit]\
\begin{itemize}
\item $\pos{\BS\Gamma_{n},\BS\Sigma_{n},\BS\Delta_{n}}:=colimit(\pos{\BS\Gamma_{i},\BS\Sigma_{i},\BS\Delta_{i}}_{i\in\pos{n,{\leq}}})$.
By the compactness of $\vdash^{\alpha\,\INS}$, $\BS\Gamma_{n}\not\vdash^{\alpha\,\INS}_{\BS\Sigma_{n}}\BS\Delta_{n}$.
\item $\pos{\BS\Gamma_{n'\leq n},\BS\Sigma_{n'\leq n},\BS\Delta_{n'\leq n}}$ is decided from the co-limit.
\end{itemize}
\end{proofcases}
Let $\{\pos{\BS\Gamma_{n\leq\alpha},\BS\Sigma_{n\leq\alpha},\BS\Delta_{n\leq\alpha}}:\pos{\BS\Gamma_{n},\BS\Sigma_{n},\BS\Delta_{n}}\to\pos{\BS\Gamma_{\alpha},\BS\Sigma_{\alpha},\BS\Delta_{\alpha}}\}_{n\in\alpha}$ be the co-limit of $\pos{\BS\Gamma_{n},\BS\Sigma_{n},\BS\Delta_{n}}_{n\in\pos{\alpha,{\leq}}}$.
By co-continuity of $\Sen^{\INS}$ and $(\cdot)_{\Dex^{\INS}}$, $\Sen^{\FOL_{\omega}(\INS)}$ is also a co-continuous.
So from the construction, $\pos{\BS\Gamma_{\alpha},\BS\Sigma_{\alpha},\BS\Delta_{\alpha}}$ is maximally consistent.
\\
Let $\Gamma_{b}:=\BS\Gamma_{\alpha}\cap\Sen^{\INS}(\BS\Sigma_{\alpha})$ and let $\A':=T_{\BS\Sigma_{\alpha},\Gamma_{b}}$ be a reachable epi-basic model.
We show that $\A'\models^{\FOL_{\omega}(\INS)}\gamma$ iff $\gamma\in\BS\Gamma_{\alpha}$ for each $\gamma\in\Sen^{\FOL_{\omega}(\INS)}(\BS\Sigma_{\alpha})$ by induction on sentence complexity.
\begin{proofcases}
\item[$\gamma$ is atomic] By the definition of $\A'$, and from the Lemma~\ref{lemma:basic-compactness}.
\item[$\gamma=\neg\gamma'$] By ($IH$), $\A'\models^{\FOL_{\omega}(\INS)}\gamma'$ iff $\gamma'\in\BS\Gamma_{\alpha}$.
From the maximally consistency of $\pos{\BS\Gamma_{\alpha},\BS\Sigma_{\alpha},\BS\Delta_{\alpha}}$ and the proof rules of $\vdash^{\alpha\,\INS}$, $\neg\gamma'\in\BS\Gamma_{\alpha}$ iff $\gamma'\not\in\BS\Gamma_{\alpha}$.
So from the definition of $\models^{\FOL_{\omega}(\INS)}$,
$\A'\models^{\FOL_{\omega}(\INS)}\neg\gamma'$ iff $\neg\gamma'\in\BS\Gamma_{\alpha}$.
\item[$\gamma=\vee\gamma'$] By ($IH$), $\A'\models^{\FOL_{\omega}(\INS)}\gamma'_{i}$ iff $\gamma'_{i}\in\BS\Gamma_{\alpha}$ for each $i\in\Dom(\gamma')$.
From the maximally consistency of $\pos{\BS\Gamma_{\alpha},\BS\Sigma_{\alpha},\BS\Delta_{\alpha}}$ and the proof rules of $\vdash^{\alpha\,\INS}$, $\vee\gamma'\in\BS\Gamma_{\alpha}$ iff $\gamma'_{i}\not\in\BS\Gamma_{\alpha}$ for some $i\in\Dom(\gamma')$.
So from the definition of $\models^{\FOL_{\omega}(\INS)}$,
$\A'\models^{\FOL_{\omega}(\INS)}\vee\gamma'$ iff $\vee\gamma'\in\BS\Gamma_{\alpha}$.
\item[$\gamma=\Exists{X}\gamma'$]\
\begin{proofcases}
\item[${\Rightarrow}$] Let $\A'\models^{\FOL_{\omega}(\INS)}\Exists{X}\gamma'$.
By the semantics, $\widehat{\A'}\models^{\FOL_{\omega}(\INS)}\gamma'$ for some $\BS\Sigma_{\alpha}^{\Dex^{\INS}}(X)$-expansion $\widehat{\A'}$ of $\A'$.
Since $\A'$ is reachable, there is a $\BS\Sigma_{\alpha}^{\Dex^{\INS}}(X)$-substitution $\theta$ that satisfies $\widehat{\A'}=\Mod^{\INS}(\theta)(\A')$.
By satisfaction condition, $\A'\models^{\FOL_{\omega}(\INS)}\Sen^{\FOL_{\omega}(\INS)}(\theta)(\gamma')$.
By ($IH$), $\Sen^{\FOL_{\omega}(\INS)}(\theta)(\gamma')\in\BS\Gamma_{\alpha}$.
From the maximally consistency of $\pos{\BS\Gamma_{\alpha},\BS\Sigma_{\alpha},\BS\Delta_{\alpha}}$ and the proof rules of $\vdash^{\alpha\,\INS}$, $\Exists{X}\gamma'\in\BS\Sigma_{\alpha}$.
\item[${\Leftarrow}$] Let $\Exists{X}\gamma'\in\BS\Gamma_{\alpha}$.
Let
$i\in\alpha$,
$X_{i}\in|(\BS\Sigma_{i})_{\Dex^{\INS}}|$,
$\gamma'_{i}\in\Sen^{\FOL_{\omega}(\INS)}(\BS\Sigma_{i}^{\Dex^{\INS}}[X_{i}])$
that satisfy
$(\BS\Sigma_{i\leq\alpha})_{\Dex^{\INS}}(X_{i})=X$,
$\Sen^{\FOL_{\omega}(\INS)}(\BS\Sigma_{i\leq\alpha}^{\Dex^{\INS}}[X_{i}])(\gamma'_{i})=\gamma'$.
Let $j\in\alpha$ such that $\phi_{\BS\Sigma_{i}}(j)=\Exists{X_{i}}\gamma'_{i}$.
Let $n=n'+1\,(n'=pair(i,j))$.
From the definition,
$\pos{\BS\Gamma_{n},\BS\Sigma_{n},\BS\Delta_{n}}:=\pos{\{\gamma'_{i}\}\cup\{\BS\Sigma_{n'}^{\Dex^{\INS}}(X_{i})(\phi_{\BS\Sigma_{i}}(j))\}\cup\BS\Sigma_{n'}^{\Dex^{\INS}}(X_{i})(\BS\Gamma_{n'}),\BS\Sigma_{n'}^{\Dex^{\INS}}[X_{i}],\BS\Sigma^{\Dex^{\INS}}(X_{i})(\BS\Delta_{n'})}$.
$\BS\Sigma_{n'}^{\Dex^{\INS}}(X_{i})(\phi_{\BS\Sigma_{i}}(j))=\Exists{X'_{i}}(\chi'(\gamma'_{i}))
\,(X'_{i}:=(\BS\Sigma_{n'}^{\Dex^{\INS}}(X_{i}))_{\Dex^{\INS}}(X_{i}),\,\chi':=(\BS\Sigma_{n'}^{\Dex^{\INS}}(X_{i}))^{\Dex^{\INS}}[X_{i}])$.
Let $\theta_{i}$ be a substitution which was covered in Example~\ref{exa:substitution}.
That satisfies $\theta_{i}(\chi'(\gamma'_{i}))=\gamma'_{i}$.
Let $\theta$ be the translation of $\theta_{i}$ to $\BS\Sigma_{\alpha}^{\Dex^{\INS}}(X)$.
From the definition, $\theta(\gamma')\in\Gamma_{\alpha}$.
By ($IH$), $\A'\models^{\FOL_{\omega}(\INS)}\theta(\gamma')$.
By semantics, $\A'\models^{\FOL_{\omega}(\INS)}\Exists{X}\gamma'$.
\end{proofcases}
\end{proofcases}
Let $\A:=\Mod^{\INS}(\BS\Sigma_{0\leq\alpha})(\A')\in|\Mod^{\INS}(\Sigma)|$.
By satisfaction condition, $\A\models^{\FOL_{\omega}(\INS)}\gamma$ for each $\gamma\in\Gamma$, and $\A\not\models^{\FOL_{\omega}(\INS)}\delta$ for each $\delta\in\Delta$.
Therefore, $\Gamma\not\vDash^{\FOL_{\omega}(\INS)}_{\Sigma}\Delta$.
\end{proof}

\end{document}